\documentclass{article}
\usepackage[utf8]{inputenc}

\RequirePackage{etex} 
\usepackage{amsmath,amssymb,amsthm,amsfonts,latexsym,bbm,xspace,graphicx,float,mathtools,mathdots}
\usepackage{braket,caption,subcaption,ellipsis,xcolor,textcomp}
\usepackage{combelow} 
\usepackage[backref,citecolor=blue, hidelinks]{hyperref}
\usepackage[nameinlink]{cleveref}
\crefname{ineq}{inequality}{inequalities}
\creflabelformat{ineq}{#2{\upshape(#1)}#3}

\usepackage[letterpaper,margin=1in]{geometry}
\usepackage{enumitem} 

\newtheorem{theorem}{Theorem}

\newtheorem*{namedtheorem}{\theoremname}
\newcommand{\theoremname}{testing}

\newtheorem{lemma}[theorem]{Lemma}

\newtheorem{fact}[theorem]{Fact}
\newtheorem{corollary}[theorem]{Corollary}

\newtheorem{assumption}{Assumption}

\theoremstyle{definition}
\newtheorem{definition}[theorem]{Definition}
\newtheorem{remark}[theorem]{Remark}

\newtheorem{example}[theorem]{Example}

\newcommand{\la}{\left \langle}
\newcommand{\ra}{\right \rangle}

\newcommand{\calK}{\mathcal{K}}
\newcommand{\calL}{\mathcal{L}}

\newcommand{\calX}{\mathcal{X}}
\newcommand{\calY}{\mathcal{Y}}

\newcommand{\bone}{\boldsymbol{1}}

\newcommand{\ignore}[1]{}

\newcommand{\defined}{\coloneqq}
\usepackage{autonum}

\usepackage{natbib}

\newcommand{\lp}{\left(}
\newcommand{\rp}{\right)}
\newcommand{\lb}{\left[}
\newcommand{\rb}{\right]}

\newcommand{\mmdhat}{\mathrm{MMD}_n}

\newcommand{\var}{\mathbb{V}\xspace}

\newcommand{\muhat}{\widehat{\mu}}
\newcommand{\nuhat}{\widehat{\nu}}
\newcommand{\omegahat}{\widehat{\omega}}

\newcommand{\tildek}{\widetilde{k}}
\newcommand{\tildel}{\widetilde{\ell}}
\newcommand{\tildea}{\widetilde{a}}
\newcommand{\tildeb}{\widetilde{b}}
\newcommand{\tildeA}{\widetilde{A}}

\newcommand{\tildeg}{\widetilde{g}}
\newcommand{\tildeh}{\widetilde{h}}
\newcommand{\tildef}{\widetilde{f}}

\newcommand{\convprob}{\stackrel{p}{\longrightarrow}}
\newcommand{\convdist}{\stackrel{d}{\longrightarrow}}

\newcommand{\reals}{\mathbb{R}}
\newcommand{\iid}{\text{i.i.d.}\xspace}

\newcommand{\Sobolev}{\mathcal{W}^{\beta,2}}

\newcommand{\ind}{\mathbbm{1}}
\newcommand{\mc}[1]{\mathcal{#1}}

\newcommand{\hsichat}{\mathrm{HSIC}_n\xspace}
\newcommand{\hsic}{\mathrm{HSIC}\xspace}
\newcommand{\mmd}{\text{MMD}\xspace}

\newcommand{\crosshsic}{\mathrm{x}\hsichat}
\newcommand{\crossmmd}{\mathrm{x}\mmdhat^2}
\newcommand{\cshsic}{\overline{\mathrm{x}}\hsichat}

\newcommand{\data}{\mathcal{D}_{1}^{2n}}
\newcommand{\dataone}{\mathcal{D}_{1}^{n}}
\newcommand{\datatwo}{\mathcal{D}_{n+1}^{2n}}

\newcommand{\term}{\texttt{term}}

\newcommand{\dcov}{\mc{V}^2}
\newcommand{\distX}{\rho_{\mc{X}}}
\newcommand{\distY}{\rho_{\mc{Y}}}
\newcommand{\dkernelX}{k_{\mc{X}}}
\newcommand{\dkernelY}{\ell_{\mc{Y}}}
\newcommand{\crossdcov}{\mc{V}^2_n}
\newcommand{\csdcov}{\overline{\mc{V}}^2_n}
\newcommand{\testdist}{\Psi^{\rho}}

\newcommand{\nullclass}{\mc{P}_n^{(0)}}
\newcommand{\altclass}{\mc{P}_n^{(1)}}
\newcommand{\cov}{\mathrm{Cov}}

\newcommand{\dof}{\texttt{dof}\xspace}
\newcommand{\ww}{\boldsymbol{w}}
\newcommand{\wtilde}{\widetilde{\boldsymbol{w}}}

\usepackage{pgfplots}
\pgfplotsset{compat=newest}
\pgfplotsset{scaled y ticks=false}
\usepgfplotslibrary{groupplots}
\usepgfplotslibrary{dateplot}
\tikzstyle{every node}=[font=\small]
\pgfplotsset{
    yticklabel style={/pgf/number format/fixed},  
}

\pgfplotsset{compat=1.11,
 /pgfplots/ybar legend/.style={
 /pgfplots/legend image code/.code={
 \draw[##1,/tikz/.cd,yshift=-0.25em]
 (0cm,0cm) rectangle (3pt,0.8em);},
 },
}

\title{A Permutation-Free Kernel Independence Test}

\author{%
  Shubhanshu Shekhar$^1$, Ilmun Kim$^{34}$ and Aaditya Ramdas$^{12}$ \\
  \texttt{shubhan2@andrew.cmu.edu, ilmun@yonsei.ac.kr, aramdas@stat.cmu.edu} \vspace{0.1in}\\
   $^1$Department of Statistics and Data Science, Carnegie Mellon University\\
   $^2$Machine Learning Department, Carnegie Mellon University\\
   $^3$Department of Statistics and Data Science, Yonsei University \\
   $^4$Department of Applied Statistics, Yonsei University \\
}
\date{}
\begin{document}

\maketitle

\begin{abstract}
    In nonparametric independence testing, we observe i.i.d.\ data $\{(X_i,Y_i)\}_{i=1}^n$, where $X \in \mc{X}, Y \in \mc{Y}$ lie in any general spaces, and we wish to test the null that $X$ is independent of $Y$.
    Modern test statistics such as  the kernel Hilbert-Schmidt Independence Criterion (HSIC)  and Distance Covariance (dCov) have intractable null distributions due to the degeneracy of the underlying U-statistics. Thus, in practice, one often resorts to using permutation testing, which provides a nonasymptotic guarantee at the expense of recalculating the quadratic-time statistics (say) a few hundred times. This paper provides a simple but nontrivial modification of HSIC and dCov (called  xHSIC and xdCov, pronounced ``cross'' HSIC/dCov) so that they have a limiting Gaussian distribution under the null, and thus do not require permutations. This requires building on the newly developed theory of cross U-statistics by~\cite{kim2020dimension}, and in particular developing several nontrivial extensions of the theory in~\cite{shekhar2022permutation}, which developed an analogous permutation-free kernel two-sample test. We show that our new tests, like the originals, are consistent against fixed alternatives, and minimax rate optimal against smooth local alternatives. Numerical simulations demonstrate that compared to the full dCov or HSIC, our variants have the same power up to a $\sqrt 2$ factor, giving practitioners a new  option for large problems or data-analysis pipelines where computation, not sample size, could be the bottleneck.
\end{abstract}

\setcounter{tocdepth}{2}
\tableofcontents

\section{Introduction}
    We consider the following problem: given observations $\data = \{(X_i, X_i): 1 \leq i \leq 2n\}$  drawn \iid from a distribution $P_{XY}$ on the observation space $\mc{X} \times \mc{Y}$, we wish to test whether $X$ and $Y$ are independent or not. Formally, this is stated as the following hypothesis testing problem: 
    \begin{align}
      H_0:P_{XY} = P_X \times P_Y, \quad \text{versus} \quad H_1: P_{XY} \neq P_X\times P_Y, 
    \end{align}
    where $P_X$ and $P_Y$ denote the marginals of the joint distribution $P_{XY}$. For general observation spaces, a popular approach for independence testing is based on the Hilbert Schmidt Independence Criterion~(HSIC), first introduced by~\citet{gretton2005measuring}. As we describe formally in~\Cref{def:population-hsic} in~\Cref{sec:preliminaries}, the (population) HSIC of a joint distribution $P_{XY}$ on $\mc{X} \times \mc{Y}$ is the sum of squared singular values~(i.e., the Hilbert-Schmidt norm) of a cross-covariance operator defined using $P_{XY}$ and positive definite kernels $k:\mc{X} \times \mc{X} \to \reals$ and $\ell:\mc{Y} \times\mc{Y} \to \reals$. Given the data $\data$, an unbiased empirical estimate of the population HSIC can be computed in quadratic time. 
    Introducing the notation $k_{ij} \equiv k(X_i, X_j)$ and $\ell_{lm} = \ell(Y_l, Y_m)$ for $1 \leq i,j,l,m \leq n$, and $(n)_i \defined \frac{n!}{(n-i)!}$, the empirical estimator of $\hsic$ can be defined as: 
   \begin{align}
       \label{eq:empirical-hsic} 
       \hsichat = \frac{1}{(2n)_2} \sum_{1 \leq i \neq j \leq 2n} k_{ij}\ell_{ij} + \frac{1}{(2n)_4}\sum_{\substack{1\leq i,j,l,m \leq 2n \\ \text{$i,j,l,m$ distinct}}} k_{ij} \ell_{lm} - \frac{2}{(2n)_3} \sum_{\substack{1\leq i,j,l \leq 2n \\ \text{$i,j,l$ distinct}}} \ell_{il}, 
   \end{align}
    For characteristic kernels, \citet{gretton2005measuring} showed that the (population) HSIC is equal to zero only if $P_{XY} = P_X \times P_Y$, and thus it can serve as a measure of independence between $P_X$ and $P_Y$.
    This  can be then used to define an independence test that rejects the null when the statistic $\hsichat$  is larger than an appropriately chosen threshold. The choice of the threshold is crucial in ensuring that the test achieves large power under the alternative, while ensuring that the type-I error is controlled~(at least asymptotically) at a specified level $\alpha \in (0,1)$. 
    
    The empirical HSIC criterion introduced above is an instance of a degenerate U-statistic~\citep{lee2019u}, and thus it has a complicated limiting null distribution --- it is an infinite weighted linear combination of independent chi-squared random variables, the exact expression of which was derived by~\citet[Theorem 2]{gretton2007kernel}. Due to the intractable nature of this null distribution, we cannot directly use this to calibrate the independence test based on the empirical HSIC statistic. Instead, in practice, the rejection threshold is often selected as the $(1-\alpha)$-quantile after recomputing the statistic $\hsichat$ $B$ times after permuting the indices of $(X_i)_{i=1}^{2n}$. Hence, this approach requires us to compute the quadratic time statistic $\hsichat$ a total of $B+1$ times, which might make this method infeasible for larger $n$ and $B$ values. 

    To address the high computational cost of the permutation test, several alternatives such as tests based on deviation bounds between empirical and population HSIC, or using parametric approximations of the null distribution have been proposed. We discuss them in details in~\Cref{subsec:related-work}. However, these existing approaches are either too conservative in practice or do not have theoretical guarantees on their performance. 

    In this paper, we propose a new and simple test that addresses the issues with existing kernel-based independence tests. In particular, we define a new unbiased empirical estimate of HSIC that has a tractable, standard normal, limiting null distribution under mild assumptions. We then use this statistic to define an independence test, that we call the cross-HSIC test, and show that it is consistent against arbitrary fixed alternatives and also achieves minimax rate-optimal power against smooth local alternatives. 

    \subsection{Overview of Results}
    \label{subsec:overview}
        We propose a new statistic, based on the ideas of sample splitting and studentization, and use it to define a new test of independence. 
        Given the observations $\data = \{(X_i, Y_i): 1 \leq i \leq 2n\}$, drawn \iid from a joint distribution $P_{XY}$ on $\mc{X} \times \mc{Y}$,  we first split it into two equal parts, $\dataone = \{(X_i, Y_i): 1 \leq i \leq n\}$ and $\datatwo = \{(X_i, Y_i): n+1 \leq i \leq 2n\}$. Using these two splits, we construct two independent empirical estimates of the cross-covariance operator~(denoted by $f_1$ and $f_2$), and define the statistic $\crosshsic$ as their Hilbert-Schmidt norm: 
        \begin{align}
            &\crosshsic = \langle f_1, f_2 \rangle_{HS},  \\
            &f_1 = \frac{1}{n(n-1)}  \sum_{i=1}^n \sum_{\substack{j=1 \\ j \neq i}}^n h_{ij}, \quad \text{and}  \quad f_2 = \frac{1}{n(n-1)}  \sum_{i=n+1}^{2n} \sum_{\substack{j=n+1 \\ j \neq i}}^{2n} h_{ij},  \\
            &h_{ij} \equiv h(Z_i, Z_j) = \frac{1}{2} \bigl\{ k(X_i, \cdot) - k(X_j, \cdot) \bigr\}\otimes \bigl\{ \ell(Y_i, \cdot) - \ell(Y_j, \cdot) \bigr\}. 
        \end{align}
        The final statistic, denoted by $\cshsic$, is obtained by normalizing $\crosshsic$ with an empirical standard deviation term, stated in~\eqref{eq:sn2-def}.  Using this statistic, we define the cross-HSIC test, $\Psi = \boldsymbol{1}_{ \cshsic > z_{1-\alpha} }$, which rejects the null when $\cshsic$ exceeds the $(1-\alpha)$-quantile of the standard normal distribution. 

        Our first set of results characterize the limiting null distribution of the $\cshsic$ statistic, and justify the rejection threshold used in defining the cross-HSIC test.  To motivate the more general results, we first consider the simple, but instructive, case of univariate observations~(that is, $\mc{X}=\mc{Y} = \reals$) and linear kernels $k(x, x') = xx'$ and $\ell(y, y') = yy'$. In this setting, we first show in~\Cref{thm:warmup-linear-1} that for a fixed null distribution $P_{XY} = P_X \times P_Y$, the existence of finite second moment is sufficient for $\cshsic$ to converge in distribution to $N(0,1)$. Then, in~\Cref{thm:warmup-linear-2}, we show that under stronger moment assumptions, the $\cshsic$ statistic converges to $N(0,1)$ uniformly over a composite class of null distributions. We then move to the case of general kernels and observation spaces, and derive analogous, but slightly more abstract, requirements for the  the asymptotic normality of the $\cshsic$ statistic. In particular, we identify sufficient conditions for the asymptotic normality of $\cshsic$   for the case of fixed $k, \ell$ and $P_{XY}$ in~\Cref{theorem:null-dist-0}, and for the more general case where $k_n, \ell_n$ and $P_{XY,n}$ are all allowed to change with the sample-size $n$ in~\Cref{theorem:null-dist-1}. Overall, our results imply that our $\cshsic$ statistic converges in distribution to a standard normal distribution, in most practically relevant scenarios, and thus the cross-HSIC test controls the type-I error asymptotically at the desired level $\alpha$. 

        Having established the type-I error control achieved by our cross-HSIC test, we next analyze its power. Again, we first consider the case of univariate observations with linear kernels to gain some intuition. In \Cref{thm:warmup-linear-1}, we show  that the cross-HSIC test is consistent against any fixed alternative when $P_X$ and $P_Y$ are (linearly) correlated, while in~\Cref{thm:warmup-linear-2}, we consider the case of local alternatives and show that the cross-HSIC test can consistently detect alternatives separated by a $\Omega(1/\sqrt{n})$ boundary. Similar results also hold for the case of more general kernels. In particular, we first show in~\Cref{prop:fixed-alternative}, that the cross-HSIC test with characteristic kernels $k$ and $\ell$ is consistent against any fixed alternative. Then in~\Cref{theorem:local-alternative}, we show that the cross-HSIC test  instantiated with Gaussian kernels achieves minimax rate-optimal power against smooth local alternatives. 

        Finally, in~\Cref{sec:dcov}, we consider a related class of statistics, called the distance-covariance~(dCov). Using the equivalence between distance-based and kernel-based statistics~\citep{sejdinovic2013equivalence},  we  introduce a new distance-based statistic, called the cross-dCov statistic. We then identify sufficient conditions for this statistic to have a standard normal limiting null distribution, and for it to be consistent against fixed alternatives.  
        
    \subsection{Related Work}
    \label{subsec:related-work}
        Nonparametric independence testing is a classical topic in statistics that still remains a subject of significant contemporary interest. A huge variety of methods have been developed recently for this problem, such as  those based on ranks~\citep{heller2013consistent, weihs2018symmetric, deb2021multivariate, shi2022distribution}, projections~\citep{zhu2017projection}, copula~\citep{dette2013copula} and mutual information~\citep{berrett2019nonparametric}. This paper focuses on a particular class of kernel (and distance) based nonparametric tests, which are very popular in machine learning and statistics. In the rest of this section, we present the details of the most closely related works to ours. 
    
        \textbf{Nonasymptotic tail inequalities for HSIC.}
            \citet{gretton2005measuring} introduced HSIC as a measure of dependence between two random variables, and obtained a high probability (nonasymptotic) deviation inequality between the empirical and population HSIC terms. The resulting test has finite sample validity, and is uniformly consistent against alternatives separated by a $\Omega (1/\sqrt{n})$ boundary (in terms of HSIC). However, this test is very conservative in practice. To address this, \citet{gretton2007kernel} suggested to calibrate the test based on the null distribution of HSIC, and due to the intractability of the  null distribution, they used parametric approximations of the null distribution to select the rejection threshold. This method, however, is a heuristic and doesn't have validity guarantees.

        \textbf{Modifications of the HSIC test statistic.} 
            Another class of permutation-free kernel independence tests take an approach similar to our paper, and modify the empirical HSIC statistic to make its null distribution tractable.  \citet{zhang2018large} developed three such tests, using block-averaged HSIC, random Fourier features~(RFFs) and Nystrom approximation. The block-averaged HSIC statistic is obtained by partitioning the data $\data$ into blocks of size $b$, and then computing the HSIC on these small blocks of data and taking their average; this takes $O(bn)$ time in total. Thus the computational cost varies from linear in $n$ for constant block sizes, to quadratic when $b = \Omega(n)$. Furthermore, when $b = o(n)$ the block-averaged HSIC statistic has a limiting Gaussian distribution under the null, which can be used for calibrating the test. The other two methods~(RFF and Nystrom) have a computational complexity of $\mc{O}(n b^2)$, where $b$ denotes the number of features used. For the RFF method, the null distribution is a finite linear combination of chi-squared random variables, while in the Nystrom method, the null distribution is based on a low-rank approximation of the kernel matrices. 

            \citet{jitkrittum2017adaptive} proposed a new measure of independence, called the finite set independence criterion~(FSIC), that is computed as the average of the squared differences between the mean-embedding values of the joint distribution $P_{XY}$, and the product of marginals $P_X \times P_Y$, at $J$ different locations in $\mc{X} \times \mc{Y}$. For the $J$ locations drawn randomly from a continuous distribution, they showed that FSIC is equal to $0$ if and only if $X$ and $Y$ are independent.  To obtain a more powerful test in practice, they also suggested to select the $J$ locations by optimizing a lower bound on the power of the test. While the resulting linear-time test performed comparably to the quadratic-time permutation test in practice, there are several important factors that distinguish our work from theirs. In particular, they analyze their tests in the regime where the kernels ($k$ and $\ell$), and the number of features $J$ are fixed; while the sample size $n$ goes to infinity. When these parameters are also allowed to change with $n$, the null distribution of their statistics may change. Furthermore, the consistency of their test relies on an accurate estimation of the $J \times J$ covariance matrix, which is not possible in the regime when $J$ also increases with $n$, and $\lim_{n \to \infty} J/n \approx 1$.  Our test does not suffer from these issues, and we prove its consistency and type-I error control in more general settings. 

            For the specific case of Gaussian kernels, \citet{li2019optimality} analyzed an independence test based on a studentized version of the empirical HSIC statistic. In particular, for Gaussian kernels with scale parameter increasing at an appropriate rate with $n$, they showed that the studentized empirical HSIC statistic has a standard normal limiting null distribution, and the resulting independence test has minimax rate-optimal power against smooth local alternatives. In contrast, our studentized cross-HSIC statistic has a limiting null distribution for a much larger class of kernels, including unbounded kernels as well as kernels induced by semi-metrics~\citep{sejdinovic2013equivalence}. Furthermore, when specialized to Gaussian kernels, our test also achieves minimax rate-optimal power, matching the performance of the test studied by~\citet{li2019optimality}.

         \textbf{Other kernel-based independence tests.}
            \citet{deb2020measuring} proposed a class of kernel-based nonparametric measures   of association~(called KMAc) between two random variables $X$ and $Y$ that satisfy the three desired criteria listed by \citet{chatterjee2021new}: they are equal to $0$ for independent $X$ and $Y$,  and are equal to $1$ when  $Y$ is a  measurable function of $X$; they admit simple empirical estimates, and finally, they have simple asymptotic theory when $X$ and $Y$ are independent. These measures significantly generalize the univariate measure of association introduced by \citet{chatterjee2021new}. However, \citet{deb2020measuring} study the asymptotics of the empirical KMAc statistics  only for the case of fixed kernels, and furthermore, they do not analyze the power of their independence tests against local alternatives. 
        
        \textbf{Distance-Covariance tests.}
            \citet{szekely2007measuring} proposed a measure of dependence between two random variables, called the distance-Covariance~(dCov) that is defined as the weighted $L^2$ norm between the characteristic function of the joint distribution, and the product of the marginal characteristic functions. For the case of euclidean spaces, \citet{szekely2007measuring} showed that the $(1-\alpha)$-quantile of a normalized version of the empirical dCov statistic can be upper bounded by that of a quadratic form of a Gaussian; thus providing a permutation-free test of independence. In practice, however, this test is quite conservative as the upper bound on the $(1-\alpha)$-quantile is tight only in the case of Bernoulli distributions. \citet{lyons2013distance} extended the definition and analysis of the dCov measure beyond euclidean spaces, to arbitrary metric spaces. This was further generalized to the case of semi-metric spaces by \citet{sejdinovic2013equivalence}, who also derived a precise equivalence between dCov and kernel-based HSIC measures~(we recall this in~\Cref{fact:dcov-hsic-equivalence}). In this paper, we use this equivalence to obtain new dCov based statistics, that also have a standard normal limiting distribution under the null. 

        \textbf{Cross U-statistics.}
            We note that the general design strategy used in this paper, based on sample-splitting and studentization, was first introduced by~\citet{kim2020dimension} for the case of one-sample U-statistics. The primary motivation of their work was to develop  inference techniques that are \emph{dimension-agnostic} --- i.e., methods based on statistics whose asymptotic distribution remains the same, irrespective of the dimension regime.  \citet{shekhar2022permutation}  applied the ideas of~\citet{kim2020dimension} to the case of two-sample testing, to develop a permutation-free two-sample test based on kernel-MMD metric. In this paper, we adapt the ideas from these two papers, to develop a permutation-free kernel independence test. As we describe in~\Cref{subsec:connections-to-mmd}, the analysis techniques developed by~\citet{shekhar2022permutation} for their cross-MMD test cannot be directly applied to our setting, and hence we develop new methods for establishing the properties of our cross-HSIC test. 

\section{Hilbert-Schmidt Independence Criterion~(HSIC)}
\label{sec:preliminaries}
    As mentioned earlier, we assume that the observations $(X, Y)$ lie in the space $\calX \times \calY$, where $\calX$ may be different from $\calY$. Let $\calK$ and $\calL$ denote reproducing kernel Hilbert Spaces~(RKHS) associated with positive-definite kernels $k:\calX \times \calX \to \reals$ and $\ell:\calY \times \calY \to \reals$ respectively. With $\phi(\cdot)$ and $ \psi(\cdot)$ denoting the associated feature maps $x \mapsto k(x, \cdot)$ and $y \mapsto \ell(y, \cdot)$, we formally introduce the Hilbert-Schmidt independence criterion~(HSIC).

    \begin{definition}
    \label{def:population-hsic}
        Let $P_{XY}$ denote a probability distribution on $\calX \times \calY$, and $\mc{H}$ and $\mc{G}$ denote RKHS associated with kernels $k:\calX \times \calX \to \mathbb{R}$ and $\ell:\calY \times \calY \to \mathbb{R}$ respectively. Then, the Hilbert--Schmidt Independence Criterion~(HSIC) is defined as the Hilbert--Schmidt norm of the associated cross-covariance operator:
        \begin{align}
            &\hsic(P_{XY}, k, \ell) \defined \|C_{XY}\|_{HS}^2, \quad \text{where}   \quad C_{XY} \defined \mathbb{E}_{XY}\lb (\phi(X) - \mu) \otimes (\psi(Y)-\nu) \rb.  \label{eq:def-population-hsic}
        \end{align}
        Here $\otimes$ denotes the tensor product, and the terms $\mu$ and $\nu$ denote $\mathbb{E}_X[\phi(X)]$ and $\mathbb{E}_Y[\psi(Y)]$ respectively. 
    \end{definition}
   
   As shown by~\citet[Lemma 1]{gretton2005measuring}, $\hsic(P_{XY}, k, \ell)$ can be expressed as follows, in terms of the kernel functions, with $(X,Y)$ and $(X', Y')$ denoting two independent draws from the  distribution $P_{XY}$.
   \begin{align}
        \label{eq:population-hsic-1}
       \hsic(P_{XY}, k, \ell) = \mathbb{E}_{XX'YY'}[k(X, X') \ell(Y, Y')] + \mathbb{E}_{XX'}[k(X,X')]\mathbb{E}_{YY'}[\ell(Y,Y')] - 2 \mathbb{E}_{XY}[\mu(X) \nu(Y)]. 
   \end{align}
   
    Given data $\data = \{(X_i, Y_i): 1 \leq i \leq 2n\}$  consisting of $2n$ independent draws from $P_{XY}$, the empirical estimate stated in~\eqref{eq:empirical-hsic} is constructed using the above expression for $\hsic$. 
    Under the null, the statistic $\hsichat$ is a degenerate one-sample U-statistic, and thus, for fixed kernels $k$ and $\ell$, its asymptotic null distribution is an infinite weighted combination of independent $\chi^2$ random variables, as shown by~\citet[Theorem~2]{gretton2007kernel}, where the weights depend on the distribution $P_{XY}$. 
   Due to the intractable nature of this distribution, practical independence tests based on the $\hsichat$ statistic are usually calibrated using the permutation distribution, that leads to a significant increase in computation. 
   
     The $\hsic$ metric is  also known to be equal to the squared MMD distance between $P_{XY}$ and the product of marginals, $P_X \times P_Y$, using the product kernel $K((x,y), (x', y')) = k(x, x') \ell(y, y')$. With the notation $\mu = \mathbb{E}_{P_X}[k(X, \cdot)]$, $\nu = \mathbb{E}_{P_Y}[\ell(Y, \cdot)]$ and $\omega = \mathbb{E}_{P_{XY}}[k(X, \cdot) \ell(Y, \cdot)]$, we can write $\hsic(P_{XY}, k, \ell)$ as 
        \begin{align}
            \hsic(P_{XY}, k, \ell) = \mmd^2(P_{XY}, P_X \times P_Y) = \langle \omega - \mu \times \nu,\; \omega - \mu \times \nu \rangle_{k \times \ell}, 
        \end{align}
    where  $\langle \cdot, \cdot \rangle_{k \times \ell}$ denotes the inner product in the RKHS associated with the product kernel $k \times \ell$. 
    We will use the above formulation in the next section to propose a new statistic with a tractable asymptotic null distribution. 
    
\section{The Cross-HSIC Test}
\label{sec:cHSIC}
    In this section, we develop a new kernel-based independence test that does not rely on permutations. To achieve this, we propose a new statistic, called cross-HSIC, that has a standard Gaussian limiting null distribution under mild conditions on the kernels and the distribution $P_{XY}$. 
    The new test statistic  relies on two key ideas of sample splitting and studentization, which we describe next. 
    
    Given $2n$ independent draws from the joint distribution $P_{XY}$, denoted by $\data = \{(X_1, X_2), \ldots, (X_{2n}, Y_{2n})\}$,  the studentized cross-HSIC statistic, $\cshsic$, is defined in two steps: 
    \begin{itemize}
        \item First, we split $\data$ into two equal parts, denoted by $\dataone = \{(X_1, Y_1), \ldots, (X_n, Y_n)\}$ and $\datatwo = \{(X_{n+1}, Y_{n+1}), \ldots, (X_{2n}, Y_{2n})\}$. With $Z_i$ denoting the paired observations $(X_i, Y_i)$, we introduce the following terms:
        \begin{align}
            &h_{ij} \equiv h(Z_i, Z_j) = \frac{1}{2} \bigl\{ \phi(X_i) - \phi(X_j) \bigr\}\otimes \bigl\{ \psi(Y_i) - \psi(Y_j) \bigr\}, \\
            &f_1 = \frac{1}{n(n-1)} \sum_{i=1}^n \sum_{\substack{j=1 \\ j\neq i}}^n h_{ij}, \quad \text{and} \quad f_2 = \frac{1}{n(n-1)} \sum_{t=n+1}^{2n} \sum_{\substack{u=n+1 \\ n\neq t}}^{2n} h_{tu}. 
        \end{align}
        These terms  can then be used to define the cross-HSIC statistic 
        \begin{align}
            &\crosshsic \defined \la f_1, f_2 \ra = \frac{1}{n(n-1)} \sum_{1 \leq i \neq j \leq n} \la h_{ij}, f_2 \ra. \label{eq:cross-hsic-def1}
        \end{align}
        It is easy to check that, for any $i \neq j$, we have $\mathbb{E}[h_{ij}] = \omega - \mu \times \nu$. Hence, $f_1$ and $f_2$ are two independent unbiased estimates of $\omega - \mu \times \nu$, which implies that $\crosshsic$ is an unbiased estimate of the HSIC metric.  
        As the expression in~\eqref{eq:cross-hsic-def1} indicates, conditioned on the second-half of the data~(i.e., $\datatwo$), $\crosshsic$ is a one-sample U-statistic. 
        
        \item Our final statistic, denoted by~$\cshsic$, is obtained by normalizing the cross-statistic $\crosshsic$ with the empirical standard deviation. Stated formally, the cross HSIC statistic is defined as 
        \begin{align}
            & \cshsic \defined \frac{ \sqrt{n} \; \crosshsic}{s_n}, \quad \text{where} \label{eq:chsic-def} \\
            & s_n^2 = \frac{4(n-1)}{(n-2)^2} \sum_{i=1}^n \lp \frac{1}{n-1} \sum_{1 \leq j \neq i \leq n} h(Z_i, Z_j) - \crosshsic \rp^2. \label{eq:sn2-def}
        \end{align}
        We note that $n^{-1}s_n^2$ is the jackknife estimator of the variance of $\mathbb{E}[h(Z_1,Z_2) | Z_1]$ also considered in \citet{jing2000berry}.
    \end{itemize}
    
    In~\Cref{sec:null-distribution}, we show that under certain assumptions on the kernel, the $\cshsic$ statistic has a standard normal limiting distribution under the null. This suggests the following natural level-$\alpha$ test of independence: 
    \begin{align}
        \Psi \equiv \Psi\lp \data \rp = \ind_{\cshsic > z_{1-\alpha}},   \label{eq:cross-HSIC-test}
    \end{align}
    where $z_{1-\alpha}$ is the $(1-\alpha)$-quantile of the $N(0,1)$ distribution. 
    
\subsection{Quadratic-time Computation of \texorpdfstring{$\cshsic$}{xHSIC}}
\label{subsec:computational-trick}

    A naive computation of the $\cshsic$ statistic has a $\mc{O}(n^4)$ computational complexity, that is infeasible for all but very small problems. However, a more careful look at the terms involved in defining the statistic indicates that it can actually be computed in quadratic time. 
    
    \begin{theorem}
        \label{theorem:quadratic-cost} 
        The~$\cshsic$, introduced in~\eqref{eq:chsic-def}, can be computed in~$\mc{O}(n^2)$ time.
    \end{theorem}
    
    \emph{Proof outline of~\Cref{theorem:quadratic-cost}.} It suffices to prove that, both, the numerator~($\crosshsic$) and the denominator~($s_n$) in the statistic~$\cshsic$ can be computed in quadratic time. 
    First, we consider the numerator, and observe that it can be decomposed as follows: 
    \begin{align}
        \crosshsic &= T_1 - T_2 - T_3 + T_4, \quad \text{where} \\
        T_1 &=  \frac{1}{n^2} \sum_{i=1}^n \sum_{j=n+1}^{2n} k(X_i, X_j) \ell(Y_j, Y_j), \label{eq:T1-def} \\
        T_2 &= \frac{1}{n^2(n-1)} \sum_{i=1}^n \sum_{n+1 \leq j_1 \neq j_2 \leq 2n} k(X_i, X_{j_1})\ell(Y_i, Y_{j_2}), \label{eq:T2-def}\\
        T_3 & = \frac{1}{n^2(n-1)} \sum_{i=n+1}^{2n} \sum_{1 \leq j_1 \neq j_2 \leq n} k(X_i, X_{j_1})\ell(Y_i, Y_{j_2}), \quad \text{and} \label{eq:T3-def}\\
        T_4 & = \frac{1}{n^2(n-1)^2} \sum_{1 \leq i_1 \neq i_2 \leq n} \;\sum_{n+1 \leq j_1 \neq j_2 \leq 2n} k(X_{i_1}, X_{j_1}) \ell(Y_{i_2}, Y_{j_2}). \label{eq:T4-def}
    \end{align}
    The above expressions indicate that a direct computation of $T_1, T_2, T_3$ and $T_4$ incur $\mc{O}(n^2)$, $\mc{O}(n^3)$, $\mc{O}(n^3)$ and $\mc{O}(n^4)$ cost respectively, making the overall computational cost $\mc{O}(n^4)$. In~\Cref{lemma:quadratic-T2-T3}, we show how the terms $T_2$ and $T_3$ can be computed in quadratic time, and in~\Cref{lemma:quadratic-T4} we show how the term $T_4$ can be computed in quadratic time. These two results together imply that the numerator of $\cshsic$ has an overall quadratic complexity. 
    Finally, we consider the denominator $s_n$, and first note that 
    \begin{align} \label{eq: alternative expression of s_n}
        s_n^2 = \frac{4(n-1)}{(n-2)^2} \lb \frac{1}{(n-1)^2} \sum_{i=1}^n \lp \sum_{j=1}^{n, j\neq i} \la h(Z_i, Z_j), f_2 \ra \rp^2 - n \crosshsic^2  \rb. 
    \end{align}
    We complete the proof by showing that the first term inside the square brackets above can also be computed in quadratic time. The details are in~\Cref{lemma:quadratic-sn}. 
    \qed

\subsection{Connections to \texorpdfstring{$\crossmmd$}{xMMD}} 
\label{subsec:connections-to-mmd} 
    The kernel-MMD distance  between two probability distributions (on the same observation space) is  a widely used metric in the two-sample testing literature. With  $\mc{X}=\mc{Y}$ and $k=\ell$, the \mmd distance between $P_X$ and $P_Y$ is defined as $\mmd(P_X, P_Y) \defined \langle \mu - \nu, \mu - \nu \rangle_k^{1/2}$.  The usual empirical estimates of $\mmd^2(P_X, P_Y)$ based on observations are known to have a complex asymptotic null distribution (similar to $\hsic$). To address this,~\citet{shekhar2022permutation} proposed the $\crossmmd$ statistic, based on splitting the observations $\data=\{(X_i, Y_i): i \in [2n]\}$ into two (usually equal) splits: $\dataone=\{(X_i, Y_i): i \in [n]\}$ and $\datatwo=\{(X_i, Y_i): n+1 \leq i \leq 2n\}$. In particular, let $(\muhat_1, \nuhat_1)$ and $(\muhat_2, \nuhat_2)$  denote the empirical estimates of $\mu$ and $\nu$ based on $\dataone$ and $\datatwo$ respectively. Then, the $\crossmmd$ is defined as 
    \begin{align}
        \crossmmd \equiv \crossmmd(\data) = \langle \muhat_1- \nuhat_1, \, \muhat_2 - \nuhat_2 \rangle_{k}.  \label{eq:cross-mmd}
    \end{align}
    \citet{shekhar2022permutation} showed that the asymptotic distribution of a studentized version of $\crossmmd$ under the null, that is with $P_X = P_Y$,  is $N(0,1)$ under mild conditions. 
 
    As we mentioned earlier in~\Cref{sec:cHSIC}, the $\hsic$ metric can be interpreted as the kernel-MMD distance, with the product kernel $k \times \ell$, between the joint distribution $P_{XY}$ and the product of marginals $P_X \times P_Y$. 
    Hence, similar to the definition of~$\crossmmd$, the $\crosshsic$ statistic based on $\data$, can be rewritten as 
    \begin{align}
        \crosshsic \equiv \crosshsic(\data) = \langle \muhat_1\times \nuhat_1 - \omegahat_1, \; \muhat_2\times \nuhat_2 - \omegahat_2 \rangle_{k \times \ell}, \label{eq:cross-hsic-def-0}
    \end{align}
    where $(\omegahat_1, \muhat_1, \nuhat_1)$ and $(\omegahat_2, \muhat_2, \nuhat_2)$ denote the empirical estimates of $(\omega, \mu, \nu)$ based on~$\dataone$ and $\datatwo$ respectively. 
    
    Given the similarity in the definitions of~$\crosshsic$ and $\crossmmd$, it might appear that the theoretical guarantees of~$\crossmmd$  also carry over directly to the case of $\crosshsic$. However, there is a subtle issue that prevents this. For analyzing $\crossmmd$, \citet{shekhar2022permutation} rely strongly on the fact that $\muhat_1$ and $\nuhat_1$ are independent, and thus, conditioned on the second half of data, the terms $\langle \muhat_1, \muhat_2 - \nuhat_2\rangle_k$ and $\langle \nuhat_1, \muhat_2 - \nuhat_2\rangle_k$ can be analyzed separately, and shown to converge~(after appropriate normalization) to conditionally independent Gaussian distributions. The final distribution of the studentized $\crossmmd$ statistic is then obtained by using the fact that the distribution of the sum of two Gaussian distributions is also Gaussian.  
    
    With~$\crosshsic$ as defined in~\eqref{eq:cross-hsic-def-0}, however, the terms $\omegahat_1$ and $\muhat_1 \times \nuhat_1$ are not independent. 
    Hence, the techniques used for analyzing $\crossmmd$ do not directly apply to our case, and  we develop a different approach for analyzing the studentized version of $\crosshsic$ statistic in the next two sections.

\section{Warmup: Testing Linear Dependence in One Dimension}
\label{sec:warmup-linear} 
    The theoretical properties of our test, under both null and alternative, are subtle and nontrivial.
    Thus, to build intuition gently for later generalizations in the paper, we first briefly analyze our cross-HSIC test in the simplest testing: testing linear dependence of real-valued random variables $X$ and $Y$. 
    
    Formally, this section studies the following (much simpler) problem: given $\data$ drawn according to a joint distribution $P_{XY}$ over $\mathbb{R}^2$, suppose we want to test whether $P_{XY}=P_X \times P_Y$ or that  $\rho = \cov(X, Y) = \mathbb{E}_{P_{XY}}[(X- \mathbb{E}_{P_X}[X])(Y-\mathbb{E}_{P_Y}[Y])$ is non-zero. Of course, there are many methods for this setting and we do not recommend using ours in particular; the only reason to focus only on our statistic is in order to get a handle of the more complex theoretical analysis that follows in the general case.

    To formally describe our results in this setting, we set the observation spaces $\mc{X}$ and $\mc{Y}$ to  $\mathbb{R}$, and assume that the kernels $k$ and $\ell$ are both linear kernels: i.e., $k(x, x') = xx'$ and $\ell(y, y') = yy'$.
    Without loss of generality, we assume that the mean of both $X$ and $Y$ are equal to $0$. In this case, the cross-HSIC statistic can be written as follows: 
    \begin{align}
        \crosshsic 
        &= \lp \frac{1}{n} \sum_{i=1}^n X_i Y_i - \frac{1}{n(n-1)} \sum_{1 \leq i \neq j \leq n} X_i Y_j \rp \times \lp \frac{1}{n} \sum_{t=1}^n X_{n+t} Y_{n+t} - \frac{1}{n(n-1)} \sum_{1 \leq t \neq u \leq n} X_{n+t} Y_{n+u} \rp. \label{eq:xhsic-linear} \\
        & = \lp \frac{1}{n} \sum_{i=1}^n X_i Y_i - \frac{1}{n(n-1)} \sum_{1 \leq i \neq j \leq n} X_i Y_j \rp \times  f_2. 
    \end{align}
    
     In our first result of this section, we consider the case in which the distribution $P_{XY}$ is fixed, and we characterize the performance of our proposed cross-HSIC test under mild second moment requirements on $X$ and $Y$. 
    \begin{theorem}
        \label{thm:warmup-linear-1} 
        Suppose $\data = \{(X_i, Y_i) \in \mathbb{R}^2: 1 \leq i \leq 2n\}$ is drawn \iid from a distribution $P_{XY}$, with  $\mathbb{E}_{P_{XY}}[(XY)^2] + \mathbb{E}_{P_X}[X^2] \mathbb{E}_{P_Y}[Y^2]<\infty$. Then, we have the following: 
        \begin{enumerate}[label=(\alph*)]
            \item If $X$ and $Y$ are independent, then $\cshsic$ computed with linear kernels $k$ and $\ell$ converges in distribution to $N(0,1)$. This means that the cross-HSIC test controls the type-I error asymptotically at level-$\alpha$ for any $P_{XY}= P_X \times P_Y$. 
            \item If $X$ and $Y$ are not independent, and $\mathbb{E}_{P_{XY}}[XY] = \rho \neq 0$, then $ \mathbb{P}(\Psi = 1) = \mathbb{P}\lp \cshsic > z_{1-\alpha} \rp \to 1$. In other words, the cross-HSIC test with linear kernels is consistent against fixed alternatives. 
        \end{enumerate}
    \end{theorem}
    The proof of this statement, presented in~\Cref{proof:linear-kernel}, relies on the special structure of the linear kernel in the case of real-valued observations, that allows the $\crosshsic$ statistic to factor into the product of two independent components. Nevertheless, in the next section, we will show that similar results can be obtained for the case of general kernels using different proof techniques. 
    
    In our next result,  we show that by imposing stronger moment assumptions on the distributions, we can extend  the above result to hold for data distributions changing with $n$. That is, we show that the cross-HSIC test controls type-I error at level-$\alpha$ uniformly over a large class of null distributions, and it is also consistent against local alternatives separated by a $\Omega(1/\sqrt{n})$ detection boundary.   
    \begin{theorem}
        \label{thm:warmup-linear-2} 
        Suppose $\data = \{(X_i, Y_i) \in \mathbb{R}^2: 1 \leq i \leq 2n\}$ is drawn \iid from a distribution $P_{XY, n}$. Introduce the term $\mc{P}_n = \{ P_{XY,n}: \frac{\mathbb{E}_{P_{XY, n}}[(XY)^4]}{n} \lp 1+  \frac{1}{\var_{P_{XY,n}}^2(XY)}  \rp = o(1)\}$, and  define the following null and alternative classes of distributions: 
        \begin{align}
            \nullclass &= \{ P_{XY, n} \in \mc{P}_n: P_{XY,n}=P_{X, n} \times P_{Y, n}\}, \\
            \altclass &= \{ P_{XY, n} \in \mc{P}_n:  \lim_{n \to \infty} |\mathbb{E}_{P_{XY, n}}[XY]| \sqrt{n}>0, \text{ and } \var_{P_{XY,n}}(XY)=1\}.  
        \end{align}
        Then, we have the following:
        \begin{enumerate}[label=(\alph*)]
            \item  The $\cshsic$ statistic converges in distribution to $N(0,1)$ uniformly over the null distribution class; which implies the following bound on the type-I error of the cross-HSIC test:
            \begin{align}
                \lim_{n \to \infty} \sup_{P_{XY, n} \in \nullclass}\;  \mathbb{E}_{P_{XY,n}}[ \Psi(\data)]  =  \alpha. 
            \end{align}%
            \item Consider a sequence of distributions $\{P_{XY, n} \in \altclass: n \geq 1\}$, such that $\mathbb{E}_{P_{XY,n}}[XY] = \rho_n$, with $\lim_{n \to \infty} \rho_n \sqrt{n} = c\neq 0$. Then, we have 
            \begin{align}
                \lim_{n \to \infty} \mathbb{E}_{P_{XY, n}}[\Psi(\data)] & = \Phi(c) \times \Phi(z_{\alpha} + c) + \Phi(-c) \times \Phi(z_{\alpha}- c).
            \end{align}
        In particular, the test is consistent for all distributions for which $\lim_{n \to \infty} |\mathbb{E}_{P_{XY,n}}[XY]|\sqrt{n} =\infty$. 
        \end{enumerate}
    \end{theorem}
    The type-I error result of the above theorem relies on the asymptotic normality of $\cshsic$ that holds uniformly over $\mathcal{P}_n$. We note that the conditions for $\mathcal{P}_n$ involve the fourth and second moments of $XY$. The next example justifies these moment conditions by showing that the asymptotic normality of $\cshsic$ is no longer guaranteed without such conditions. 
    
    \begin{example} \label{exam: Poisson}Suppose that $X_1,\ldots,X_n, Y_1,\ldots, Y_n$ are \iid random variables following a Bernoulli distribution with parameter $p_n$. Assume that $np_n^2 = \lambda >0$ and $\lambda$ is not an integer for all $n$. In this scenario, the limiting null distribution of $\cshsic$ is not Gaussian and it instead satisfies
    	\begin{align}
    		\lim_{n \rightarrow \infty} \mathbb{P}(\cshsic \leq 0) = \mathbb{P}\bigl(\mathrm{sign}(V') \times V \leq 0\bigr),
    	\end{align}
    	where $V,V'$ are \iid centered Poisson random variables with parameter $\lambda$. 
    \end{example}
    We first note that depending on the value of $\lambda$, the limiting probability in \Cref{exam: Poisson} is far from $1/2$, which should be the case if $\cshsic$ is asymptotically Gaussian. 
    It is also worth mentioning that $n^{-1} \mathbb{E}_{P_{XY, n}}[(XY)^4] \var_{P_{XY,n}}^{-2}(XY) \gtrsim \lambda^{-1}$ under the conditions of \Cref{exam: Poisson}. Thus the condition for $\mathcal{P}_n$ is violated. Lastly, we note that $\lambda$ is assumed to be a non-integer for technical reasons and it may be removed with more effort. The detailed derivation of \Cref{exam: Poisson} can be found in \Cref{appendix: poisson example}. 
    
    In the next section, we will see that  we can  obtain similar results with general observation spaces, controlling the type-I error of cross-HSIC test uniformly over composite null classes and establishing its consistency against local alternatives.

\section{Asymptotic Gaussian Distribution under Null}
\label{sec:null-distribution}
    
    In~\Cref{sec:warmup-linear}, we considered a simpler version of this problem, where the observations are real-valued,  and the linear kernel is employed for both $k$ and $\ell$, showing via direct arguments that under the null hypothesis and some moment conditions,  the studentized cross-HSIC statistic, $\cshsic$ converges in distribution to $N(0,1)$. 
    In this section, we generalize this to hold for more general kernels. In particular, we first consider the case where the kernels $k$ and $\ell$, as well as the distribution $P_{XY}$, are fixed with $n$ in~\Cref{theorem:null-dist-0}. Finally, we consider the most general with with $n$-varying kernels and distribution in~\Cref{theorem:null-dist-1}. 
    
    First, we introduce the ``centered"  kernels $\tildek$ and $\tildel$, defined as \begin{align}
        \tildek(x, \cdot) = k(x, \cdot) - \mu, \quad \text{and} \quad \tildel(y, \cdot) \defined \ell(y, \cdot) - \nu. \label{eq:centered-kernels-def}
    \end{align}
    Let $\mc{Z} = \mc{X} \times \mc{Y}$, and introduce the kernel $\tildeg: \mc{Z} \times \mc{Z} \to \mathbb{R}$, defined as 
    \begin{align}
        \tildeg(z, z') = \tildeg\big( (x, y), (x', y') \big) = \langle \tildek(x, \cdot), \tildek(x', \cdot) \rangle \langle \tildel(y, \cdot), \tildel(y', \cdot) \rangle \label{eq:centered-inner-product-0}
    \end{align}
    Note that $\tildeg$ is a symmetric function that is also square integrable if the kernel $k \times \ell$  is square integrable, when $(X, Y) \sim P_{XY, n}$ for some $n \geq 1$. Under this assumption, it admits the following orthonormal expansion: 
    \begin{align}
        \tildeg(z, z') = \sum_{i=1}^{\infty} \lambda_{i,n} e_{i, n}(z) e_{i,n} (z'), \label{eq:eigendecomposition}
    \end{align}
    where $\{(\lambda_{i,n}, e_{i,n}): i \geq 1\}$ for the orthonormal sequences of eigenvalue-eigenfunction pairs  of the Hilber-Schmidt operator associated with the product kernel $K = k \times \ell$; that is, $e \mapsto \int_{\mc{Z}} e(z)K(z, \cdot) dP_{XY,n}(z)$, for any $e \in L^2(P_{XY,n})$. 
    
    We now state our first main result of this section that shows that if $\tildeg(Z_1, Z_2)$ has a finite second moment under the null, then $\cshsic$ converges in distribution to a standard normal random variable. 
    \begin{theorem}
        \label{theorem:null-dist-0} Suppose the $\cshsic$ statistic is computed with fixed kernels $k$ and $\ell$, with observations $\data$ drawn \iid from a fixed distribution $P_{XY}$. If $k, \ell$ and $P_{XY}$ satisfy the condition $\mathbb{E}[\tildeg(Z_1, Z_2)^2] < \infty$ with $Z_1, Z_2 \sim P_{XY} = P_X \times P_Y$, then we have $\cshsic \convdist N(0,1)$. 
    \end{theorem}
    
    
    The finite second moment assumption on $\tildeg$  is  commonly used in the literature for studying the limiting distribution of HSIC statistic~\citep[e.g.,][]{gretton2007kernel} using the asymptotic theory of degenerate U- or V-statistics~\citep[e.g.,][]{serfling2009approximation}. In \Cref{theorem:null-dist-0}, we show that under the same condition,  our $\cshsic$ attains an asymptotic normal limiting  distribution. However, this type of fixed-asymptotic analysis excludes more dynamic and arguably more interesting settings where $k,\ell, P_{XY}$ may change with the sample size. Therefore we now present a more general condition in Assumption~\ref{assump:null-main} which requires higher moment conditions than the finite-second moment of $\tildeg$. 
          
    \begin{assumption}
        \label{assump:null-main}
        Let $\{P_{XY,n}:n \geq 1\}$ denote a sequence of probability distributions on $\mc{Z} = \mc{X} \times \mc{Y}$, with $P_{XY,n} = P_{X,n} \times P_{Y, n}$, and let $\{K_n = k_n \times  \ell_n: n \geq 1\}$ denote a sequence of  positive definite kernels on $\mc{Z}$. 
        With $Z_{1,n}, Z_{2,n}, Z_{3,n}$ denoting three independent draws from $P_{XY, n}$, we assume that 
        \begin{align}
          &\lim_{n \to \infty} \; \frac{\mathbb{E}[\tildeg(Z_{1,n}, Z_{2,n})^4]n^{-1} + \mathbb{E}[\tildeg(Z_{1,n}, Z_{2,n})^2 \tildeg(Z_{1,n}, Z_{3,n})^2] }{n\mathbb{E}[\tildeg(Z_{1,n}, Z_{2,n})^2]^2} = 0,  \quad \text{and} \\
          &\lim_{n \to \infty}\; \frac{\lambda_{1,n}^2}{\sum_{i=1}^{\infty} \lambda_{i,n}^2} \text{ exists}. 
        \end{align}
        Recall that $\tildeg$ was introduced in~\eqref{eq:centered-inner-product-0}, and $\{\lambda_{i,n}:i \geq 1\}$ in~\eqref{eq:eigendecomposition}. 
    \end{assumption}
    
     We now present the next result which establishes the asymptotic normality of the $\cshsic$ statistic for $n$-varying kernels and distributions. 
    \begin{theorem}
        \label{theorem:null-dist-1}
        Suppose the $\cshsic$ is computed with kernels $\{k_n, \ell_n: n \geq 1\}$, and let $\nullclass$ denote the set of distributions $P_{XY,n} = P_{X,n} \times P_{Y,n}$ satisfying~\Cref{assump:null-main}, for each $n \geq 1$. 
        Then,   the $\cshsic$ statistic computed with kernels $k_n$ and $\ell_n$ converges in distribution to $N(0,1)$ uniformly over $\nullclass$. 
    \end{theorem}
    The proof of this statement is given in~\Cref{proof:null-distribution}, and it proceeds by verifying that under~\Cref{assump:null-main}, the conditions required for the Berry-Esseen theorem for studentized U-statistics derived by~\citet{jing2000berry} are satisfied. 

    \begin{remark}
        Note that we have presented all the results of this section under the assumption that the two splits, $\dataone$ and $\datatwo$, are drawn \iid from the same distribution $P_{XY}$. However, a closer inspection of the proof of these results indicates that the  asymptotic normality of the $\cshsic$ statistic~(and the $\csdcov$ statistic, introduced later in~\Cref{sec:dcov})  holds even when the two splits are only independent, and  not identically distributed. Thus, the techniques developed in this paper can be used to analyze the $\cshsic$ statistic in  more general scenarios; such as when $\dataone$ and $\datatwo$ are obtained by separately processing independent outputs of some common source. 
    \end{remark}

\section{Power of the cross-HSIC Test}
\label{sec:power}
        The results of the previous section establish the limiting standard normal distribution of the $\cshsic$ statistic under very general conditions, which in turn, implies that the cross-HSIC test controls type-I error at the desired level $\alpha$ asymptotically. In this section, we analyze the power of our test under the alternative. In particular, we first consider the case of an arbitrary fixed alternative, and prove the consistency of the cross-HSIC test with characteristic kernels in this case. Then,  we consider the case of smooth local alternatives, and show that the cross-HSIC test with Gaussian kernels has minimax rate-optimal power. The proof of both these results can be inferred from a more general result~(\Cref{theorem:general-consistency}) that establishes sufficient conditions for the consistency of our cross-HSIC  test, which we state and prove in~\Cref{appendix:general-consistency}.

 \paragraph{Consistency against fixed alternatives.}
    First, we show that under some mild moment assumptions on the kernels,  our cross-HSIC test is consistent against an arbitrary fixed alternative. That is, for some joint distribution $P_{XY}$ such that $P_{XY} \neq P_X \times P_Y$,  the cross-HSIC test with fixed characteristic kernels $k$ and $\ell$ satisfying some moment conditions is consistent. 
    
    \begin{theorem}
        \label{prop:fixed-alternative} 
        Let $P_{XY}$ be a distribution such that $\hsic(P_{XY}, \calK, \calL)>0$, where $\calK$ and $\calL$ denote the RKHS associated with characteristic kernels $k$ and $\ell$. Then, if $\mathbb{E}[k(X,X)\ell(Y,Y)] + \mathbb{E}[k(X,X)]\mathbb{E}[\ell(Y,Y)]<\infty$, the cross-HSIC test is consistent; that is, $\lim_{n \to \infty} \mathbb{P}_{P_{XY}}(\Psi = 1) = 1$. 
    \end{theorem}
    The details of obtaining this result from the more general result of~\Cref{theorem:general-consistency} mentioned above are given in~\Cref{proof:fixed-alternative}.  We now study the case of local alternatives, where the alternative distributions are allowed to change with $n$.

   \paragraph{Consistency against smooth local alternatives.}
    In our next result, we show that when constructed using Gaussian kernels, the cross-HSIC test is also minimax rate optimal against local alternatives with smooth density functions that are separated in $L^2$ sense. 
    In particular, we set $\mc{X} = \reals^{d_1}$ and $\mc{Y} = \reals^{d_2}$, with $d\defined d_1 + d_2$, and assume that the distributions $P_{XY}$, $P_X$ and $P_Y$ have smooth marginals $p_{XY}$, $p_X$ and $p_Y$ respectively. Furthermore, we assume that there exist constants $M, M_X$ and $M_Y$ such that $p_X \in \Sobolev(M_X)$, $p_Y \in \Sobolev(M_Y)$ and $M = M_X \times M_Y$. Here, $\Sobolev(M)$ denotes the ball with radius $M$ in the fractional Sobolev space of order $\beta>0$. We can then define the null class of distributions, $\nullclass$,  and the $\Delta_n$-separated alternative class of distributions, $\altclass$, as follows for all $n \geq 1$
    \begin{align}
        \nullclass &= \{p_{XY} = p_X \times p_Y: p_X \in \Sobolev(M_X), \; p_Y \in \Sobolev(M_Y)\}, \quad \text{and} \label{eq:def-null-class}\\
        \altclass &= \{p_{XY}:  p_X \in \Sobolev(M_X), \; p_Y \in \Sobolev(M_Y), \text{ and } \|p - p_X \times p_Y\|_{L_2} \geq \Delta_n\}. \label{eq:def-alt-class}
    \end{align}
    For the independence testing problem described above, we will show that the cross-HSIC test, when instantiated using Gaussian kernels with appropriate scale factors, is minimax near-optimal. In particular, we define $k_n(x, x') = \exp \lp -c_n \|x-x'\|^2 \rp$ and $\ell_n(y, y') = \exp \lp -c_n \|y-y'\|^2 \rp$, where we have overloaded the term $\|\cdot\|$ to represent the Euclidean norm on both $\mc{X}$ and $\mc{Y}$. 
    \begin{theorem}
        \label{theorem:local-alternative}
        Suppose $\mc{X} = \reals^{d_1}$ and $\mc{Y}=\reals^{d_2}$ with $d=d_1+d_2$, and let $\{\Delta_n: n \geq 1\}$ is a non-negative sequence with $\lim_{n \to \infty} \Delta_n n^{2\beta/(d+4\beta)} = \infty$. Suppose the cross-HSIC test is instantiated with Gaussian kernels, $k_n(x,x') = \exp \lp - c_n \|x-x'\|^2\rp$ and $\ell_n(y, y') = \exp \lp -c_n \|y-y'\|^2 \rp$, with $c_n \asymp n^{4/(d+4\beta)}$. Then, we have the following: 
        \begin{align}
            \lim_{n \to \infty} \sup_{P_{XY,n} \in \nullclass} \mathbb{E}[\Psi] = \alpha, \quad \text{and} \quad \lim_{n \to \infty} \inf_{P_{XY,n} \in \altclass} \mathbb{E}[\Psi] =1.
        \end{align}
    \end{theorem}
    The proof of this statement is given in~\Cref{proof:local-alternative}. The first part of this result implies that the cross-HSIC test controls the type-I error at the specified level $\alpha \in (0,1)$ uniformly over the entire class of null distributions with smooth densities, $\nullclass$. The second part of the result implies that our cross-HSIC test has a detection boundary of the order $\mc{O}\lp n^{-2\beta/(d+4\beta)}\rp$ in terms of the $L^2$-distance. As shown by  \citet{li2019optimality}, this rate cannot  be improved in the worst case, thus establishing the minimax rate-optimality of  our test. More formally, \citet{li2019optimality} showed that if $\lim_{n \to \infty} \Delta_n n^{2\beta/(d+4\beta)} <\infty$, then there exists an $\alpha \in (0,1)$ for which there exists no independence test that is consistent against such local alternatives (separated by $\Delta_n$).
   
\section{The Cross Distance Covariance Test (xdCov)}
    \label{sec:dcov}
    A popular alternative to kernel based method for measuring the dependence between two distributions is the distance-covariance metric introduced by~\citet{szekely2007measuring}. When, $\mc{X} = \reals^{d_1}$ and $\mc{Y} = \reals^{d_2}$ for $d_1, d_2 \geq 1$, the distance-covariance associated with the joint distribution $P_{XY}$ is 
    \begin{align}
        \dcov(P_{XY}) &= \mathbb{E}_{X,X', Y, Y'}[\|X-X'\| \|Y-Y'\|] + \mathbb{E}_{X,X'}[\|X-X'\|] \mathbb{E}_{Y, Y'} \|Y-Y'\|] \\
        & \quad - 2 \mathbb{E}_{X,Y}[ \mathbb{E}_{X'}[\|X-X'\|] \mathbb{E}_{Y'}[\|Y-Y'\|] ]. 
        \label{eq:dcov-1}
    \end{align}
    In the above display, we overload the notation $\|\cdot\|$ to represent the Euclidean norm on both $\mc{X}$ and $\mc{Y}$. The distance-covariance metric has the property that it is equal to $0$ if and only if $X$ and $Y$ are independent; that is, $P_{XY} = P_X \times P_Y$. This measure was extended beyond Euclidean spaces to general metric spaces by~\citet{lyons2013distance}, and further generalized to semi-metric spaces~(i.e., distance-measures that don't satisfy the triangle inequality) by~\citet{sejdinovic2013equivalence}. In particular, if $(\mc{X}, \distX)$ and $(\mc{Y}, \distY)$ are semi-metric spaces, we can define the corresponding distance-covariance in a manner analogous to~\eqref{eq:dcov-1}, as 
    \begin{align}
        \dcov(P_{XY}, \distX, \distY)  &= \mathbb{E}_{X,X', Y, Y'}[\distX(X, X') \distY(Y, Y')] + \mathbb{E}_{X,X'}[\distX(X, X')] \mathbb{E}_{Y, Y'}[\distY(Y, Y')] \\
        & \quad - 2 \mathbb{E}_{X,Y}[ \mathbb{E}_{X'}[\distX(X, X')] \mathbb{E}_{Y'}[\distY(Y, Y')] ]. \label{eq:dcov-2}
    \end{align}
    \citet{sejdinovic2013equivalence} showed an equivalence between the above distance-covariance metric and the $\hsic$ computed with  with the so-called \emph{distance kernels} that we recall next. 
    \begin{fact}
    \label{fact:dcov-hsic-equivalence}
    Define the distance kernels $\dkernelX$ and $\dkernelY$  as 
    \begin{align}
        \dkernelX(x, x') \defined \frac{1}{2} \lp \distX(x, x_0) + \distX(x', x_0) - 2\distX(x, x') \rp,  \quad \text{and} \quad 
        \dkernelY(y, y') \defined \frac{1}{2} \lp \distY(y, y_0) + \distY(y', y_0) - 2\distY(y, y') \rp,  \label{eq:distance-kernel-Y}
    \end{align}
    where $x_0$ and $y_0$ are arbitrary elements of $\mc{X}$ and $\mc{Y}$ respectively. Then, we have $\dcov(P_{XY}, \distX, \distY) = \hsic(P_{XY}, \dkernelX, \dkernelY)$. 
    \end{fact}
    
    Using this equivalence, we can define a cross-distance-covariance statistic similar to the cross-HSIC statistic of~\Cref{sec:cHSIC} by using sample-splitting and studentization. 
    \begin{definition}
    \label{def:cross-dcov} 
        Given observations $\data$ drawn from a distribution $P_{XY}$, we can define the cross-distance-covariance statistic, denoted by $\crossdcov$, as follows: 
        \begin{align}
            &\crossdcov = \la \frac{1}{n(n-1)} \sum_{i=1}^n \sum_{\substack{j=1\\j\neq i}}^n h_{ij},\; \frac{1}{n(n-1)} \sum_{t=n+1}^{2n} \sum_{\substack{u=n+1\\u\neq t}}^{2n} h_{tu} \ra \quad \text{where} \\ 
            &2h_{ij} = a_{ii} + a_{jj} - a_{ij} - a_{ji}, \quad \text{and} \quad a_{ij} = \dkernelX(X_i, \cdot) \dkernelY(Y_j, \cdot). 
        \end{align}
        We can then define a studentized version of cross-distance-covariance statistic, denoted by $\csdcov$, by normalizing $\crossdcov$ with $s_n/\sqrt{n}$, with $s_n$ defined similar to~\eqref{eq:sn2-def}. 
    \end{definition}

    Having defined the studentized cross-dCov statistic~($\csdcov$), we can now characterize its limiting null distribution by exploiting its equivalence to the studentized cross-HSIC statitic, and using the results derived in~\Cref{sec:null-distribution}. 
    \begin{corollary}
        \label{corollary:dcov-null}
        For a fixed distribution $P_{XY} = P_X \times P_Y$, and a fixed pair of points $(x_0, y_0) \in \mc{X} \times \mc{Y}$, if $\mathbb{E}[\rho_{\mc{X}}(X, x_0)^2]\mathbb{E}[\rho_{\mc{Y}}(Y, y_0)^2]<\infty$, then $\csdcov \convdist N(0,1)$. 

        More generally, suppose the sequence of distributions $\{P_{XY,n}: n \geq 1\}$ and the distance kernels $\{(k_{\mc{X}, n}, \ell_{\mc{Y},n}): n \geq 1\}$ and  satisfy \Cref{assump:null-main}. Then, the statistic $\csdcov$  converges in distribution to $N(0,1)$ uniformly over the composite class of null distributions satisfying~\Cref{assump:null-main}. 
    \end{corollary}
    
    The above asymptotic normality of the $\csdcov$ statistic under the null suggests the definition of an independence test~($\testdist$), based on the cross-distance-covariance statistic that rejects the null when $\csdcov$ exceeds $z_{1-\alpha}$, the $1-\alpha$ quantile of the standard normal distribution. That is, $\testdist = \bone_{\csdcov>z_{1-\alpha}}$. Using the analogous results for the $\cshsic$ statistic, we can obtain sufficient conditions for $\testdist$ to be consistent. 
    
    \begin{corollary}
        \label{corollary:dcov-consistency}
        For a fixed alternative distribution, $P_{XY} \neq P_X \times P_Y$, a sufficient condition for consistency of $\testdist$ is that $\mathbb{E}[\distX(X,x_0)^2\distY(Y,y_0)^2] + \mathbb{E}[\distX(X,x_0)^2]\,\mathbb{E}[\distY(Y,y_0)^2] <\infty$. 

        More generally, for a sequence of local alternatives  $\{P_{XY,n}: n \geq 1\}$, if  the distance kernels $\dkernelX$ and $\dkernelY$ satisfy \Cref{theorem:general-consistency} in~\Cref{appendix:general-consistency}, the test $\testdist$ is consistent. 
    \end{corollary}

    \begin{remark}
        \label{remark:dcov-test-1} 
        
        \citet{szekely2007measuring}  proposed a permutation-free independence test based on the dCov statistic, which relied on the fact that asymptotically, a suitably normalized variant of dCov is stochastically dominated by a quadratic form of a centered standard normal random variable~\citep[Theorem 6]{szekely2007measuring}. 
        However, while this result is tight for the case of Bernoulli distributions, in general this can lead to very conservative independence tests. Furthermore, the validity of this approach when the distributions and distance-measures can change with the sample-size has not been established. Our proposed cross-dCov test addresses both these issues. 
    \end{remark}

\section{Experiments} 
\label{sec:experiments}
    In this section, we experimentally validate the theoretical results presented in the previous sections.
    
    \subsection{Null Distribution}
        \paragraph{Sufficiency of finite second moment.} In the first experiment, we verify the claim of~\Cref{theorem:null-dist-0} that states that finite second moment of the kernel is sufficient for the asymptotic normality of the $\cshsic$ statistic under the null. In particular, we use linear kernels $k$ and $\ell$, and consider the case where $P_X$ and $P_Y$ are distribution in $\mathbb{R}^d$; with each component drawn independently from a $t$-distribution with $\texttt{dof}$ degrees of freedom. Recall that such distributions have finite moments of order up to $\texttt{dof}-1$. The null distribution of the $\cshsic$ statistic on such distributions with $\texttt{dof} \in \{1,2,3\}$ are shown in~\Cref{fig:null-x-hsic-t-dist}. For $\dof=1$ and $\dof=3$, the null distribution appears to be clearly non-Gaussian and Gaussian respectively; with $\dof=2$ representing an intermediate state.  
    
        \begin{figure}[htb!]
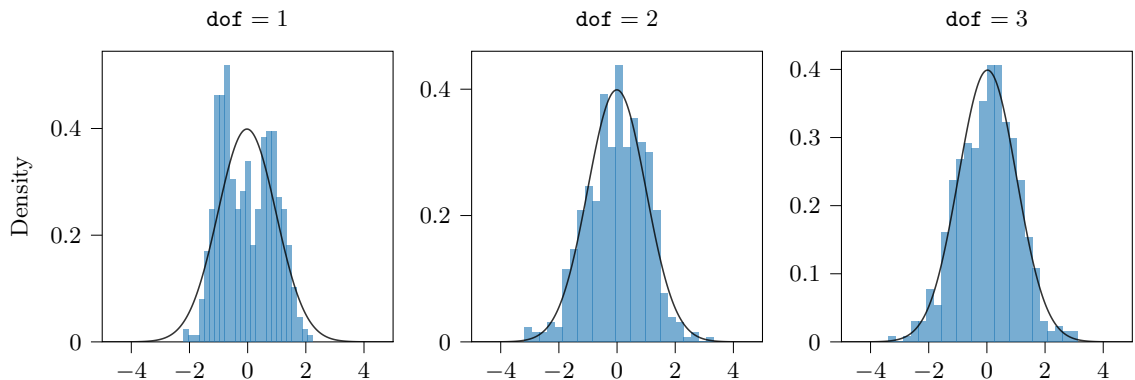

             \def\figwidth{0.33\linewidth}
             \def\figheight{0.33\linewidth} 
            \centering
            \input{Figures/Second_Moment_NullPdf_HSIC_n_500_d_10_dof_1}
            \input{Figures/Second_Moment_NullPdf_HSIC_n_500_d_10_dof_2}
            \input{Figures/Second_Moment_NullPdf_HSIC_n_500_d_10_dof_3}
            \caption{The figures show the null distribution of $\cshsic$  with $n=500$, $d=10$ using $500$ trials. Both $P_X$ and $P_Y$ consist of $d$ independent components drawn from $t$-distributions with degrees of freedom $\mathrm{dof} \in \{1,2,3\}$; the data has finite variance only in the third plot. Thus, the plots above indicate that the existence of finite second moment is sufficient for the asymptotic normality of the $\cshsic$ statistic, and appear somewhat necessary as well (the second plot may be close to Gaussian, but the first plot is far from it).}
            \label{fig:null-x-hsic-t-dist}
        \end{figure}
    
        \paragraph{Effect of kernels and dimension regimes.}  
            In the next, experiment, we verify that the $\cshsic$ statistic has a limiting null distribution for different choices of kernels~(Gaussian vs Rational Quadratic) and in different dimension regimes~($d/n = 1/2$ vs $d/n = 1/20$). The observations are drawn from independent multivariate Gaussian distributions with unit covariance matrix. The results, plotted in~\Cref{fig:null-x-hsic} show that as expected, the null distribution of $\cshsic$ approaches the standard normal distribution, even for relatively small sample sizes~(all the plots have $n=200$). 
            \begin{figure}[htb!]
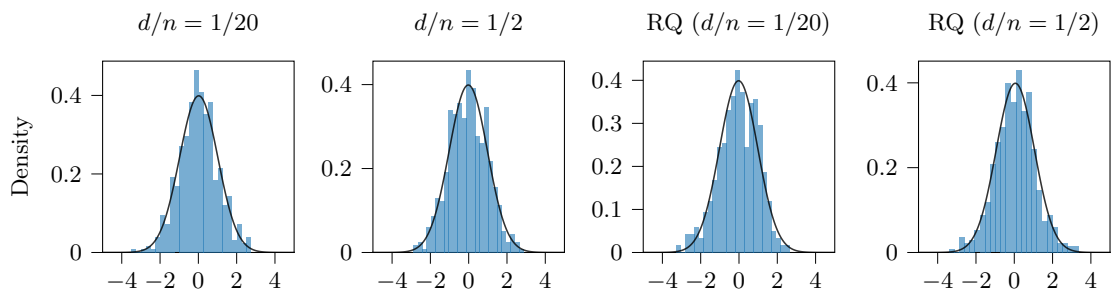

                \def\figwidth{0.25\linewidth}
                \def\figheight{0.25\linewidth} 
                \centering
                \input{Figures/NullPdf_HSIC_n_200_d_10}
                \input{Figures/NullPdf_HSIC_n_200_d_100}
                \input{Figures/NullPdf_HSIC_n_200_d_10_RationalQuadratic}
                \input{Figures/NullPdf_HSIC_n_200_d_100_RationalQuadratic}
                \caption{Plots of the null distribution of the $\cshsic$ statistic for two commonly used kernels, Gaussian~(first two plots) and Rational-Quadratic, under two dimension regimes each, with $d/n \in \{1/2, 1/20\}$, with $n=200$.}
                \label{fig:null-x-hsic}
            \end{figure}
        \paragraph{Control of type-I error.} 
            The previous experiment shows that visually, the null distribution of $\cshsic$ statistic approaches the standard normal distribution even at relatively small $n$ values. We now show in~\Cref{fig:type-I-x-hsic} that this also translates into a tight control over the type-I error at the desired level $\alpha$~($=0.05$ in the plots) of the resulting cross-HSIC test based on the $\cshsic$ statistic. 
             \begin{figure}[htb!]
                \def\figwidth{0.25\linewidth}
                \def\figheight{0.25\linewidth} 
                \centering
\begin{tikzpicture}

\definecolor{darkgray176}{RGB}{176,176,176}
\definecolor{steelblue31119180}{RGB}{31,119,180}

\begin{axis}[
height=\figheight,
tick align=outside,
tick pos=left,
title={$d/n=1/20$},
width=\figwidth,
x grid style={darkgray176},
xlabel={Sample size (n)},
xmin=0.5, xmax=209.5,
xtick style={color=black},
y grid style={darkgray176},
ylabel={Type-I error},
ymin=0, ymax=1.0,
ytick style={color=black}
]
\addplot [semithick, steelblue31119180]
table {%
10 0.05
20 0.048
30 0.048
40 0.046
50 0.04
60 0.04
70 0.074
80 0.054
90 0.06
100 0.034
110 0.042
120 0.052
130 0.046
140 0.05
150 0.058
160 0.042
170 0.05
180 0.05
190 0.058
200 0.066
};
\addplot [semithick, black, opacity=0.15]
table {%
10 0.05
20 0.05
30 0.05
40 0.05
50 0.05
60 0.05
70 0.05
80 0.05
90 0.05
100 0.05
110 0.05
120 0.05
130 0.05
140 0.05
150 0.05
160 0.05
170 0.05
180 0.05
190 0.05
200 0.05
};
\end{axis}

\end{tikzpicture}
\begin{tikzpicture}

\definecolor{darkgray176}{RGB}{176,176,176}
\definecolor{steelblue31119180}{RGB}{31,119,180}

\begin{axis}[
height=\figheight,
tick align=outside,
tick pos=left,
title={$d/n=1/2$},
width=\figwidth,
x grid style={darkgray176},
xlabel={Sample size (n)},
xmin=0.5, xmax=209.5,
xtick style={color=black},
y grid style={darkgray176},
ylabel={},
ymin=0, ymax=1.0,
ytick style={color=black}
]
\addplot [semithick, steelblue31119180]
table {%
10 0.058
20 0.066
30 0.044
40 0.048
50 0.042
60 0.054
70 0.062
80 0.054
90 0.048
100 0.046
110 0.062
120 0.07
130 0.052
140 0.042
150 0.056
160 0.064
170 0.052
180 0.04
190 0.054
200 0.05
};
\addplot [semithick, black, opacity=0.15]
table {%
10 0.05
20 0.05
30 0.05
40 0.05
50 0.05
60 0.05
70 0.05
80 0.05
90 0.05
100 0.05
110 0.05
120 0.05
130 0.05
140 0.05
150 0.05
160 0.05
170 0.05
180 0.05
190 0.05
200 0.05
};
\end{axis}

\end{tikzpicture}
\begin{tikzpicture}

\definecolor{darkgray176}{RGB}{176,176,176}
\definecolor{steelblue31119180}{RGB}{31,119,180}

\begin{axis}[
height=\figheight,
tick align=outside,
tick pos=left,
title={RQ ($d/n=1/20$)},
width=\figwidth,
x grid style={darkgray176},
xlabel={Sample size (n)},
xmin=0.5, xmax=209.5,
xtick style={color=black},
y grid style={darkgray176},
ylabel={},
ymin=0, ymax=1.0,
ytick style={color=black}
]
\addplot [semithick, steelblue31119180]
table {%
10 0.044
20 0.032
30 0.034
40 0.044
50 0.038
60 0.032
70 0.054
80 0.044
90 0.024
100 0.05
110 0.044
120 0.032
130 0.028
140 0.042
150 0.056
160 0.034
170 0.044
180 0.058
190 0.038
200 0.05
};
\addplot [semithick, black, opacity=0.15]
table {%
10 0.05
20 0.05
30 0.05
40 0.05
50 0.05
60 0.05
70 0.05
80 0.05
90 0.05
100 0.05
110 0.05
120 0.05
130 0.05
140 0.05
150 0.05
160 0.05
170 0.05
180 0.05
190 0.05
200 0.05
};
\end{axis}

\end{tikzpicture}
\begin{tikzpicture}

\definecolor{darkgray176}{RGB}{176,176,176}
\definecolor{steelblue31119180}{RGB}{31,119,180}

\begin{axis}[
height=\figheight,
tick align=outside,
tick pos=left,
title={RQ ($d/n=1/2$)},
width=\figwidth,
x grid style={darkgray176},
xlabel={Sample size (n)},
xmin=0.5, xmax=209.5,
xtick style={color=black},
y grid style={darkgray176},
ylabel={},
ymin=0, ymax=1.0,
ytick style={color=black}
]
\addplot [semithick, steelblue31119180]
table {%
10 0.034
20 0.04
30 0.046
40 0.052
50 0.058
60 0.054
70 0.046
80 0.034
90 0.066
100 0.054
110 0.06
120 0.042
130 0.046
140 0.05
150 0.064
160 0.046
170 0.068
180 0.05
190 0.052
200 0.05
};
\addplot [semithick, black, opacity=0.15]
table {%
10 0.05
20 0.05
30 0.05
40 0.05
50 0.05
60 0.05
70 0.05
80 0.05
90 0.05
100 0.05
110 0.05
120 0.05
130 0.05
140 0.05
150 0.05
160 0.05
170 0.05
180 0.05
190 0.05
200 0.05
};
\end{axis}

\end{tikzpicture}
                \caption{The figures show the variation of the type-I error of the cross-HSIC test under two different dimension regimes: $d/n \in \{1/2, 1/20\}$, and for two commonly used kernels: Gaussian~(the first two plots) and Rational Quadratic. In all cases, we see that the type-I error is controlled at level $\alpha=0.05$ for $n \geq 100$.}
                \label{fig:type-I-x-hsic}
            \end{figure}
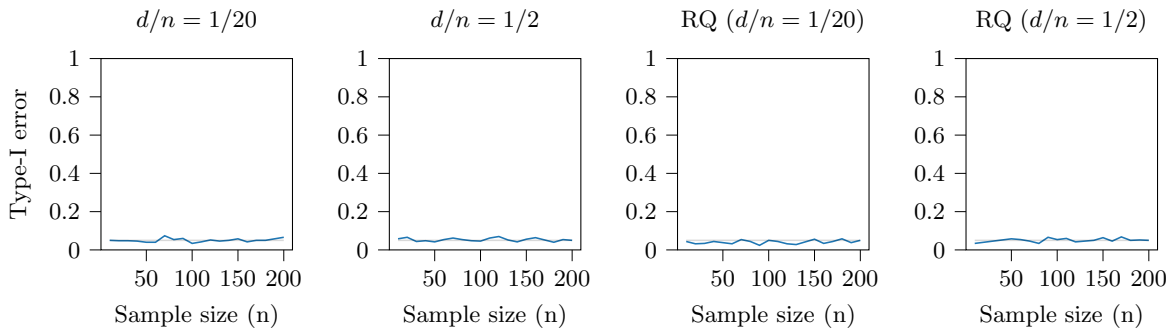
            
    \subsection{Power}
        In this section, we empirically compare the power of our cross-HSIC test with the permutation HSIC test for different choices of kernels and different dimension regimes (similar to the previous subsection). As shown in~\Cref{fig:power-x-hsic}, the power of our cross-HSIC test is smaller than the permutation HSIC test by approximately a factor of $\sqrt{2}$ due to sample-splitting.
        
        In~\Cref{fig:power-x-hsic}, we plot the power curves of our cross-HSIC test and the HSIC permutation test for two kernels~(Gaussian and Rational-Quadratic), for different levels of dependence (measured by $\epsilon$). For this experiment, we set $P_X$ to a multivariate Gaussian distribution in $d$ dimensions with identity covariance matrix; and then generated $Y$ as 
        \begin{align}
            \label{eq:dependent-source}
            Y = \epsilon \times X^b + (1-\epsilon) \times (X')^b, \quad \text{for } b >0. 
        \end{align}
        In the above display, the exponentiation is done component-wise and $X'$ is an independent copy of $X$. 
             \begin{figure}[htb!]
                \def\figwidth{0.33\linewidth}
                \def\figheight{0.33\linewidth} 
                \centering
\begin{tikzpicture}

\definecolor{darkgray176}{RGB}{176,176,176}
\definecolor{darkorange25512714}{RGB}{255,127,14}
\definecolor{lightgray204}{RGB}{204,204,204}
\definecolor{steelblue31119180}{RGB}{31,119,180}

\begin{axis}[
height=\figheight,
legend cell align={left},
legend style={
  fill opacity=0.8,
  draw opacity=1,
  text opacity=1,
  at={(0.97, 0.03)},
  anchor=south east,
  draw=lightgray204
},
tick align=outside,
tick pos=left,
title={$\epsilon=0.3$},
width=\figwidth,
x grid style={darkgray176},
xlabel={Sample-size (n)},
xmin=-34.5, xmax=944.5,
xtick style={color=black},
y grid style={darkgray176},
ylabel={Power},
ymin=0.0256, ymax=1.0464,
ytick style={color=black}
]
\addplot [semithick, steelblue31119180]
table {%
10 0.104
56 0.184
103 0.304
150 0.468
197 0.644
244 0.764
291 0.844
337 0.896
384 0.952
431 1
478 0.988
525 0.996
572 1
618 1
665 1
712 1
759 1
806 1
853 1
900 1
};
\addlegendentry{HSIC-perm}
\addplot [semithick, darkorange25512714]
table {%
10 0.072
56 0.112
103 0.156
150 0.276
197 0.308
244 0.416
291 0.48
337 0.616
384 0.712
431 0.744
478 0.864
525 0.916
572 0.944
618 0.96
665 0.98
712 0.976
759 0.996
806 0.996
853 1
900 1
};
\addlegendentry{$\mathrm{xHSIC}$}
\end{axis}

\end{tikzpicture}
\begin{tikzpicture}

\definecolor{darkgray176}{RGB}{176,176,176}
\definecolor{darkorange25512714}{RGB}{255,127,14}
\definecolor{lightgray204}{RGB}{204,204,204}
\definecolor{steelblue31119180}{RGB}{31,119,180}

\begin{axis}[
height=\figheight,
legend cell align={left},
legend style={
  fill opacity=0.8,
  draw opacity=1,
  text opacity=1,
  at={(0.97,0.03)},
  anchor=south east,
  draw=lightgray204
},
tick align=outside,
tick pos=left,
title={$\epsilon=0.4$},
width=\figwidth,
x grid style={darkgray176},
xlabel={Sample-size (n)},
xmin=-19.5, xmax=629.5,
xtick style={color=black},
y grid style={darkgray176},
ylabel={Power},
ymin=0.013, ymax=1.047,
ytick style={color=black}
]
\addplot [semithick, steelblue31119180]
table {%
10 0.144
41 0.28
72 0.54
103 0.74
134 0.844
165 0.952
196 0.96
227 0.98
258 1
289 1
320 1
351 1
382 1
413 1
444 1
475 1
506 1
537 1
568 1
600 1
};
\addlegendentry{HSIC-perm}
\addplot [semithick, darkorange25512714]
table {%
10 0.06
41 0.104
72 0.208
103 0.268
134 0.428
165 0.54
196 0.656
227 0.744
258 0.86
289 0.892
320 0.944
351 0.98
382 0.976
413 0.988
444 0.988
475 1
506 1
537 1
568 1
600 1
};
\addlegendentry{x-HSIC}; 
\legend{};
\end{axis}

\end{tikzpicture} 
\begin{tikzpicture}

\definecolor{darkgray176}{RGB}{176,176,176}
\definecolor{darkorange25512714}{RGB}{255,127,14}
\definecolor{lightgray204}{RGB}{204,204,204}
\definecolor{steelblue31119180}{RGB}{31,119,180}

\begin{axis}[
height=\figheight,
legend cell align={left},
legend style={
  fill opacity=0.8,
  draw opacity=1,
  text opacity=1,
  at={(0.03,0.97)},
  anchor=north west,
  draw=lightgray204
},
tick align=outside,
tick pos=left,
title={$\epsilon=0.5$},
width=\figwidth,
x grid style={darkgray176},
xlabel={Sample-size (n)},
xmin=-4.5, xmax=314.5,
xtick style={color=black},
y grid style={darkgray176},
ylabel={Power},
ymin=0.0214, ymax=1.0466,
ytick style={color=black}
]
\addplot [semithick, steelblue31119180]
table {%
10 0.252
25 0.468
40 0.64
55 0.772
71 0.784
86 0.856
101 0.948
116 0.976
132 0.98
147 0.996
162 1
177 1
193 1
208 1
223 1
238 1
254 1
269 1
284 1
300 1
};
\addlegendentry{HSIC-perm}
\addplot [semithick, darkorange25512714]
table {%
10 0.068
25 0.1
40 0.16
55 0.236
71 0.236
86 0.356
101 0.48
116 0.572
132 0.664
147 0.752
162 0.868
177 0.9
193 0.912
208 0.944
223 0.972
238 0.984
254 0.992
269 0.996
284 1
300 0.996
};
\addlegendentry{x-HSIC}; 
\legend{};
\end{axis}

\end{tikzpicture}\\
\begin{tikzpicture}

\definecolor{darkgray176}{RGB}{176,176,176}
\definecolor{darkorange25512714}{RGB}{255,127,14}
\definecolor{lightgray204}{RGB}{204,204,204}
\definecolor{steelblue31119180}{RGB}{31,119,180}

\begin{axis}[
height=\figheight,
legend cell align={left},
legend style={
  fill opacity=0.8,
  draw opacity=1,
  text opacity=1,
  at={(0.97,0.03)},
  anchor=south east,
  draw=lightgray204
},
tick align=outside,
tick pos=left,
title={RQ ($\epsilon=0.3$)},
width=\figwidth,
x grid style={darkgray176},
xlabel={Sample-size (n)},
xmin=-34.5, xmax=944.5,
xtick style={color=black},
y grid style={darkgray176},
ylabel={Power},
ymin=-0.008, ymax=1.048,
ytick style={color=black}
]
\addplot [semithick, steelblue31119180]
table {%
10 0.112
56 0.268
103 0.6
150 0.808
197 0.972
244 0.992
291 1
337 1
384 1
431 1
478 1
525 1
572 1
618 1
665 1
712 1
759 1
806 1
853 1
900 1
};
\addlegendentry{HSIC-perm}
\addplot [semithick, darkorange25512714]
table {%
10 0.04
56 0.112
103 0.34
150 0.5
197 0.776
244 0.86
291 0.92
337 0.972
384 0.988
431 1
478 1
525 1
572 1
618 1
665 1
712 1
759 1
806 1
853 1
900 1
};
\addlegendentry{x-HSIC}; 
\legend{}; 
\end{axis}

\end{tikzpicture} 
\begin{tikzpicture}

\definecolor{darkgray176}{RGB}{176,176,176}
\definecolor{darkorange25512714}{RGB}{255,127,14}
\definecolor{lightgray204}{RGB}{204,204,204}
\definecolor{steelblue31119180}{RGB}{31,119,180}

\begin{axis}[
height=\figheight,
legend cell align={left},
legend style={
  fill opacity=0.8,
  draw opacity=1,
  text opacity=1,
  at={(0.97,0.03)},
  anchor=south east,
  draw=lightgray204
},
tick align=outside,
tick pos=left,
title={RQ ($\epsilon=0.4$)},
width=\figwidth,
x grid style={darkgray176},
xlabel={Sample-size (n)},
xmin=-19.5, xmax=629.5,
xtick style={color=black},
y grid style={darkgray176},
ylabel={Power},
ymin=0.0298, ymax=1.0462,
ytick style={color=black}
]
\addplot [semithick, steelblue31119180]
table {%
10 0.124
41 0.52
72 0.86
103 0.96
134 0.996
165 1
196 1
227 1
258 1
289 1
320 1
351 1
382 1
413 1
444 1
475 1
506 1
537 1
568 1
600 1
};
\addlegendentry{HSIC-perm}
\addplot [semithick, darkorange25512714]
table {%
10 0.076
41 0.164
72 0.44
103 0.704
134 0.876
165 0.964
196 0.992
227 1
258 1
289 1
320 1
351 1
382 1
413 1
444 1
475 1
506 1
537 1
568 1
600 1
};
\addlegendentry{x-HSIC}; 
\legend{}; 
\end{axis}

\end{tikzpicture} 
\begin{tikzpicture}

\definecolor{darkgray176}{RGB}{176,176,176}
\definecolor{darkorange25512714}{RGB}{255,127,14}
\definecolor{lightgray204}{RGB}{204,204,204}
\definecolor{steelblue31119180}{RGB}{31,119,180}

\begin{axis}[
height=\figheight,
legend cell align={left},
legend style={
  fill opacity=0.8,
  draw opacity=1,
  text opacity=1,
  at={(0.97,0.03)},
  anchor=south east,
  draw=lightgray204
},
tick align=outside,
tick pos=left,
title={RQ ($\epsilon=0.5$)},
width=\figwidth,
x grid style={darkgray176},
xlabel={Sample-size (n)},
xmin=-4.5, xmax=314.5,
xtick style={color=black},
y grid style={darkgray176},
ylabel={Power},
ymin=-0.0164, ymax=1.0484,
ytick style={color=black}
]
\addplot [semithick, steelblue31119180]
table {%
10 0.236
25 0.54
40 0.776
55 0.932
71 0.976
86 1
101 1
116 1
132 1
147 1
162 1
177 1
193 1
208 1
223 1
238 1
254 1
269 1
284 1
300 1
};
\addlegendentry{HSIC-perm}
\addplot [semithick, darkorange25512714]
table {%
10 0.032
25 0.172
40 0.364
55 0.62
71 0.692
86 0.884
101 0.948
116 0.984
132 0.996
147 0.996
162 1
177 1
193 1
208 1
223 1
238 1
254 1
269 1
284 1
300 1
};
\addlegendentry{x-HSIC}; 
\legend{}; 
\end{axis}

\end{tikzpicture}
                \caption{The figures in the top row show the power curves for HSIC permutation test, and cross-HSIC test with Gaussian kernels, while the bottom row corresponds to the same tests with Rational-Quadratic~(RQ) kernel. In all figures, we have $d=10$ and $b=2$, while $\epsilon$ is set to $0.3, 0.4,$ and $0.5$ in the three columns. Recall that $\epsilon$ and $b$ correspond to the expression in~\eqref{eq:dependent-source}. 
                }
                \label{fig:power-x-hsic}
            \end{figure}
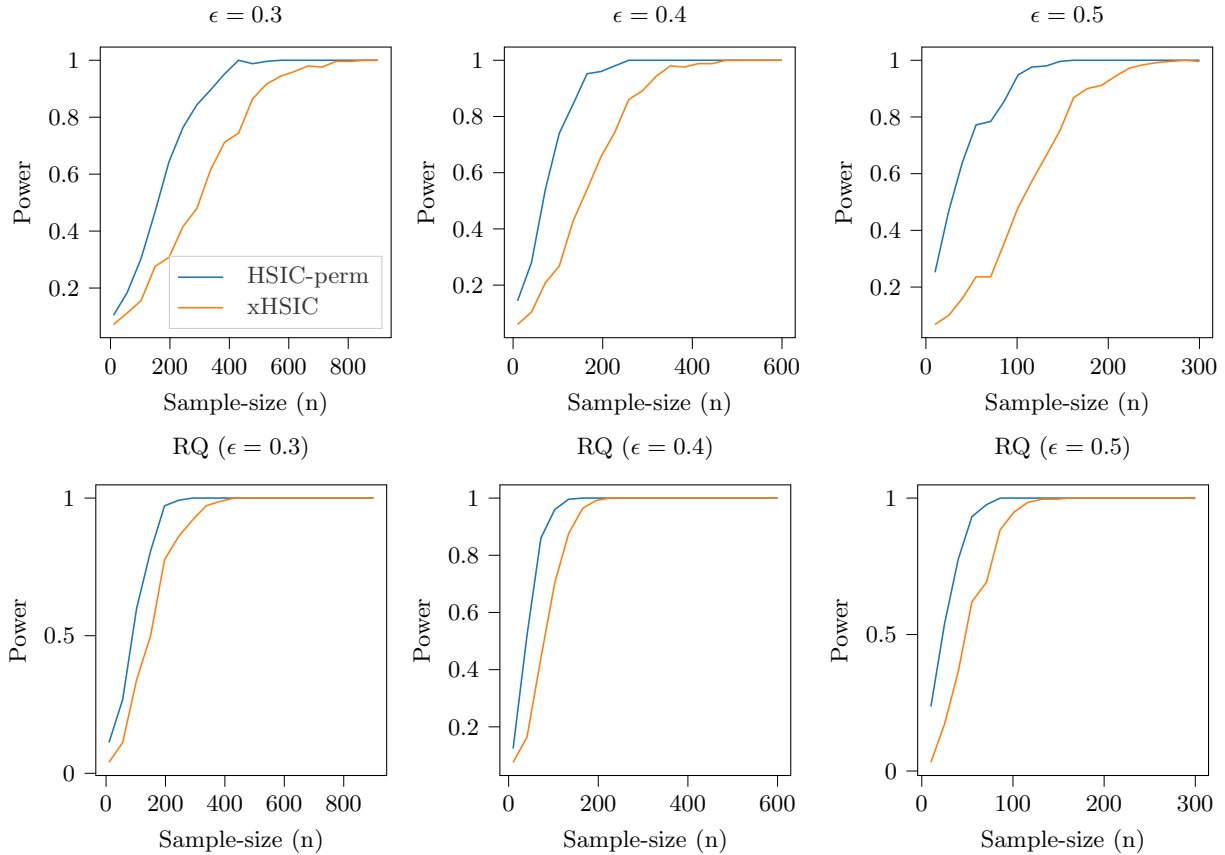
        Finally, to highlight the trade-off achieved by our cross-HSIC test, we plot the power against the running-time of the cross-HSIC test against that of the HSIC permutation test in~\Cref{fig:power-compoute-hsic}

        \begin{figure}[htb!]
            \def\figwidth{0.33\linewidth}
            \def\figheight{0.33\linewidth} 
            \centering
\begin{tikzpicture}

\definecolor{darkgray176}{RGB}{176,176,176}
\definecolor{darkorange25512714}{RGB}{255,127,14}
\definecolor{lightgray204}{RGB}{204,204,204}
\definecolor{steelblue31119180}{RGB}{31,119,180}

\begin{axis}[
height=\figheight,
legend cell align={left},
legend style={
  fill opacity=0.8,
  draw opacity=1,
  text opacity=1,
  at={(0.03,0.97)},
  anchor=north west,
  draw=lightgray204
},
log basis x={10},
tick align=outside,
tick pos=left,
title={Power vs Computation},
width=\figwidth,
x grid style={darkgray176},
xlabel={Running Time (Seconds)},
xmin=0.000714712244023049, xmax=13.8128148491252,
xmode=log,
xtick style={color=black},
xtick={1e-05,0.0001,0.001,0.01,0.1,1,10,100,1000},
xticklabels={
  \(\displaystyle {10^{-5}}\),
  \(\displaystyle {10^{-4}}\),
  \(\displaystyle {10^{-3}}\),
  \(\displaystyle {10^{-2}}\),
  \(\displaystyle {10^{-1}}\),
  \(\displaystyle {10^{0}}\),
  \(\displaystyle {10^{1}}\),
  \(\displaystyle {10^{2}}\),
  \(\displaystyle {10^{3}}\)
},
y grid style={darkgray176},
ylabel={Power},
ymin=-0.008, ymax=1.048,
ytick style={color=black}
]
\addplot [
  draw=darkorange25512714,
  fill=darkorange25512714,
  mark=triangle*,
  only marks,
  opacity=0.6,
    scatter, 
  visualization depends on={\thisrow{sizedata} \as\perpointmarksize},
  scatter/@pre marker code/.style={/tikz/mark size=\perpointmarksize},
  scatter/@post marker code/.style={}
]
table{%
x  y  sizedata
0.0593019223213196 0.15 0.7978845608028654
0.0827778482437134 0.05 1.4927053303604616
0.107946157455444 0.21 1.9706294962377493
0.227671394348145 0.27 2.3534213434057993
0.367848544120789 0.37 2.682126661392722
0.519360957145691 0.42 2.9640095915284457
0.719321272373199 0.5 3.231186201200472
0.947983055114746 0.55 3.477898169151024
1.24455168008804 0.56 3.70823233942262
1.57231640338898 0.71 3.9250730555360964
2.20357163906097 0.79 4.122832512025492
2.71612585306168 0.85 4.318907191682884
3.21402317285538 0.87 4.506458780298101
3.88698330163956 0.97 4.686510657907603
4.52512806415558 0.96 4.85989645515594
5.10633770227432 0.97 5.0209703231304035
6.08884870767593 0.97 5.183179949983594
6.77686154127121 0.99 5.340464942499636
7.84050616502762 0.98 5.493248329560876
8.81976769685745 1 5.641895835477563
};
\addlegendentry{HSIC-perm}
\addplot [
  draw=steelblue31119180,
  fill=steelblue31119180,
  mark=*,
  only marks,
  opacity=0.6,
  scatter, 
  visualization depends on={\thisrow{sizedata} \as\perpointmarksize},
  scatter/@pre marker code/.style={/tikz/mark size=\perpointmarksize},
  scatter/@post marker code/.style={}
]
table{%
x  y  sizedata
0.00123802423477173 0.04 0.7978845608028654
0.00111932516098022 0.08 1.4927053303604616
0.00142442464828491 0.07 1.9706294962377493
0.00305082321166992 0.13 2.3534213434057993
0.00468114137649536 0.05 2.682126661392722
0.00631772041320801 0.17 2.9640095915284457
0.0084230899810791 0.24 3.231186201200472
0.0109951972961426 0.28 3.477898169151024
0.0140377020835876 0.23 3.70823233942262
0.0174449944496155 0.39 3.9250730555360964
0.0230372905731201 0.44 4.122832512025492
0.0284050631523132 0.48 4.318907191682884
0.0330060195922852 0.61 4.506458780298101
0.0395018982887268 0.63 4.686510657907603
0.0458557939529419 0.67 4.85989645515594
0.0530057096481323 0.76 5.0209703231304035
0.061233479976654 0.76 5.183179949983594
0.0670866227149963 0.79 5.340464942499636
0.0783382797241211 0.81 5.493248329560876
0.088278181552887 0.88 5.641895835477563
};
\addlegendentry{$\mathrm{xHSIC}$}
\end{axis}

\end{tikzpicture}
\begin{tikzpicture}

\definecolor{darkgray176}{RGB}{176,176,176}
\definecolor{darkorange25512714}{RGB}{255,127,14}
\definecolor{lightgray204}{RGB}{204,204,204}
\definecolor{steelblue31119180}{RGB}{31,119,180}

\begin{axis}[
height=\figheight,
legend cell align={left},
legend style={
  fill opacity=0.8,
  draw opacity=1,
  text opacity=1,
  at={(0.03,0.97)},
  anchor=north west,
  draw=lightgray204
},
log basis x={10},
tick align=outside,
tick pos=left,
title={Power vs Computation},
width=\figwidth,
x grid style={darkgray176},
xlabel={Running Time (Seconds)},
xmin=0.000780557454610865, xmax=5.6544249759252,
xmode=log,
xtick style={color=black},
xtick={1e-05,0.0001,0.001,0.01,0.1,1,10,100},
xticklabels={
  \(\displaystyle {10^{-5}}\),
  \(\displaystyle {10^{-4}}\),
  \(\displaystyle {10^{-3}}\),
  \(\displaystyle {10^{-2}}\),
  \(\displaystyle {10^{-1}}\),
  \(\displaystyle {10^{0}}\),
  \(\displaystyle {10^{1}}\),
  \(\displaystyle {10^{2}}\)
},
y grid style={darkgray176},
ylabel={Power},
ymin=0.0145, ymax=1.0155,
ytick style={color=black}
]
\addplot [
  draw=darkorange25512714,
  fill=darkorange25512714,
  mark=triangle*,
  only marks,
  opacity=0.6,
  scatter, 
  visualization depends on={\thisrow{sizedata} \as\perpointmarksize},
  scatter/@pre marker code/.style={/tikz/mark size=\perpointmarksize},
  scatter/@post marker code/.style={}
]
table{%
x  y  sizedata
0.0939647960662842 0.16 1.0300645387285057
0.0745704317092896 0.19 1.6286750396763996
0.0784309768676758 0.31 2.0601290774570113
0.101192662715912 0.34 2.4157154730437167
0.271330718994141 0.35 2.744691963229459
0.388918466567993 0.36 3.020742194219052
0.508085062503815 0.52 3.2735963151943896
0.623620369434357 0.59 3.508273402369496
0.713830964565277 0.68 3.742410318509555
0.870520465373993 0.66 3.949327084834294
1.01283274650574 0.78 4.145929793656026
1.18640481948853 0.85 4.3336224206596095
1.36256204366684 0.83 4.525255353143864
1.91881767034531 0.9 4.697817093297286
2.18859865665436 0.93 4.86426097912058
2.46109400510788 0.89 5.025194951831626
2.8398472738266 0.95 5.191361770309175
3.236625187397 0.95 5.342451353184843
3.51406677961349 0.97 5.489383926458918
3.77515378236771 0.95 5.641895835477563
};
\addlegendentry{HSIC-perm}
\addplot [
  draw=steelblue31119180,
  fill=steelblue31119180,
  mark=*,
  only marks,
  opacity=0.6,
    scatter, 
  visualization depends on={\thisrow{sizedata} \as\perpointmarksize},
  scatter/@pre marker code/.style={/tikz/mark size=\perpointmarksize},
  scatter/@post marker code/.style={}
]
table{%
x  y  sizedata
0.00137717247009277 0.06 1.0300645387285057
0.00116911888122559 0.13 1.6286750396763996
0.0012628436088562 0.11 2.0601290774570113
0.00141819953918457 0.12 2.4157154730437167
0.00314977884292603 0.14 2.744691963229459
0.00480786561965942 0.08 3.020742194219052
0.00664048671722412 0.15 3.2735963151943896
0.00812045574188233 0.19 3.508273402369496
0.00832841634750366 0.31 3.742410318509555
0.00926328420639038 0.23 3.949327084834294
0.012204601764679 0.4 4.145929793656026
0.0129652190208435 0.41 4.3336224206596095
0.016937096118927 0.46 4.525255353143864
0.0201682448387146 0.6 4.697817093297286
0.0216219115257263 0.51 4.86426097912058
0.0263251233100891 0.59 5.025194951831626
0.0303431534767151 0.63 5.191361770309175
0.0321394228935242 0.68 5.342451353184843
0.0360158801078796 0.72 5.489383926458918
0.0397423028945923 0.72 5.641895835477563
};
\addlegendentry{$\mathrm{xHSIC}$}
\end{axis}

\end{tikzpicture}
            \caption{The figure shows the power versus running time curves for our cross-HSIC statistic and the HSIC permutation-test~(with $150$ permutations) on two different problems~(with $b=2$ and $\epsilon \in \{0.30, 0.35\}$. The size of the marker in the figures is proportional to the sample-size used for estimating the power.  As indicated by the figures, if sample-size is not an issue, the cross-HSIC test can achieve the same power at a significantly lower running time (or computational cost) as compared to the permutation-test. }
            \label{fig:power-compoute-hsic}
        \end{figure}
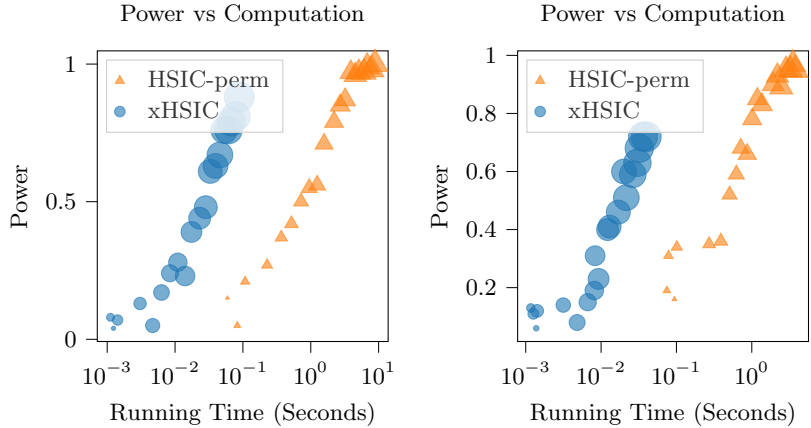
        
\section{Conclusion and Future Work}
    \label{sec:conclusion}
    In this paper, we proposed a variant of the HSIC statistic, called the cross-HSIC statistic, that is constructed using  the ideas of sample-splitting and studentization. Under very mild conditions, we showed that this statistic has a standard normal limiting null distribution. Based on this result, we proposed a simple permutation-free independence test, that rejects the null when the studentized cross-HSIC statistic exceeds $z_{1-\alpha}$, the $(1-\alpha)$-quantile of the standard normal distribution. We present a through theoretical analysis of the performance of this test, and empirically validate the theoretical predictions on some experiments with synthetic data. 

    An interesting direction for future work is extending the ideas developed in this paper to design roubust versions of the cross-HSIC test.  The cross-HSIC statistic that we proposed can be thought of as the studentized sample covariance of projected feature maps onto one-dimensional spaces. The association between these projected feature maps can be measured through other methods as well such as Kendall's rank coefficient and Spearman's rank coefficient. Given that these rank-based approaches are more robust to outliers than the sample covariance, it would be interesting to develop robust cross-HSIC and investigate its theoretical guarantees and empirical performance. 


\newpage 
\bibliography{ref}
\bibliographystyle{abbrvnat}
\newpage
\appendix 

\section[Quadratic Computational Complexity of \texorpdfstring{$\cshsic$}{xHSIC}]{Quadratic Computational Complexity of \texorpdfstring{$\cshsic$}{xHSIC}~(\texorpdfstring{\Cref{theorem:quadratic-cost}}{Theorem 2})}
\label{appendix:quadratic-time}
    We begin by introducing the necessary notation for proving this result. First let $K$ and $L$ denote $2n \times 2n$ matrices, defined as 
    \begin{align}
        [K]_{ij} =  \begin{cases}
            0 & \text{ if } 1 \leq i, j \leq n, \text{ or } n+1 \leq i, j \leq 2n \\
            k(X_i, X_j) & \text{ otherwise}. 
        \end{cases}
    \end{align}
        and 
     \begin{align}
        [L]_{ij} =  \begin{cases}
            0 & \text{ if } 1 \leq i, j \leq n, \text{ or } n+1 \leq i, j \leq 2n \\
            \ell(X_i, X_j) & \text{ otherwise}. 
        \end{cases}
    \end{align}   
    Also introduce the following two vectors in $\mathbb{R}^{2n}$. 
    \begin{align}
        \bone_u = (\underbrace{1, \ldots, 1}_{n \text{ terms}}, 0, \ldots, 0) \quad \text{and} \quad  \bone_l= (\underbrace{0, \ldots, 0}_{n \text{ terms}}, 1, \ldots, 1)
    \end{align}
    Finally, note that we will use $\circ$ to denote the elementwise product of two matrices. 

    \paragraph{Proof of~\Cref{theorem:quadratic-cost}.} To show the quadratic complexity of the cross-HSIC statistic, we it suffices to show that $T_2, T_3,$ and $T_4$ can each be computed in quadratic time.   
    
    \begin{lemma}
        \label{lemma:quadratic-T2-T3}
        $T_2$ and $T_3$ satisfy the following: 
        \begin{align}
            &T_2 = \frac{1}{n^2(n-1)} \lp \bone_l K L \bone_l - \frac{1}{2} \text{tr}\lp KL \rp \rp, \quad \text{and} \quad T_3 = \frac{1}{n^2(n-1)} \lp \bone_u K L \bone_u - \frac{1}{2} \text{tr}\lp KL \rp \rp, 
        \end{align}
        where $tr(\cdot)$ denote the trace of a matrix. 
    \end{lemma}
    \begin{proof}
        We show the details of the calculation for the term $T_2$~(the result for $T_3$ follows the exact same steps). In particular, we have the following: 
        \begin{align}
            n^2(n-1) T_2 &= \sum_{i=1}^n \sum_{n+1 \leq j_1 \neq j_2 \leq 2n}\; k(X_i, X_{j_1}) \ell(Y_i, Y_{j_2} \\ 
                &= \sum_{i=1}^n \sum_{j_1 =n+1}^{2n}\sum_{j_2 =n+1}^{2n}  k(X_i, X_{j_1}) \ell(Y_i, Y_{j_2} - \sum_{i=1}^n \sum_{j=n+1}^{2n} k(X_i, X_{j}) \ell(Y_i, Y_{j} \\
                & = \bone_l^T KL \bone_l - \frac{1}{2} tr(KL). 
        \end{align}
        Both the matrix multiplications involved in the last expression can be done in $\mc{O}(n^2)$ time, implying the quadratic complexity of $T_2$. 
    \end{proof}
    
    \begin{lemma}
        \label{lemma:quadratic-T4}
        The term $T_4$ satisfies 
        \begin{align}
            T_4 = \frac{1}{n^2(n-1)^2} \lb (\bone_l^T K \bone_u) (\bone_l^TL\bone_u)  - \bone_l^T KL \bone_l - \bone_u^T KL \bone_u + \frac{1}{2} \text{tr}(KL)\rb. 
        \end{align}
    \end{lemma}

    \begin{proof}
        We again proceed by expanding the summation defining $T_4$. 
        \begin{align}
            n^2(n-1)^2 T_4 &= \sum_{1 \leq i_1 \neq i_2 \leq n} \sum_{n+1 \leq j_1 \neq j_2 \leq 2n} k(X_{i_1}, X_{j_1}) \ell(Y_{i_2}, Y_{j_2})  \\
            & = \sum_{i_1=1}^{n} \sum_{i_2=1}^n \sum_{j_1=n+1}^{2n} \sum_{j_2=n+1}^{2n} k(X_{i_1}, X_{j_1}) \ell(Y_{i_2}, Y_{j_2}) - \sum_{i=1}^n \sum_{j_1=n+1}^{2n} \sum_{j_2=n+1}^{2n} k(X_{i}, X_{j_1}) \ell(Y_{i}, Y_{j_2}) \\
            &\; - \sum_{i_1=1}^n \sum_{i_2=1}^n \sum_{j=n+1}^{2n}  k(X_{i_1}, X_{j}) \ell(Y_{i_2}, Y_{j}) + \sum_{i=1}^n \sum_{j=n+1}^{2n} k(X_i, X_j) \ell(Y_i, Y_j).  \\
            & \defined T_{4,1} - T_{4,2} - T_{4,3} + T_{4,4}
        \end{align}
        The first term, $T_{4,1}$, can be factored into a product of two terms, each of which can be computed in quadratic time, as follows: 
        \begin{align}
            T_{4,1} &= \sum_{i_1, j_1} k(X_{i_1}, X_{j_1}) \sum_{i_2, j_2} \ell(Y_{i_2}, Y_{j_2}) 
             = \lp \bone_l^T K \bone_u \rp \lp \bone_l^T L \bone_u\rp.  
        \end{align}
        The next three terms were evaluated in~\Cref{lemma:quadratic-T2-T3}, as 
        \begin{align}
            T_{4,2} &= \bone_l^T KL \bone_l, \quad  
            T_{4,3} = \bone_u^T KL \bone_u, \quad  \text{and} \quad 
            T_{4,4} = \frac{1}{2} tr \lp KL \rp. 
        \end{align}
        Together, the above expressions imply the required result. 
    \end{proof}
    Thus the previous two lemmas imply that the statistic $\crosshsic$ can be computed in quadratic time. To establish the quadratic computational complexity of $\cshsic$, it remains to show that $s_n^2$ can also be computed in quadratic time. We now give an outline of this result, omitting some long but standard calculations. 
    
    \begin{lemma}
        \label{lemma:quadratic-sn}
        The variance term $s_n^2$ can be computed in quadratic time. 
    \end{lemma}

    \noindent \emph{Proof outline of~\Cref{lemma:quadratic-sn}.} In the first step, we expand the expression of $s_n^2$ to get the following: 
    \begin{align}
        s_n^2 &= \frac{ 4(n-1)}{(n-2)^2} \sum_{i=1}^n \lp \frac{1}{n-1} \sum_{j \neq i} h(Z_i, Z_j) - \crosshsic \rp^2  \\
        &= \frac{ 4(n-1)}{(n-2)^2} \sum_{i=1}^n \lp  \lp \frac{1}{n-1} \sum_{j=1}^{j\neq i, n} h(Z_i, Z_j) \rp^2 + \crosshsic^2 - 2 \crosshsic \lp \frac{1}{n-1} \sum_{j=1}^{j\neq i, n} h(Z_i, Z_j) \rp \rp \\
        & = \frac{4(n-1)}{(n-2)^2(n-1)^2} \sum_{i=1}^n \lp  \frac{1}{n-1} \sum_{j=1}^{j\neq i, n} h(Z_i, Z_j) \rp^2 - \frac{4n(n-1)}{(n-2)^2} \crosshsic^2. 
    \end{align}
    We have already proved that $\crosshsic$ can be computed with quadratic cost. Excluding the leading factors, the first term can written as 
    \begin{align}
        &\sum_{i=1}^n \lp  \frac{1}{n-1} \sum_{j=1}^{j\neq i, n} h(Z_i, Z_j) \rp^2 = \sum_{i=1}^n w_i^2, \quad  
        \text{where} \quad  w_i = \frac{1}{n-1} \sum_{j=1}^{j\neq i, n} h(Z_i,Z_j). 
    \end{align}
    Thus to complete the proof, it suffices to show that each $\boldsymbol{w} = (w_1, \ldots, w_n)$ can be computed in quadratic time. By expanding the terms involved in the definition of $\boldsymbol{w}$, it can be verified that (with $I_n$ denoting the $n\times n$ identity matrix, and $\bone$ denoting the all ones vector of appropriate dimension): 
    \begin{align}
        &\ww = I_n \wtilde, \quad \text{where} \\
        &2(n-1)\wtilde = n (K\circ L)\bone + \frac{1}{2} tr(KL)  \bone - (KL + LK) \bone - (K\bone_l) \circ (L \bone_l) - \frac{1}{2n} \bone_l^T KL \bone_l \bone  \\
        &\; + \frac{1}{2n} \lp (\bone^T L \bone )K\bone_l  + (\bone^T K \bone)L \bone_l \rp.  
    \end{align}
    Note that every term involved in the definition of $\wtilde$, and hence in the definition of $\ww$, can be computed in quadratic time. This in turn, implies that the variance term $s_n^2$ can also be computed in quadratic time as required. 
    \qed 
    
\section[Testing Linear Dependence]{Testing Linear Dependence~(\Cref{sec:warmup-linear})} 
\label{proof:linear-kernel}
    Both~\Cref{thm:warmup-linear-1} and~\Cref{thm:warmup-linear-2} rely on some simplified representation of the $\cshsic$ statistic for linear kernels in one dimension. We present the details of the common steps here, before moving on to the specifics of the proof of~\Cref{thm:warmup-linear-1} in~\Cref{proof:warmup-linear-1} and of~\Cref{thm:warmup-linear-2} in~\Cref{proof:warmup-linear-2}. 
    
     Since $\crosshsic$ is location-invariant, we may assume that $\mathbb{E}[X]=\mathbb{E}[Y]=0$ without loss of generality. With this assumption, $\crosshsic$ can be written as: 
        \begin{align} 
        	\crosshsic ~=~ & {f}_2 \cdot \biggl( \frac{1}{n} \sum_{i=1}^n X_iY_i - \frac{1}{n(n-1)} \sum_{1 \leq i \neq j \leq n} X_i Y_j \biggr) \\[.5em]\label{eq: xHSIC for linear kernel}
        	=~ &  {f}_2 \cdot \biggl( \frac{1}{n} \sum_{i=1}^n X_iY_i - \frac{1}{n(n-1)} \biggl\{ \biggl( \sum_{i=1}^n X_iY_i \biggr)^2 - \sum_{i=1}^n X_i^2 Y_i^2 \bigg\} \biggr). 
        \end{align}
        On the other hand, the squared denominator of the studentized statistic is
        \begin{align}
        	s_n^2 ~=~ {f}_2^2 \cdot (\mathrm{I}_n - \mathrm{II}_n),
        \end{align}
        where 
        \begin{align}
        	\mathrm{I}_n ~=~ &\frac{1}{(n-1)(n-2)^2} \sum_{i=1}^n \lp \sum_{j=1}^{n, j\neq i} (X_i-X_j)(Y_i-Y_j) \rp^2 \quad \text{and} \label{eq:I_n} \\[.5em] 
        	\mathrm{II}_n ~=~ & \frac{4n(n-1)}{(n-2)^2}  \biggl( \frac{1}{n} \sum_{i=1}^n X_iY_i - \frac{1}{n(n-1)} \biggl\{ \biggl( \sum_{i=1}^n X_iY_i \biggr)^2 - \sum_{i=1}^n X_i^2 Y_i^2 \bigg\} \biggr)^2. \label{eq:II_n}
        \end{align}
        More specifically, recall from \eqref{eq: alternative expression of s_n} that $s_n^2$ can be expressed as
        \begin{align}
            s_n^2 ~=~ \frac{4}{(n-1)(n-2)^2} \sum_{i=1}^n \lp \sum_{j=1}^{n, j\neq i} \la h(Z_i, Z_j), f_2 \ra \rp^2 - \frac{4n(n-1)}{(n-2)^2} \crosshsic^2.
        \end{align}
        Since $\la h(Z_i, Z_j), f_2 \ra = \frac{1}{2}(X_i-X_j)(Y_i-Y_j) \cdot f_2$ for the linear kernel, we see that 
        \begin{align}
            \frac{4}{(n-1)(n-2)^2} \sum_{i=1}^n \lp \sum_{j=1}^{n, j\neq i} \la h(Z_i, Z_j), f_2 \ra \rp^2 = f_2^2 \cdot \mathrm{I}_n.
        \end{align}
        Similarly, using expression~\eqref{eq: xHSIC for linear kernel}, we see that
        \begin{align}
            \frac{4n(n-1)}{(n-2)^2} \crosshsic^2 ~=~ f_2^2 \cdot \mathrm{II}_n.
        \end{align}
        Combining the above two expressions yields $s_n^2 = f_2^2 \cdot (\mathrm{I}_n - \mathrm{II}_n)$. 
        
    \subsection{Proof of\texorpdfstring{~\Cref{thm:warmup-linear-1}}{Theorem 3}}
    \label{proof:warmup-linear-1}
        \emph{Proof of part (a).} Let us write the variance of $XY$ as $v$. By the central limit theorem and the law of large numbers,
        \begin{align}
        	\frac{1}{\sqrt{n}} \sum_{i=1}^n X_iY_i  \convdist N(0, v),   \quad \frac{1}{n}\sum_{i=1}^n X_i^2 Y_i^2 \convprob v,\quad \text{and}\quad \frac{1}{n} \sum_{i=1}^n X_i Y_i \convprob \mathbb{E}[X]\mathbb{E}[Y] = 0. 
        \end{align}
        As a result, we have the following: 
        \begin{align}
            \mathrm{II}_n = \frac{4n(n-1)}{(n-2)^2} \lp o_P(1) + \frac{1}{n-1} O_P(1) \rp^2 = o_P(1).
        \end{align}
        In other words, $\mathrm{II}_n \convprob 0$. To control $\mathrm{I}_n$, note that 
        \begin{align}
        	& \sum_{i=1}^n \lp \sum_{j=1}^{n, j\neq i} (X_i-X_j)(Y_i-Y_j) \rp^2 \\[.5em] 
        	= ~ & \sum_{1 \leq i \neq j \leq n} (X_i-X_j)^2(Y_i-Y_j)^2 + \sum_{\substack{1\leq i,j,q \leq n \\ \text{$i,j,q$ distinct}}} (X_i-X_j)(Y_i-Y_j) (X_i-X_q)(Y_i-Y_q).
        \end{align}
        Using the law of large numbers for U-statistics~\citep{hoeffding1961strong}, we have
        \begin{align}
        	A_n \defined & \frac{1}{(n-1)(n-2)^2}  \sum_{1 \leq i \neq j \leq n} (X_i-X_j)^2(Y_i-Y_j)^2 \convprob 0 \quad \text{and} \label{eq:linear-0-1}\\[.5em]
        	B_n \defined & \frac{1}{(n-1)(n-2)^2}  \sum_{\substack{1\leq i,j,q \leq n \\ \text{$i,j,q$ distinct}}} (X_i-X_j)(Y_i-Y_j) (X_i-X_q)(Y_i-Y_q) \convprob  v, \label{eq:linear-0-2}
        \end{align}
        under the assumption that $\mathbb{E}[X^2] < \infty$ and $\mathbb{E}[Y^2] < \infty$. This further implies that $\mathrm{I}_n - \mathrm{II}_n \convprob v$. We also note that 
        \begin{align}
            \frac{1}{\sqrt{n}(n-1)} \sum_{1 \leq i \neq j \leq n} X_i Y_j = \frac{n}{\sqrt{n}(n-1)} \underbrace{\biggl(\frac{1}{\sqrt{n}} \sum_{i=1}^n X_iY_i \biggr)^2}_{O_P(1)} - \frac{n}{\sqrt{n}(n-1)} \cdot \underbrace{\frac{1}{n} \sum_{i=1}^n X_i^2 Y_i^2}_{O_P(1)} = o_P(1), 
        \end{align}
        and thus Slutsky's theorem yields
        \begin{align}
            \sqrt{n} \biggl( \frac{1}{n} \sum_{i=1}^n X_iY_i - \frac{1}{n(n-1)} \sum_{1 \leq i \neq j \leq n} X_i Y_j \biggr) \convdist N(0,v).
        \end{align}
        This approximation combined with the continuous mapping theorem shows that $\mathrm{sign}({f}_2) \convdist \mathrm{Rademacher}$, taking either $-1$ or $1$ with an equal probability. Combining the above pieces, we have under the finite second moment condition of $X$ and $Y$
        \begin{align}
        	\frac{\sqrt{n}\crosshsic}{s_n} = \mathrm{sign}({f}_2) \times  \frac{\sqrt{n} \biggl( \frac{1}{n} \sum_{i=1}^n X_iY_i - \frac{1}{n(n-1)} \sum_{1 \leq i \neq j \leq n} X_i Y_j \biggr)}{\sqrt{\mathrm{I}_n - \mathrm{II}_n}} \convdist N(0,1),
        \end{align}
        where we use the fact that ${f}_2$ is independent of the first half of the data. This shows that the (pointwise) asymptotic normality holds as long as  $0 < \mathbb{E}[X^2] < \infty$ and $0 < \mathbb{E}[Y^2] < \infty$.  \hfill \qed 
        
        \emph{Proof of part~(b).} When the alternative is true, and $\mathbb{E}[XY] = \rho \neq 0$, there are two main differences: 
        \begin{align}
            \frac{1}{n} \sum_{i=1}^n X_i Y_i \convprob \rho,  \quad \text{and} \quad 
            B_n \convprob \mathbb{E}[(XY)^2] + 3 \rho^2 = \var(XY) + 4\rho^2, 
        \end{align}
        due to the law of large numbers for \iid sums and U-statistics respectively.  As a result, we have the following: 
        \begin{align}
            \mathrm{II}_n \convprob  4\rho^2, \quad \text{and} \quad \mathrm{I}_n \convprob \var(XY) + 4\rho^2. 
        \end{align}
        Together, these two results imply that $(\mathrm{I}_n - \mathrm{II}_n) \convprob \var(XY)$. 
        
        The power of the cross-HSIC test is 
        \begin{align}
            \mathbb{E}[\Psi] &= \mathbb{P} \lp \text{sign}(f_2) \frac{  \frac{1}{\sqrt{n}}  \sum_{i=1}^n X_i Y_i - \frac{1}{\sqrt{n}} o_P(1)}{\sqrt{\mathrm{I}_n - \mathrm{II}_n}}  > z_{1-\alpha}\rp. 
        \end{align}
        Let $Z$ denote a $N(0,1)$ random variable. Then, we have 
        \begin{align}
            \lim_{n \to \infty} \mathbb{E}[\Psi] = \lim_{n \to \infty} \Bigg\{  \mathbb{P}\lp f_2>0 \rp \mathbb{P}\lp Z > z_{1-\alpha} -  \frac{\rho\sqrt{n}}{\sqrt{\var(XY)}} \rp +  \mathbb{P}\lp f_2 \leq 0 \rp \mathbb{P}\lp Z < -z_{1-\alpha} - \frac{\rho \sqrt{n}}{\sqrt{\var(XY)}} \rp + o(1) \Bigg\}. 
        \end{align}
        To conclude the proof, we note that $\frac{\sqrt{n}}{n(n-1)} \sum_{n+1 \leq i \neq j \leq 2n} X_iY_j = o_P(1)$ and thus by the central limit theorem and Slutsky's theorem,
        \begin{align}
            \lim_{n \rightarrow \infty} \mathbb{P}(f_2 > 0) = \lim_{n \rightarrow \infty} \mathbb{P} \biggl( \frac{1}{\sqrt{n}} \sum_{i=n+1}^{2n} X_iY_i - \rho + o_P(1) > \sqrt{n} \rho \biggr) = 1 - \lim_{n \rightarrow \infty} \mathbb{P}(Z > \sqrt{n} \rho).
        \end{align}
        This implies that $\lim_{n \to \infty} \mathbb{P}(f_2>0) = \boldsymbol{1}_{\rho>0}$. Similarly, we have $\lim_{n \to \infty} \mathbb{P}(f_2\leq 0) = \boldsymbol{1}_{\rho<0}$ (note that $\lim_{n \to \infty} \mathbb{P}(f_2 = 0) = 0$). In either case, we have $\lim_{n \to \infty} \mathbb{E}[\Psi]=1$ under the alternative, as required.
        
    \subsection{Proof of~\texorpdfstring{\Cref{thm:warmup-linear-2}}{Theorem 4}}
    \label{proof:warmup-linear-2}
        \subsubsection{Proof of part (a)} 
        To prove that the distribution of $\cshsic$ converges to $N(0,1)$ uniformly over the null class, we first show that that asymptotic normality holds for an arbitrary sequence of distributions $\{P_{XY,n} \in \nullclass: n \geq 1\}$.

         We start by analyzing the terms $\mathrm{I}_n$ and $\mathrm{II}_n$ first stated in~\eqref{eq:I_n} and~\eqref{eq:II_n} respectively. Introduce the variance  $v_n = \var_{P_{XY,n}}(XY)$, and note that
        \begin{align}
            \frac{v_n}{n-1} \lp \frac{\sum_{i=1}^n X_i Y_i}{\sqrt{n v_n}}  \rp^2 \convprob 0. \label{eq:linear-2-1}
        \end{align}
        This is because  $\frac{v_n}{n-1} \lesssim \frac{\mathbb{E}_{P_{XY,n}}[(XY)^4]}{n} \to 0$,  and  $\frac{\sum_{i=1}^n X_i Y_i}{\sqrt{n v_n}}   \convdist N(0,1)$  under the fourth moment assumption.  Next, by the weak law of large numbers for triangular arrays~\citep[Theorem 2.2.6]{durrett2019probability}, we also have 
        \begin{align}
             \frac{1}{n} \sum_{i=1}^n X_i^2Y_i^2  &=  \frac{1}{n} \sum_{i=1}^n X_i^2Y_i^2 - v_i + \frac{1}{n}\sum_{i=1}^n v_i  \leq o_P(1) + \max_{1 \leq i \leq n} v_i. 
        \end{align}
        Since $v_n \lesssim \mathbb{E}_{P_{XY,n}}[(XY)^4]$, the above expression implies 
        \begin{align}
            \lim_{n \to \infty} \frac{1}{n(n-1)} \sum_{i=1}^n X_i^2 Y_i^2 \lesssim  \lim_{n \to \infty}\lp \frac{ o_P(1)}{n}  + \frac{\mathbb{E}_{P_{XY,n}}[(XY)^4]}{n} \rp \stackrel{p}{=} 0. \label{eq:linear-2-2}
        \end{align}
        Finally, we consider the remaining term, $\frac{1}{n} \sum_{i=1}^n X_i Y_i$ of $\mathrm{II}_n$,  and observe the following: 
        \begin{align}
            \frac{1}{n} \sum_{i=1}^n X_i Y_i = \sqrt{\frac{v_n}{n}} \times \lp \frac{1}{\sqrt{n v_n}} \sum_{i=1}^n X_i Y_i \rp = \sqrt{\frac{v_n}{n}} \times O_P(1) \convprob 0,   \label{eq:linear-2-3}
        \end{align}
        where we used the fact that $\frac{1}{\sqrt{n v_n}} \sum_{i=1}^n X_i Y_i \convdist N(0,1)$  and $\frac{v_n}{n} \to 0$ under the assumptions of~\Cref{thm:warmup-linear-2}. 
         Combining the results of~\eqref{eq:linear-2-1}, \eqref{eq:linear-2-2} and~\eqref{eq:linear-2-3}, we conclude that $\mathrm{II}_n \convprob 0$, for an arbitrary sequence of $P_{XY,n}$ chosen from $\nullclass$. 
         
         Next, we look at the term $\mathrm{I}_n$, and note that the term $A_n$ introduced in~\eqref{eq:linear-0-1}, satisfies
         \begin{align}
             \mathbb{E}[A_n/v_n] &= \frac{1}{n(n-1)^2v_n} \mathbb{E}\lb  \sum_{i=1}^n \sum_{j\neq i} (X_i-X_j)^2 (Y_i-Y_j)^2 \rb = \frac{1}{(n-1)} \to 0. 
         \end{align}
         Since $A_n/v_n$ is a non-negative random variable for all $n \geq 1$, this implies that $A_n/v_n \convprob 0$, by using Markov's inequality. We now consider the final term $B_n$, introduced in~\eqref{eq:linear-0-2}, and obtain the following result: 
        \begin{lemma}
            \label{lemma:linear-kernel-1}
            Under the assumptions of~\Cref{thm:warmup-linear-2}, we have 
            \begin{align}
                \frac{B_n}{v_n} \convprob 1. \label{eq:linear-2-4}
            \end{align}
        \end{lemma}
        The statement is proved in~\Cref{proof:lemma-linear-kernel-1} at the end of this section.  As a consequence of the above result, we have $(\mathrm{I}_n - \mathrm{II}_n)/v_n \convprob 1$. 
         
         Combining all the above statements, we get that 
        \begin{align}
            \frac{\sqrt{n}\, \crosshsic}{s_n}  = \text{sign}(f_2) \sqrt{\frac{v_n}{\mathrm{I}_n - \mathrm{II}_n}} \; \lp  \frac{1}{\sqrt{nv_n}} \sum_{i=1}^n X_i Y_i - \frac{1}{(n-1) \sqrt{nv_n}} \sum_{i \neq j} X_i Y_j  \rp \convdist N(0,1). 
        \end{align}
        Here we have used the following facts: 
        \begin{itemize}
            \item $\text{sign}(f_2)$ is a Rademacher random variable due to symmetry of $f_2$ under the null.  
            \item $(\mathrm{I}_n - \mathrm{II}_n)/v_n \convprob 1$, due to ~\eqref{eq:linear-2-4}. 
            \item $D_n \defined  \frac{1}{(n-1)\sqrt{nv_n}} \sum_{i \neq j} X_i Y_j \convprob 0$. This is because $\mathbb{E}[D_n] = 0$, and $\var(D_n) \lesssim  1/n \to 0$. 
            \item $\frac{1}{\sqrt{n v_n}} \sum_{i=1}^n X_i Y_i \convdist N(0,1)$ by the (fourth moment) Lyapunov version of CLT. 
        \end{itemize}
        The final result follows by two applications of Slutsky's inequality, and the fact that $f_2$ is independent of the other terms. 
        
        We have thus proved that for any sequence of distributions $\{P_{XY,n} \in \nullclass: n \geq 1\}$, the statistic $\cshsic \convdist N(0,1)$. In other words, for this arbitrary sequence, we have $\lim_{n \to \infty} \mathbb{E}_{P_{XY,n}}\lb \Psi \rb \leq \alpha$. 
        To make this result hold uniformly over the null class of distributions, we fix an arbitrary $\epsilon>0$, and for every $n \geq 1$, select $P_{XY,n}$ such that $\mathbb{E}_{P_{XY,n}}[\Psi] \geq \sup_{P \in \nullclass}\mathbb{E}_P[\Psi] - \epsilon$. Thus, we have 
        \begin{align}
            \lim_{n \to \infty} \sup_{P \in \nullclass} \mathbb{E}_P[\Psi] \leq \lim_{n \to \infty} \mathbb{E}_{P_{XY,n}}[\Psi] + \epsilon \leq \alpha + \epsilon. 
        \end{align}
        Since $\epsilon>0$ is arbitrary, the result follows. 

        \subsubsection{Proof of part (a) with finite $2+\delta$ moments} 
            Before proceeding to the proof of part~(b) of~\Cref{thm:warmup-linear-2}, we first note that the uniform asymptotic normality of the $\cshsic$ statistic with linear kernels can in fact be obtained under weaker moment assumptions as well. 
       	For a fixed $\delta \in (0,2]$, consider a class of distributions
       	\begin{align}
       		\mathcal{P}_{n,\delta} = \bigg\{ P_{XY,n} :  \frac{\mathbb{E}\bigl[ |X - \mathbb{E}[X] |^{2 +\delta} \bigr]}{n^{\delta/4} v_X^{1+\delta/2}} = o(1) \ \text{and} \  \frac{\mathbb{E}\bigl[ |Y - \mathbb{E}[Y] |^{2+ \delta} \bigr]}{n^{\delta/4}v_Y^{1+\delta/2}} = o(1) \bigg\},
       	\end{align}
       	where $v_X = \var(X)$ and $v_Y = \var(Y)$. In this part of the proof, we show that $\cshsic$ is asymptotically $N(0,1)$ uniformly over the class $\mathcal{P}_{n,\delta}^{(0)}$ given as
       	\begin{align}
       		\mathcal{P}_{n,\delta}^{(0)} = \bigl\{ P_{XY,n} \in \mathcal{P}_{n,\delta} : P_{XY,n} = P_{X,n} \times P_{Y,n} \big\}.
       	\end{align}
       	Notice that any $P_{XY,n} \in \mathcal{P}_{n,\delta}^{(0)}$ satisfies 
       	\begin{align} \label{eq: lyapunov condition}
       		 \frac{\mathbb{E}\bigl[|X-\mathbb{E}(X)|^{2+\delta} \cdot |Y-\mathbb{E}(Y)|^{2+\delta}\bigr]}{n^{\frac{\delta}{2}} \big\{ \var\bigl[\bigl(X-\mathbb{E}(X)\bigr)\bigl(Y-\mathbb{E}(Y)\bigr)\bigr] \big\}^{2+\delta}} = o(1).
       	\end{align} 
		Since $\cshsic$ is both location-invariant and scale-invariant, we may assume that $\mathbb{E}[X] = \mathbb{E}[Y] = 0$ and $v_X = v_Y = 1$ without loss generality. We therefore assume that $X$ and $Y$ are standardized throughout this proof. 
				
		For a triangular array of random variables with a distribution $P_n \in \mathcal{P}_{n,\delta}^{(0)}$, the Lyapunov central limit theorem shows that 
        \begin{align}
        	\frac{1}{\sqrt{n}} \sum_{i=1}^n X_iY_i \convdist N(0,1). 
        \end{align}
   		We also note that $\mathbb{E}[|X^2Y^2|^{1 + \delta/2}] n^{-\delta/2} = o(1)$ by condition~\eqref{eq: lyapunov condition}. Thus Lemma 17 of \cite{lundborg2022projected} yields
   		\begin{align}
   			& \frac{1}{n} \sum_{i=1}^n X_i^2 Y_i^2 - 1 = o_{P_n}(1),
   		\end{align}
    	which implies that $\frac{1}{n^{3/2}} \sum_{i=1}^n X_i^2 Y_i^2 = o_{P_n}(1)$. Therefore by Slutsky's theorem, we have
        \begin{align}
        	& \sqrt{n} \biggl( \frac{1}{n} \sum_{i=1}^n X_iY_i - \frac{1}{n(n-1)} \sum_{1 \leq i \neq j \leq n} X_i Y_j \biggr) \\[.5em]
        	 = ~ & \frac{1}{\sqrt{n}} \sum_{i=1}^n X_iY_i - \frac{1}{\sqrt{n}(n-1)} \biggl\{ \biggl( \sum_{i=1}^n X_iY_i \biggr)^2 - \sum_{i=1}^n X_i^2 Y_i^2 \bigg\}  \convdist N(0,1).
        \end{align}
		Following the same logic, we also have $\sqrt{n} f_2 \convdist N(0,1)$ and so $\mathrm{sign}(f_2) \convdist \mathrm{Rademacher}$. 
		
		Now let us turn to the terms $\mathrm{I}_n$ and $\mathrm{II}_n$ where $\mathrm{I}_n$ and $\mathrm{II}_n$ are recalled in \eqref{eq:I_n} and \eqref{eq:II_n}, respectively. Note that the second term satisfies
		\begin{align}
			\mathrm{II}_n = O(1) \cdot f_1^2.
		\end{align}
		We already showed that $f_1^2 = O_{P_n}(n^{-1})$ and thus $\mathrm{II}_n = o_{P_n}(1)$. For the first term, note that 
		\begin{align}
			\mathrm{I}_n ~=~ & \{1 + o(1)\} \cdot \frac{1}{n^3} \sum_{i=1}^n \sum_{j=1}^n \sum_{k=1}^n (X_i -X_j)(Y_i-Y_j) (X_i -X_k)(Y_i-Y_k) \\[.5em]
			= ~ & \{1 + o(1)\} \cdot \Biggl\{\frac{1}{n} \sum_{i=1}^n X_i^2 Y_i^2 + O(1) \biggl( \frac{1}{n} \sum_{i=1}^n X_i^2Y_i \biggr) \biggl( \frac{1}{n} \sum_{j=1}^n Y_j  \biggr)+ O(1)  \biggl( \frac{1}{n} \sum_{i=1}^n X_iY_i^2 \biggr) \biggl( \frac{1}{n} \sum_{j=1}^n X_j  \biggr)  \\[.5em]
			& \hskip 6em + O(1) \biggl( \frac{1}{n} \sum_{i=1}^n X_i^2 \biggr) \biggl( \frac{1}{n} \sum_{j=1}^n Y_j \biggr)^2 + O(1) \biggl( \frac{1}{n} \sum_{i=1}^n Y_i^2 \biggr) \biggl( \frac{1}{n} \sum_{j=1}^n X_j \biggr)^2  \Biggr\}. 
		\end{align}
		It can be checked that $\mathbb{E}[|X|^{1+\delta/2}] n^{-\delta/2} = o(1)$, $\mathbb{E}[|Y|^{1+\delta/2}] n^{-\delta/2} = o(1)$, $\mathbb{E}[|X^2 Y|^{1+\delta/2}] n^{-\delta/2} = o(1)$ and $\mathbb{E}[|XY^2|^{1+\delta/2}] n^{-\delta/2} = o(1)$ for $P_{XY,n} \in \mathcal{P}_{n,\delta}^{(0)}$. For instance, by Cauchy--Schwarz inequality, we see that 
		\begin{align}
			\mathbb{E}[|X^2 Y|^{1+\delta/2}]  n^{-\delta/2} \leq \sqrt{\mathbb{E}[|X^2 Y^2|^{1+\delta/2}] n^{-\delta/2}} \sqrt{\mathbb{E}[|X^2|^{1+\delta/2}] n^{-\delta/2}} = o(1).
		\end{align}
		Thus Lemma 17 of \cite{lundborg2022projected} yields
		\begin{align}
			& \frac{1}{n} \sum_{i=1}^n X_i^2 Y_i = o_{P_n}(1), \  \frac{1}{n} \sum_{i=1}^n X_i Y_i^2 = o_{P_n}(1), \\
			& \frac{1}{n} \sum_{i=1}^n X_i= o_{P_n}(1) \; \ \text{and} \; \ \frac{1}{n} \sum_{i=1}^n Y_i = o_{P_n}(1).
		\end{align}
		These approximations verify that $\mathrm{I}_n = 1 + o_{P_n}(1)$. Therefore
		\begin{align}
			\cshsic = \mathrm{sign}(f_2)\frac{\frac{1}{\sqrt{n}}\sum_{i=1}^n X_iY_i + o_{P_n}(1)}{1 + o_{P_n}(1)} \convdist N(0,1),
		\end{align} 
		by Slutsky's theorem. 
		
		We have thus proved that for any sequence of distributions $\{P_{XY,n} \in \nullclass: n \geq 1\}$, the statistic $\cshsic \convdist N(0,1)$. In other words, for this arbitrary sequence, we have $\lim_{n \to \infty} \mathbb{E}_{P_{XY,n}}\lb \Psi \rb \leq \alpha$. 
		To make this result hold uniformly over the null class of distributions, we fix an arbitrary $\epsilon>0$, and for every $n \geq 1$, select $P_{XY,n}$ such that $\mathbb{E}_{P_{XY,n}}[\Psi] \geq \sup_{P \in \nullclass}\mathbb{E}_P[\Psi] - \epsilon$. Thus, we have 
		\begin{align}
			\lim_{n \to \infty} \sup_{P \in \nullclass} \mathbb{E}_P[\Psi] \leq \lim_{n \to \infty} \mathbb{E}_{P_{XY,n}}[\Psi] + \epsilon \leq \alpha + \epsilon. 
		\end{align}
		Since $\epsilon>0$ is arbitrary, the result follows.

        \subsubsection{Proof of part (b)}
            We now study the power of the cross-HSIC test against a sequence of local alternatives $\{P_{XY,n} \in \altclass: n \geq 1\}$. Our first result considers the two components of the  term $\mathrm{I}_n$. 
            \begin{lemma}
                \label{lemma:linear-kernel-2} 
                Recall that $A_n$ and $B_n$ were defined in~\eqref{eq:linear-0-1} and~\eqref{eq:linear-0-2} respectively, and $\rho_n = \mathbb{E}_{P_{XY,n}}[XY]$. Then we have the following under the alternative: 
                \begin{align}
                    A_n \convprob 0, \quad  \text{and} \quad  B_n  - \lp 1 + 4\rho_n^2 \rp \convprob  0. 
                \end{align}
            \end{lemma}
            
            As a consequence of the above lemma, we can conclude that  
            \begin{align}
                \mathrm{I}_n - \lp 1 + 4\mathbb{E}_{P_{XY,n}}[XY]^2 \rp \convprob 0. 
            \end{align}
            
            Next, we consider the second component, $\mathrm{II}_n$, of $s_n^2$. 
            \begin{lemma}
                \label{lemma:linear-kernel-3}
                Under the alternative, we have  $\mathrm{II}_n - 4\rho_n^2 \convprob 0$. 
            \end{lemma}
            
            Combining the results of the two lemmas above, we get that 
            \begin{align}
                \mathrm{I}_n  - \mathrm{II}_n  \convprob 1. 
            \end{align}
            
            Next, we look at the second term in the definition of $\crosshsic$, and note that 
            \begin{align}
                \mathbb{E} \lb \bigg( \frac{1}{n(n-1)} \sum_{i \neq j } X_i Y_j \bigg)^2 \rb \lesssim \frac{1}{n^2}  \lp \mathbb{E}[X^2Y^2]  + \mathbb{E}[XY]^2 \rp \to 0. 
            \end{align}
            This implies that $\frac{1}{n(n-1)}\sum_{i \neq j} X_iY_j \convprob 0$. Using these we have the following: 
            \begin{align}
                \frac{ \sqrt{n} \crosshsic }{s_n} = \frac{1}{\sqrt{\mathrm{I}_n - \mathrm{II}_n}} \times \text{sign}(f_2) \lp \frac{1}{\sqrt{n}} \sum_{i=1}^n X_i Y_i + o_P(1) \rp. 
            \end{align}
            Thus, the power of the cross-HSIC test can be written as 
            \begin{align}
                \mathbb{P}\lp \cshsic > z_{1-\alpha} \rp = \mathbb{P} \lp f_2>0 \rp \times \mathbb{P}\lp \sum_{i=1}^n \frac{X_i Y_i}{\sqrt{n}} > z_{1-\alpha} \rp  + \mathbb{P} \lp f_2<0 \rp \times \mathbb{P}\lp \sum_{i=1}^n \frac{X_i Y_i}{\sqrt{n}} < -z_{1-\alpha} \rp + o_P(1)  \label{eq:linear-3-1}
            \end{align}
             We now look at all the terms involved in~\eqref{eq:linear-3-1}. 
            Let us assume that $\mathbb{E}_{P_{XY,n}}[XY] = c_n/\sqrt{n}$ with $\lim_{n \to \infty}c_n = c$, where $c \in \mathbb{R} \cup\{\infty\}$. 
            \begin{align}
                \mathbb{P} \lp f_2 > 0 \rp = \mathbb{P} \lp \frac{1}{\sqrt{n}} \sum_{t=n+1}^{2n} X_t Y_t - \mathbb{E}[X_tY_t] > - \mathbb{E}[XY] \sqrt{n} \rp \to 1- \Phi(-c)= \Phi(c). 
            \end{align}
            Similarly, the other term, $\mathbb{P}(f_2\leq 0) \to \Phi(-c)$. Next,  we note that  
            \begin{align}
                \mathbb{P}\lp \sum_{t=1}^n \frac{ X_i Y_i - \mathbb{E}_{P_{XY, n}}[XY]}{\sqrt{n}} > z_{1-\alpha} - \mathbb{E}_{P_{XY, n}}[XY] \sqrt{n}\rp  \to \Phi(c - z_{1-\alpha}). 
            \end{align}
            Proceeding similarly with the last term, we obtain 
            \begin{align}
                \mathbb{P}\lp \sum_{t=1}^n \frac{ X_i Y_i - \mathbb{E}_{P_{XY, n}}[XY]}{\sqrt{n}} < -z_{1-\alpha} - \mathbb{E}_{P_{XY, n}}[XY] \sqrt{n}\rp  \to \Phi(-c - z_{1-\alpha}), 
            \end{align}
            which implies that 
            \begin{align}
                \lim_{n \to \infty} \mathbb{P}\lp \cshsic > z_{1-\alpha} \rp  = \Phi(c) \Phi(-z_{1-\alpha} + c) + \Phi(-c) \Phi(-z_{1-\alpha} - c).  
            \end{align}
        \subsection{Details of \texorpdfstring{\Cref{exam: Poisson}}{Example 5}} \label{appendix: poisson example}
            Based on our previous discussion, $\cshsic$ with linear kernels can be expressed as
            \begin{align}
            	\cshsic = \mathrm{sign}(f_2) \frac{\sqrt{n} f_1}{\sqrt{\mathrm{I}_n - \mathrm{II}_n}},
            \end{align}
            where $\mathrm{I}_n$ and $\mathrm{II}_n$ can be found in \eqref{eq:I_n} and \eqref{eq:II_n}, respectively. Moreover, since $\mathrm{I}_n - \mathrm{II}_n \geq 0$, we have
            \begin{align}
            	\mathbb{P}(\cshsic \leq 0) = \mathbb{P}\bigl(\mathrm{sign}(f_2) \cdot n f_1 \leq 0\bigr).
            \end{align} 
            To approximate this probability to the target probability in \Cref{exam: Poisson}, note that 
            \begin{align}
            	nf_1 =  \sum_{i=1}^n (X_i - p_n)(Y_i-p_n) - \frac{1}{n-1} \sum_{1 \leq i \neq j \leq n} (X_i - p_n) (Y_j - p_n) 
            \end{align}
        	by the location-invariance property. We also note that since 
        	\begin{align}
        		\var \biggl( \frac{1}{n-1} \sum_{1 \leq i \neq j \leq n} (X_i - p_n) (Y_j - p_n)  \biggr) = \frac{n}{n-1} p_n^2(1-p_n)^2 \rightarrow 0,
        	\end{align}
        	an application of Chebyshev's inequality ensures that 
        	\begin{align}
        		nf_1 = \sum_{i=1}^n (X_i - p_n)(Y_i-p_n) + o_P(1). 
        	\end{align}
        	A similar argument shows that $p_n \sum_{i=1}^n (X_i - p_n) =  o_P(1)$ and $p_n \sum_{i=1}^n (Y_i - p_n) =  o_P(1)$. Therefore
        	\begin{align}
        		nf_1 = \sum_{i=1}^n X_iY_i - np_n^2+ o_P(1). 
        	\end{align}
        	Since $X_iY_i$s are \iid Bernoulli random variables with parameter $p_n^2$ satisfying $np_n^2 = \lambda>0$, Poisson limit theorem \citep[Theorem 3.6.1 of][]{durrett2019probability} together with Slutsky's theorem yields
        	\begin{align}
        		nf_1 \convdist \mathrm{Poisson}(\lambda) - \lambda. 
        	\end{align}
        	Similarly it also follows that $\mathrm{sign}(f_2) = \mathrm{sign}(n f_2) \convdist \mathrm{sign}(\mathrm{Poisson}(\lambda) - \lambda)$ by the continuous mapping theorem. Strictly speaking, the sign function is discontinuous at $x=0$, and thus the continuous mapping theorem does not directly apply. However the same conclusion follows since $\mathbb{P}(\mathrm{Poisson}(\lambda) - \lambda = 0) = 0$, i.e., the set of discontinuity points has zero probability, as $\lambda$ is a non-integer by our assumption. Combining the pieces and observing that $f_1$ and $f_2$ are independent, we therefore conclude that $\mathrm{sign}(f_2) \cdot n f_1 \convdist \mathrm{sign}(V') \times V$ where $V,V'$ are i.i.d.~centered Poisson random variables. Lastly, the distribution of $\mathrm{sign}(V') \times V$ is continuous at $0$ since $\lambda$ is not an integer. Hence the definition of convergence in distribution implies the desired result.

        \subsection{Proof of Auxiliary Lemmas}        
            \subsubsection{Proof of~\texorpdfstring{\Cref{lemma:linear-kernel-1}}{Lemma~21}}
            \label{proof:lemma-linear-kernel-1}
            To prove this result, we first note that $\mathbb{E}[(X_i-X_j)(X_i-X_q)(Y_i-Y_j)(Y_i-Y_q)] = \mathbb{E}[X_i^2 Y_i^2] = v_n$, under the null for any $n \geq 1$. Hence, $\mathbb{E}[B_n/v_n]=1$. To complete the proof, we will show that the second moment of $B_n/v_n-1$ converges to $0$ with $n$, which in turn implies by an application of Chebychev's inequality that $B_n/v_n-1 \convprob 0$ as required. 
            Using the notation $h_{ij}$ to denote $(X_i-X_j)(Y_i-Y_j)$, we have 
            \begin{align}
                \mathbb{E}\lb \bigg( \frac{B_n}{v_n} - 1 \bigg)^2\rb \lesssim \frac{1}{n^6} \sum_{\substack{1\leq i,j,q \leq n \\ \text{$i,j,q$ distinct}}} \sum_{\substack{1\leq t,u,r \leq n \\ \text{$t,u,r$ distinct}}} 
                \mathbb{E}\lb \lp \frac{h_{ij} h_{iq}}{v_n} -1 \rp \lp  \frac{ h_{tu} h_{tr}}{v_n} -1 \rp \rb.  \label{eq:lemma-linear-1-1}
            \end{align}
            First consider the case where all six terms, $i,j,q,t,u,r$ are distinct --- there are $\Omega(n^6)$ such terms, and in all these instances the summand in~\eqref{eq:lemma-linear-1-1} is equal to zero, due to independence of the two factors. For all other $\mc{O}(n^5)$ terms with at two common indices, we can upper bound the summand with $\mc{O}\lp (\mathbb{E}[(XY)^4] + v_n^2)/v_n^2\rp$ by using Cauchy-Schwarz inequality. Thus, overall, we get 
            \begin{align}
                \mathbb{E}\lb \bigg( \frac{B_n}{v_n} - 1 \bigg)^2\rb \lesssim \frac{ \mathbb{E}[(XY)^4] + v_n^2}{v_n^2 n} \to 0. 
            \end{align}
            This is sufficient to conclude that $B_n/v_n - 1 \convprob 0$, as required.  
        
        \subsubsection{Proof of~\texorpdfstring{\Cref{lemma:linear-kernel-2}}{Lemma 22}}        
        \label{proof:lemma-linear-kernel-2}
            Since $A_n$ is a nonnegative random variable, to show that it converges in probability to $0$, it suffices to show that $\mathbb{E}[A_n] \to 0$. 
            \begin{align}
                \mathbb{E}[A_n] = \frac{1}{n(n-1)^2} \sum_{i \neq j}\mathbb{E}[(X_i-X_j)^2 (Y_i-Y_j)^2] \lesssim \frac{1}{n^3} \lp n^2 \mathbb{E}[X^2 Y^2] \rp \leq \frac{1}{\sqrt{n}}\sqrt{ \frac{ \mathbb{E}[(XY)^4]}{n}} \to 0. 
            \end{align}
            This completes the first part of the proof. The second statement, $B_n - (1+4\rho_n^2) \convprob 0$, is obtained by following the exact same steps used in proving~\Cref{lemma:linear-kernel-1} in ~\Cref{proof:lemma-linear-kernel-1}, and we omit the details.  
        
        \subsubsection{Proof of~\texorpdfstring{\Cref{lemma:linear-kernel-3}}{Lemma 23}}
        \label{proof:lemma-linear-kernel-3}
            Recall that the term $\mathrm{II}_n$ can be written as 
            \begin{align}
                \mathrm{II}_n &\asymp 4 \lp \frac{1}{n} \sum_{i=1}^n X_i Y_i - \frac{1}{(n-1)}\lp \lp \sum_{i=1}^n \frac{X_iY_i}{\sqrt{n}} \rp^2 -  \frac{1}{n} \sum_{i=1}^n X_i^2Y_i^2 \rp \rp^2  \\
                & \defined 4 \lp E_n - \frac{1}{n-1} \lp F_n^2 - G_n \rp \rp^2. 
            \end{align}
            
            We now look at the three terms $E_n, F_n$ and $G_n$ separately. First note that 
            \begin{align}
                F_n \convdist N(0,1), 
            \end{align}
            by an application of Lyapunov CLT, since $\frac{\mathbb{E}[X_i^4Y_i^4]}{n \var(X_iY_i)} = \frac{\mathbb{E}[X_i^4Y_i^4]}{n } \to 0$, by assumption. This implies that $F_n^2$ is a $O_P(1)$ term, and thus $\frac{1}{n-1} F_n^2 \convprob 0$. 
            
            Next, note that $\var(G_n) \lesssim \frac{\mathbb{E}[(XY)^4]}{n} \to 0$, which implies that $G_n - \mathbb{E}[G_n] \convprob 0$. Thus we have 
            \begin{align}
                \frac{1}{n-1} G_n = \frac{1}{n-1} \lp o_P(1) + \mathbb{E}[G_n] \rp = o_P(1) + \frac{\mathbb{E}[X^2Y^2]}{n-1} = o_P(1). 
            \end{align}
            Thus, we have proved that $\frac{1}{n-1} (F_n^2 - G_n) \convprob 0$, which implies that $\mathrm{II}_n \asymp 4 E_n^2 + o_P(1)$. Thus, we have the following: 
            \begin{align}
                \mathrm{II}_n \asymp 4 \lp E_n - \rho_n\rp^2 + 4\rho_n^2 + o_P(1), \quad \text{or} \quad \mathrm{II}_n - 4\rho_n^2 \asymp 4(E_n-\rho_n)^2 + o_P(1). 
            \end{align}
            Finally, we note that $\mathbb{E}[E_n - \rho_n] = 0$ and $\var(E_n -\rho_n) = \frac{1}{n} \var(XY) = \frac{1}{n} \to 0$, which implies that $E_n - \rho_n \convprob 0$ by an application of Chebychev's inequality. This concludes the proof.

    \section[Pointwise Asymptotic Null Distribution]{Pointwise Asymptotic Null Distribution~\texorpdfstring{(\Cref{theorem:null-dist-0})}{(Theorem 6)}}
        \label{proof:null-dist-0}
   	In this section, we prove that $\cshsic$ is asymptotically $N(0,1)$ under the null, provided that (i) the kernels $k$ and $\ell$ are fixed in $n$, (ii) the data generating distribution $P_{XY}$ is fixed in $n$, and (iii) $0 < \mathbb{E}[\tildeg^2(Z_1,Z_2)] < \infty$ for $Z_1,Z_2$ independent draws from $P_{XY}$. We first present an outline of the proof in~\Cref{appendix:outline-null-0}, breaking the overall argument into three steps, and then present the details of these steps in the next three subsections. 
   	
   \begin{remark}[moment condition for $\bar{\mathrm{x}}$dCov] Suppose that we use the Euclidean distance kernel in Fact~\ref{fact:dcov-hsic-equivalence} with $x_0=x'$ and $y_0=y'$. Then the kernels become $k(x,x') = - \frac{1}{2} \| x - x'\|$ and  $\ell(y,y') = - \frac{1}{2} \| y - y'\|$. Due to the non-linearity of the Euclidean norm, it is difficult to obtain an explicit form of $\mathbb{E}[\tildeg^2(Z_1,Z_2)]$ associated with the Euclidean distance kernel. Nevertheless, Jensen's inequality gives a sufficient condition for $\mathbb{E}[\tildeg^2(Z_1,Z_2)] < \infty$. First note that 
  	\begin{align}
  		\tildeg\bigl( (x,x'), (y,y')\bigr) ~=~ \frac{1}{4} & \big\{ \| x-x'\| - \mathbb{E}[\|X - x'\|] - \mathbb{E}[\|x - X'\|] + \mathbb{E}[\|X - X'\|] \big\} \\
  		 \times &  \big\{ \| y-y'\| - \mathbb{E}[\|Y - y'\|] - \mathbb{E}[\|y - Y'\|] + \mathbb{E}[\|Y - Y'\|] \big\}.
  	\end{align}
  	Then Jensen's inequality (more specifically, $\{ \mathbb{E}[\|X - X'\|]\}^2 \leq  \mathbb{E}[\|X - X'\|^2]$) along with independence of $X$ and $Y$ yields
  	\begin{align}
  		\mathbb{E}[\tildeg^2(Z_1,Z_2)] \lesssim \mathbb{E}[\|X - X'\|^2] \mathbb{E}[\|Y - Y'\|^2].
  	\end{align}
  	Therefore, $\bar{\mathrm{x}}$dCov is asymptotically $N(0,1)$ given that $\mathbb{E}[\|X - X'\|^2] < \infty$ and $\mathbb{E}[\|Y - Y'\|^2] < \infty$.
 
   \end{remark}
 	
  	\subsection{Outline of the proof}
            \label{appendix:outline-null-0}
  	Throughout the proof, we work with the centered features $\widetilde{\phi}(\cdot) := \phi(\cdot) - \mu$ and $\widetilde{\psi}(\cdot):= \psi(\cdot) - \nu$. This can be done without loss of generality due to the following simple observation: 
  	 \begin{align}
  		h(Z_i,Z_j) = h_{ij} & = \frac{1}{2} \bigl\{ \phi(X_i) - \phi(X_j) \bigr\} \bigl\{ \psi(Y_i) - \psi(Y_j) \bigr\} \\
  		& =  \frac{1}{2} \bigl\{ \phi(X_i) - \mu + \mu - \phi(X_j) \bigr\} \bigl\{ \psi(Y_i) - \nu + \nu - \psi(Y_j) \bigr\} \\
  		& =   \frac{1}{2} \bigl\{ \widetilde{\phi}(X_i) - \widetilde{\phi}(X_j) \bigr\} \bigl\{ \widetilde{\psi}(Y_i) - \widetilde{\psi}(Y_j) \bigr\}.
  	\end{align}
  	Having this observation in mind, the main technical ingredient of the proof is the orthonormal expansion of $\tildeg$ in equation~\eqref{eq:eigendecomposition}. In fact, we can express this orthonormal expansion as the product of the orthonormal expansions of $\tildek$ and $\tildel$, respectively. This leads to
  	\begin{align}
  		\tildeg(z,z') ~=~ & \tildeg((x,y),(x',y')) = \tildek(x,x') \tildel(y,y') \\
  		= ~ &\bigg\{ \sum_{k=1}^{\infty} \lambda_{X,k}  e_{X,k}(x)e_{X,k}(x')  \bigg\} \bigg\{ \sum_{k'=1}^{\infty} \lambda_{Y,k'}  e_{Y,k'}(y)e_{Y,k'}(y')  \bigg\} \\
  		= ~ & \sum_{k=1}^{\infty}\sum_{k'=1}^{\infty} \lambda_{X,k} \lambda_{Y,k'}  e_{X,k}(x)e_{X,k}(x') e_{Y,k'} (y) e_{Y,k'} (y')\\
  		= ~ &  \sum_{k=1}^{\infty} \lambda_{k} e_{k}(z) e_{k} (z').
  	\end{align}
  	 In the proof, one of the main challenges is to handle the infinite sum associated with eigenvalues and eigenfunctions. When the sum is finite, then the usual fixed-dimensional law of large numbers and multivariate central limit theorem establish the desired result in a straightforward manner. However, when dealing with an infinite sum, we need extra care, and we bypass this technical difficulty by leveraging the truncation argument used in the asymptotic analysis of degenerate U-statistics~\citep[e.g.][]{serfling2009approximation}.

    \medskip 
    
  	We break the proof into several pieces for readability. 
  	\begin{itemize}
  		\item Step 1: In the first step, we prove 
  		\begin{align}
  			n \crosshsic ~=~ & \sum_{k=1}^\infty \lambda_k \biggl( \frac{1}{\sqrt{n}}\sum_{i=1}^{n} e_k(Z_i) \biggr) \biggl( \frac{1}{\sqrt{n}}\sum_{t=n+1}^{2n} e_{k}(Z_t) \biggr) + o_P(1).
  		\end{align}
  		\item Step 2: In the second step, we prove
  		\begin{align}
  			(n-1) s_n^2 =  \frac{1}{n} \sum_{i=1}^n \bigg\{ \sum_{k=1}^\infty \lambda_k e_k(Z_i) \biggl( \frac{1}{\sqrt{n}}\sum_{t=n+1}^{2n} e_k(Z_t) \biggr) \bigg\}^2 + o_P(1).
  		\end{align}
  		\item Step 3: In the final step, we prove the bivariate central limit theorem
  		\begin{align}
  			\begin{pmatrix}
  				\sum_{k=1}^\infty \lambda_k \biggl( \frac{1}{\sqrt{n}}\sum_{i=1}^{n} e_k(Z_i) \biggr) \biggl( \frac{1}{\sqrt{n}}\sum_{t=n+1}^{2n} e_{k}(Z_t) \biggr)  \\
  				\frac{1}{n} \sum_{i=1}^n \bigg\{ \sum_{k=1}^\infty \lambda_k e_k(Z_i) \biggl( \frac{1}{\sqrt{n}}\sum_{t=n+1}^{2n} e_k(Z_t) \biggr) \bigg\}^2
  			\end{pmatrix} 
  			\convdist \begin{pmatrix}
  				\sum_{k=1}^\infty \lambda_k W_k \widetilde{W}_k \\[.5em]
  				\sum_{k=1}^\infty \lambda_k^2  \widetilde{W}_k^2
  			\end{pmatrix},
  		\end{align}
  		where $W_1, \widetilde{W}_1, W_2,  \widetilde{W}_2, \ldots$ are i.i.d.~$N(0,1)$. Then Slutsky's theorem along with the continuous mapping theorem proves that 
  		\begin{align}
  			\cshsic ~=~ & \frac{n \crosshsic}{\sqrt{(n-1)s_n^2}} \frac{\sqrt{n-1}}{\sqrt{n}} \\
  			=~ & \frac{ \sum_{k=1}^\infty \lambda_k \biggl( \frac{1}{\sqrt{n}}\sum_{i=1}^{n} e_k(Z_i) \biggr) \biggl( \frac{1}{\sqrt{n}}\sum_{t=n+1}^{2n} e_{k}(Z_t) \biggr) + o_P(1)}{\sqrt{\frac{1}{n} \sum_{i=1}^n \bigg\{ \sum_{k=1}^\infty \lambda_k e_k(Z_i) \biggl( \frac{1}{\sqrt{n}}\sum_{t=n+1}^{2n} e_k(Z_t) \biggr) \bigg\}^2 + o_P(1)}}\{1 + o_P(1)\} \\[.5em]
  			\convdist ~& \frac{\sum_{k=1}^\infty \lambda_k W_k \widetilde{W}_k }{\sqrt{\sum_{k=1}^\infty \lambda_k^2  \widetilde{W}_k^2}} \overset{d}{=} N(0,1).
  		\end{align}
  	\end{itemize}
  	In the subsequent subsections, we present detailed proofs of the above results in each step.
  	
   \medskip 
  	
  	\subsection[Proof of Step 1]{Proof of Step 1 (Numerator)} \label{Section: Proof of Step 1 (Numerator)}
  	We start by decomposing $\crosshsic$ as
  	\begin{align}
  		\crosshsic ~=~ & \frac{1}{n(n-1)}\sum_{1 \leq i \neq j \leq n} \langle h_{ij}, f_2 \rangle \\
  		= ~&  \underbrace{\frac{1}{n} \sum_{i=1}^n \langle \widetilde{\phi}(X_i) \widetilde{\psi}(Y_i), f_2 \rangle}_{\crosshsic^{(1)}} -  \underbrace{\frac{1}{n(n-1)}\sum_{1 \leq i \neq j \leq n} \langle  \widetilde{\phi}(X_i) \widetilde{\psi}(Y_j), f_2 \rangle}_{\crosshsic^{(2)}}.
  	\end{align}
  	Notice that $\crosshsic^{(2)}$ is a degenerate U-statistic of order 2 whose kernel satisfies 
  	\begin{align}
  		\mathbb{E} \bigl[ \langle  \widetilde{\phi}(X_i) \widetilde{\psi}(Y_j), f_2 \rangle  | f_2,Z_i \bigr] = \mathbb{E} \bigl[ \langle  \widetilde{\phi}(X_i) \widetilde{\psi}(Y_j), f_2 \rangle  | f_2,Z_j \bigr] = 0.
  	\end{align}
  	Under the null, $\mathbb{E}[\crosshsic^{(2)} | f_2] = 0$ and 
  	\begin{align}
  		\mathbb{E}\Bigl[\bigl\{\crosshsic^{(2)}\bigr\}^2 \big| f_2 \Bigr] =  O\biggl( \frac{1}{n^2} \biggr) \mathbb{E}\bigl[ \langle \widetilde{\phi}(X_1) \widetilde{\psi}(Y_1), f_2 \rangle^2 | f_2 \bigr].
  	\end{align}
   Moreover,
   \begin{align}
   	\langle \widetilde{\phi}(X_1) \widetilde{\psi}(Y_1), f_2 \rangle^2 ~=~ & \biggl( \frac{1}{n(n-1)} \sum_{n+1 \leq t \neq u \leq 2n} \langle  \widetilde{\phi}(X_1) \widetilde{\psi}(Y_1), h(Z_t,Z_u)\rangle \biggr)^2 \\
   	\lesssim ~ & \biggl( \frac{1}{n} \sum_{t=n+1}^{2n} \langle  \widetilde{\phi}(X_1) \widetilde{\psi}(Y_1),   \widetilde{\phi}(X_t) \widetilde{\psi}(Y_t) \rangle  \biggr)^2 \\
   	& ~~ + \biggl(  \frac{1}{n(n-1)} \sum_{n+1 \leq t \neq u \leq 2n} \langle  \widetilde{\phi}(X_t) \widetilde{\psi}(Y_u), \widetilde{\phi}(X_1) \widetilde{\psi}(Y_1) \rangle \biggr)^2,
   \end{align}
   and therefore its expectation is bounded by
   \begin{align} \label{Eq: bound of tilde terms product f2}
   	\mathbb{E} \bigl[ \langle \widetilde{\phi}(X_1) \widetilde{\psi}(Y_1), f_2 \rangle^2\bigr] \lesssim \frac{1}{n} \mathbb{E}\bigl[ \langle \widetilde{\phi}(X_1) \widetilde{\psi}(Y_1), \widetilde{\phi}(X_2) \widetilde{\psi}(Y_2) \rangle^2 \bigr] = \frac{1}{n} \mathbb{E}\bigl[ \tildeg(Z_1,Z_2)^2 \bigr].
   \end{align}
	Hence, for $t \geq 0$, Chebyshev's inequality yields
	\begin{align}
		\mathbb{E}\bigl[\mathbb{P}\bigl( |\crosshsic^{(2)}| \geq t  | f_2 \bigr)\bigr] \lesssim~ & \frac{1}{t^2n^2} \mathbb{E}\Bigl[\mathbb{E}\bigl[ \langle \widetilde{\phi}(X_1) \widetilde{\psi}(Y_1), f_2 \rangle^2 | f_2 \bigr]\Bigr] \\
		\lesssim ~ & \frac{1}{t^2n^3} \mathbb{E}\bigl[ \tildeg(Z_1,Z_2)^2 \bigr],
	\end{align}
  	which concludes that $\crosshsic^{(2)} = O_P(n^{-3/2}) = o_P(n^{-1})$.
  	
  	Next we study $\crosshsic^{(1)}$, which equals 
  	\begin{align}
  		& \crosshsic^{(1)} ~=~ \frac{1}{n} \sum_{i=1}^n \langle \widetilde{\phi}(X_i) \widetilde{\psi}(Y_i), f_2 \rangle \\
  		= ~ &  \frac{1}{n^2} \sum_{i=1}^n \sum_{t = n+1}^{2n} \langle \widetilde{\phi}(X_i) \widetilde{\psi}(Y_i),  \widetilde{\phi}(X_t) \widetilde{\psi}(Y_t)\rangle - \frac{1}{n^2(n-1)} \sum_{i=1}^n \sum_{n+1 \leq t \neq u \leq 2n} \langle  \widetilde{\phi}(X_i) \widetilde{\psi}(Y_i), \widetilde{\phi}(X_t) \widetilde{\psi}(Y_u) \rangle. 
  	\end{align}
  	The second term above has zero expectation and its variance is bounded above by 
  	\begin{align}
  		O\biggl( \frac{1}{n^3} \mathbb{E}\bigl[ \tildeg(Z_1,Z_2)^2 \bigr] \biggr).
  	\end{align}
  	As a result, we have
  	\begin{align}
  		\crosshsic^{(1)} ~=~ & \frac{1}{n^2} \sum_{i=1}^n \sum_{t = n+1}^{2n} \langle \widetilde{\phi}(X_i) \widetilde{\psi}(Y_i),  \widetilde{\phi}(X_t) \widetilde{\psi}(Y_t)\rangle + O_P(n^{-3/2}) \\
  		= ~ &  \sum_{k=1}^\infty \lambda_k \biggl( \frac{1}{n}\sum_{i=1}^{n} e_k(Z_i) \biggr) \biggl( \frac{1}{n}\sum_{t=n+1}^{2n} e_{k}(Z_t) \biggr)  + O_P(n^{-3/2}).
  	\end{align}
  	Putting all together, we see 
  	\begin{align}
  		n \crosshsic ~=~ & \sum_{k=1}^\infty \lambda_k \biggl( \frac{1}{\sqrt{n}}\sum_{i=1}^{n} e_k(Z_i) \biggr) \biggl( \frac{1}{\sqrt{n}}\sum_{t=n+1}^{2n} e_{k}(Z_t) \biggr) + o_P(1),
  	\end{align} 
  	as claimed in Step 1.
  	
  	\medskip 
  	
  	\subsection[Proof of Step 2]{Proof of Step 2 (Denominator)} 
  	In this step, we would like to show $(n-1) s_n^2$ approximates
  	\begin{align}
  		(n-1) s_n^2 =  \frac{1}{n} \sum_{i=1}^n \bigg\{ \sum_{k=1}^\infty \lambda_k e_k(Z_i) \biggl( \frac{1}{\sqrt{n}}\sum_{t=n+1}^{2n} e_k(Z_t) \biggr) \bigg\}^2 + o_P(1).
  	\end{align}
  	This part is much more challenging to prove than Step 1 partly due to the complexity of $s_n^2$. To simplify the problem a bit, we first prove that $\crosshsic^2 = O_P(n^{-2})$. Indeed, from the previous results in Step 1, for this claim to hold, it suffices to prove 
  	\begin{align}
  		\frac{1}{n^2} \sum_{i=1}^n \sum_{t = n+1}^{2n} \langle \widetilde{\phi}(X_i) \widetilde{\psi}(Y_i),  \widetilde{\phi}(X_t) \widetilde{\psi}(Y_t)\rangle   = O_P(n^{-1}).
  	\end{align}
  	This follows by the Chebyshev's argument as before given that it has the expectation zero and the variance bounded by $O(n^{-2} \mathbb{E}\bigl[ \tildeg(Z_1,Z_2)^2 \bigr])$. Consequently, we have $\crosshsic^2 = O_P(n^{-2})$, and thus
  	 \begin{align} 
  		s_n^2 ~=~ & \frac{4(n-1)}{(n-2)^2} \lb \frac{1}{(n-1)^2} \sum_{i=1}^n \lp \sum_{j=1}^{n, j\neq i} \la h(Z_i, Z_j), f_2 \ra \rp^2 - n \crosshsic^2  \rb \\
  		=~ & \frac{4(n-1)}{(n-1)^2(n-2)^2} \sum_{i=1}^n \lp \sum_{j=1}^{n, j\neq i} \la h(Z_i, Z_j), f_2 \ra \rp^2 + O_P(n^{-2}).
  	\end{align}
  	The first term above further approximates 
  	\begin{align}
  		& \frac{4(n-1)}{(n-1)^2(n-2)^2} \sum_{i=1}^n \lp \sum_{j=1}^{n, j\neq i} \la h(Z_i, Z_j), f_2 \ra \rp^2 \\ 
  		~=~ &  \frac{4(n-1)}{(n-1)^2(n-2)^2} \sum_{1 \leq i \neq j \leq n} \la h(Z_i, Z_j), f_2 \ra^2  +  \frac{4(n-1)}{(n-1)^2(n-2)^2} \sum_{\substack{1\leq i,j,q \leq n \\ \text{$i,j,q$ distinct}}}  \la h(Z_i, Z_j), f_2 \ra \la h(Z_i, Z_q), f_2 \ra \\
  		= ~ & \frac{4}{(n-1)(n-2)^2} \sum_{\substack{1\leq i,j,q \leq n \\ \text{$i,j,q$ distinct}}}  \la h(Z_i, Z_j), f_2 \ra \la h(Z_i, Z_q), f_2 \ra + O_P(n^{-2}),
  	\end{align}
  	where the last step uses Markov's inequality together with 
  	\begin{align}
  		\mathbb{E}\bigl[ \la h(Z_i, Z_j), f_2 \ra^2 \bigr] ~\lesssim~ &  \mathbb{E} \biggl[ \la \widetilde{\phi}(X_i) \widetilde{\psi}(Y_i), f_2 \ra^2 \biggr] + \mathbb{E} \biggl[ \la \widetilde{\phi}(X_i) \widetilde{\psi}(Y_j), f_2 \ra^2 \biggr] \\
  		+ & \mathbb{E} \biggl[ \la \widetilde{\phi}(X_j) \widetilde{\psi}(Y_j), f_2 \ra^2 \biggr] + \mathbb{E} \biggl[ \la \widetilde{\phi}(X_j) \widetilde{\psi}(Y_i), f_2 \ra^2 \biggr]  ~\lesssim ~ \frac{1}{n} \mathbb{E}\bigl[ \tildeg(Z_1,Z_2)^2 \bigr],
  	\end{align}
  	due to the upper bound in \eqref{Eq: bound of tilde terms product f2}. Thus the main term to investigate is 
  	\begin{align}
  		s_{\mathrm{main},n}^2 :=~  & \frac{4}{n(n-1)(n-2)}\sum_{\substack{1\leq i,j,q \leq n \\ \text{$i,j,q$ distinct}}}  \la h(Z_i, Z_j), f_2 \ra \la h(Z_i, Z_q), f_2 \ra \\
  		= ~ &  \frac{1}{n(n-1)(n-2)}\sum_{\substack{1\leq i,j,q \leq n \\ \text{$i,j,q$ distinct}}}  \la \bigl\{ \widetilde{\phi}(X_i) - \widetilde{\phi}(X_j) \bigr\} \bigl\{ \widetilde{\psi}(Y_i) - \widetilde{\psi}(Y_j) \bigr\}, f_2 \ra \\
  		& ~~~~~~~~~~~~~~~~~~~~~~~~~~~~~~~~~\times \la \bigl\{ \widetilde{\phi}(X_i) - \widetilde{\phi}(X_q) \bigr\} \bigl\{ \widetilde{\psi}(Y_i) - \widetilde{\psi}(Y_q) \bigr\}, f_2 \ra.
  	\end{align}
  	We also claim that 
  	\begin{align} \label{Eq: s_n approximation}
  		s_{\mathrm{main},n}^2 =  \frac{1}{n} \sum_{i=1}^n \la \widetilde{\phi}(X_i)\widetilde{\psi}(Y_i),  \frac{1}{n}\sum_{t=n+1}^{2n} \widetilde{\phi}(X_t)\widetilde{\psi}(Y_t) \ra^2 + o_P(n^{-1}),
  	\end{align}
  	which yields the desired result in Step 2. Note that $s_{\mathrm{main},n}^2$ is a U-statistic of order 3 with a varying kernel in $n$ conditional on $f_2$. Therefore we are not able to directly apply the usual approximation theory of U-statistics, which focuses on a fixed kernel, in the process of obtaining approximation~\eqref{Eq: s_n approximation}. It turns out that the analysis is non-trivial especially only with the finite second moment of $\tildeg$. We postpone the detailed (long) analysis to Appendix~\ref{appendix: s_n approximation}. 
  	
  	\medskip 
  	
  	\subsection{Proof of Step 3} \label{Section: Proof of Step 3}
  	By the Cram\'{e}r--Wold device, the bivariate central limit theorem holds if for each $t_1,t_2 \in \mathbb{R}$,
  	\begin{align}
  		& \underbrace{t_1 \sum_{k=1}^\infty \lambda_k \biggl( \frac{1}{\sqrt{n}}\sum_{i=1}^{n} e_k(Z_i) \biggr) \biggl( \frac{1}{\sqrt{n}}\sum_{t=n+1}^{2n} e_{k}(Z_t) \biggr)  + t_2 \frac{1}{n} \sum_{i=1}^n \bigg\{ \sum_{k=1}^\infty \lambda_k e_k(Z_i) \biggl( \frac{1}{\sqrt{n}}\sum_{t=n+1}^{2n} e_k(Z_t) \biggr) \bigg\}^2}_{T_n} \\
  		\convdist ~& \underbrace{t_1 \sum_{k=1}^\infty \lambda_k W_k \widetilde{W}_k +t_2\sum_{k=1}^\infty \lambda_k^2  \widetilde{W}_k^2}_{T}.
  	\end{align}
  	In order to establish this, we make use of the truncation argument as in Chapter 5 of \citet{serfling2009approximation}. First of all, for some fixed $K$, define 
  	\begin{align}
  		&T_{n,K} :=  t_1 \sum_{k=1}^K \lambda_k \biggl( \frac{1}{\sqrt{n}}\sum_{i=1}^{n} e_k(Z_i) \biggr) \biggl( \frac{1}{\sqrt{n}}\sum_{t=n+1}^{2n} e_{k}(Z_t) \biggr) \\
  		& \qquad \qquad \qquad + t_2 \frac{1}{n} \sum_{i=1}^n \bigg\{ \sum_{k=1}^K \lambda_k e_k(Z_i) \biggl( \frac{1}{\sqrt{n}}\sum_{t=n+1}^{2n} e_k(Z_t) \biggr) \bigg\}^2, \\
  		&T_{K} :=  t_1 \sum_{k=1}^K \lambda_k W_k \widetilde{W}_k +t_2\sum_{k=1}^K \lambda_k^2  \widetilde{W}_k^2.
  	\end{align}
  	Then by the triangle inequality
  	\begin{align}
  		\mathbb{E}[| T_{n} - T_{n,K}| ] \leq~ & |t_1| \mathbb{E}\biggl[ \bigg| \sum_{k=K+1}^\infty \lambda_k \biggl( \frac{1}{\sqrt{n}}\sum_{i=1}^{n} e_k(Z_i) \biggr) \biggl( \frac{1}{\sqrt{n}}\sum_{t=n+1}^{2n} e_{k}(Z_t) \biggr) \bigg| \biggr] \\
  		~& |t_2|  \mathbb{E} \biggl[  \bigg| \frac{1}{n} \sum_{i=1}^n \bigg\{ \sum_{k=1}^\infty \lambda_k e_k(Z_i) \biggl( \frac{1}{\sqrt{n}}\sum_{t=n+1}^{2n} e_k(Z_t) \biggr) \bigg\}^2 \\
  		& \hskip 5em - \frac{1}{n} \sum_{i=1}^n \bigg\{ \sum_{k=1}^K \lambda_k e_k(Z_i) \biggl( \frac{1}{\sqrt{n}}\sum_{t=n+1}^{2n} e_k(Z_t) \biggr) \bigg\}^2  \bigg|\biggr] \\
  		= ~& |t_1| (\mathrm{I}) + |t_2| (\mathrm{II}).
  	\end{align}
  	By the Cauchy--Schwarz inequality and using that $\{e_k\}_{k=1}^\infty$ are orthonormal and $\mathbb{E}[e_k(Z)] = 0$,
  	\begin{align}
  		(\mathrm{I}) \leq \sqrt{\sum_{k=K+1}^\infty \lambda_k^2}.
  	\end{align}
  	On the other hand, letting
  	\begin{align}
  		\sum_{k=1}^\infty \lambda_k e_k(Z_i) \biggl( \frac{1}{\sqrt{n}}\sum_{t=n+1}^{2n} e_k(Z_t) \biggr) ~=~& \sum_{k=1}^K \lambda_k e_k(Z_i) \biggl( \frac{1}{\sqrt{n}}\sum_{t=n+1}^{2n} e_k(Z_t) \biggr) \\
  		& + \sum_{k=K+1}^\infty \lambda_k e_k(Z_i) \biggl( \frac{1}{\sqrt{n}}\sum_{t=n+1}^{2n} e_k(Z_t) \biggr) \\
  		= ~ & A_i+ B_i
  	\end{align}
  	and using $|(A_i+B_i)^2 - A_i^2| \leq B_i^2 + 2|A_iB_i|$, we have
  	\begin{align}
  		(\mathrm{II}) \leq \frac{1}{n} \sum_{i=1}^n \mathbb{E}\bigl[B_i^2 + 2|A_iB_i|\bigr]  \leq \frac{1}{n} \sum_{i=1}^n \mathbb{E}\bigl[B_i^2\bigr] + 2\frac{1}{n} \sum_{i=1}^n \sqrt{\mathbb{E} \bigl[A_i^2\bigr]}  \sqrt{\mathbb{E} \bigl[B_i^2\bigr]},
  	\end{align}
  	where the last inequality uses the Cauchy--Schwarz inequality. Again by using the orthonormal property of  $\{e_k\}_{k=1}^\infty$ and $\mathbb{E}[e_k(Z)] = 0$,
  	\begin{align}
  		\mathbb{E} \bigl[A_i^2\bigr] = \sum_{k=1}^K \lambda_k^2 \quad \text{and} \quad \mathbb{E} \bigl[B_i^2\bigr] = \sum_{k=K+1}^\infty \lambda_k^2.
  	\end{align}
  	Therefore, noting that $\mathbb{E}[\tildeg^2(Z_1,Z_2)] = \sum_{k=1}^K \lambda_k^2 <\infty$,
  	\begin{align}
  		\mathbb{E}[| T_{n} - T_{n,K}| ] \leq  |t_1| \sqrt{\sum_{k=K+1}^\infty \lambda_k^2} + |t_2| \sum_{k=K+1}^\infty \lambda_k^2  + 2|t_2| \sqrt{\sum_{k=1}^K \lambda_k^2}\sqrt{\sum_{k=K+1}^\infty \lambda_k^2},
  	\end{align}
  	which goes to zero as $K \rightarrow \infty$ uniformly over $n$. Hence $T_{n,K}$ converges to $T_{n}$ in distribution as $K \rightarrow \infty$ uniformly over $n$. Similarly, we obtain
  	\begin{align}
  		\mathbb{E}[|T - T_K|] ~\leq~ & |t_1| \mathbb{E}\biggl[ \bigg| \sum_{k=K+1}^\infty \lambda_k W_k \widetilde{W}_k \bigg| \biggr] + |t_2| \mathbb{E} \biggl[\sum_{k=K+1}^\infty \lambda_k^2  \widetilde{W}_k^2 \biggr] \\
  		\leq ~ & |t_1| \sqrt{\sum_{k=K+1}^\infty \lambda_k^2} + |t_2| \sum_{k=K+1}^\infty \lambda_k^2,
  	\end{align}
  	which goes to zero as $K \rightarrow \infty$. In addition, for each fixed $K$, the multivariate central limit theorem yields
  	\begin{align}
  		\begin{pmatrix}
  			\frac{1}{\sqrt{n}}\sum_{i=1}^{n} e_1(Z_i) \\[.5em]
  			\vdots  \\[.5em]
  			\frac{1}{\sqrt{n}}\sum_{i=1}^{n} e_K(Z_i) \\[.5em]
  			 \frac{1}{\sqrt{n}}\sum_{t=n+1}^{2n} e_1(Z_t) \\[.5em]
  			 \vdots \\[.5em]
  			 \frac{1}{\sqrt{n}}\sum_{t=n+1}^{2n} e_K(Z_t)
  		\end{pmatrix} 
  		\convdist N_{2K}(0, I),
  	\end{align}
  	and the law of large numbers shows 
  	\begin{align}
  		\frac{1}{n} \sum_{i=1}^n e_k(Z_i) e_{k'}(Z_i) \convprob \begin{cases}
  			1 & \text{$k=k'$,} \\
  			0 & \text{otherwise.}
  		\end{cases}
  	\end{align}
  	Using these results, Slutsky's theorem together with the continuous mapping theorem proves that $T_{n,K} \convdist T_{K}$ for each fixed $K$. 
  	
  	Having these preliminary results in place, we finish the proof of Step 3 using the argument in \citet{serfling2009approximation} as follows. Let $\varphi_n$, $\varphi_{n,K}$, $\varphi$ and $\varphi_K$ be the characteristic functions of $T_n$, $T_{n,K}$, $T$ and $T_{K}$, respectively. Let $\epsilon >0$ be some fixed number. Then we can choose $K_1>0$ such that for all $K \geq K_1$ and for all $n$, 
  	\begin{align}
  		|\varphi_n(s) - \varphi_{n,K}(s)| ~=~& |\mathbb{E}[e^{isT_n} - e^{is T_{n,K}}| \\
  		\leq ~& \mathbb{E}|e^{is(T_n - T_{n,K})} - 1| \\
  		\leq ~ & |s| \mathbb{E}[|T_n - T_{n,K}|]  < \frac{\epsilon}{3}.
  	\end{align}
  	We can also choose $K_2>0$ such that for all $K \geq K_2$, 
  	\begin{align}
  		|\varphi(s) - \varphi_{K}(s)| < \frac{\epsilon}{3}.
  	\end{align}
  	Since $T_{n,K} \convdist T_{K}$ for each $K$, we can choose $N$ such that for all $n \geq N$ and $K_0 = \max\{K_1,K_2\}$,
  	\begin{align}
  		|\varphi_{n,K}(s) - \varphi_{K}| \leq \frac{\epsilon}{n}.
  	\end{align}
  	Therefore, by the triangle inequality,
  	\begin{align}
  		|\varphi_n(s)-\varphi(s)| \leq |\varphi_n(s) - \varphi_{n,K_0}(s)| + |\varphi_{n,K_0}(s) - \varphi_{K_0}| + |\varphi(s) - \varphi_{K_0}(s)| \leq \epsilon. 
  	\end{align}
  	Since $\epsilon$ was arbitrary, we conclude that $\lim_{n \to \infty} \varphi_n(s)  = \varphi(s)$ as desired.

  	\bigskip 
  	
  	\subsection{Details of approximation~\texorpdfstring{\eqref{Eq: s_n approximation}}{(24)}} \label{appendix: s_n approximation}
  	The aim of this section is to establish approximation~\eqref{Eq: s_n approximation}. First note that $s_{\mathrm{main},n}^2$ has the following decomposition:
  	\begin{align}
  		s_{\mathrm{main},n}^2 = ~& \frac{1}{n} \sum_{i=1}^n \la \widetilde{\phi}(X_i)\widetilde{\psi}(Y_i), f_2 \ra^2 - \frac{2}{n(n-1)} \sum_{1 \leq i \neq j \leq n}  \la \widetilde{\phi}(X_i)\widetilde{\psi}(Y_i), f_2 \ra \la \widetilde{\phi}(X_i)\widetilde{\psi}(Y_j), f_2 \ra \\
  		-& \frac{2}{n(n-1)} \sum_{1 \leq i \neq j \leq n}  \la \widetilde{\phi}(X_i)\widetilde{\psi}(Y_i), f_2 \ra \la \widetilde{\phi}(X_j)\widetilde{\psi}(Y_i), f_2 \ra \\
  		+&  \frac{3}{n(n-1)} \sum_{1 \leq i \neq j \leq n}  \la \widetilde{\phi}(X_i)\widetilde{\psi}(Y_i), f_2 \ra \la \widetilde{\phi}(X_j)\widetilde{\psi}(Y_j), f_2 \ra \\
  		- &   \frac{4}{n(n-1)(n-2)}\sum_{\substack{1\leq i,j,q \leq n \\ \text{$i,j,q$ distinct}}} \la \widetilde{\phi}(X_i)\widetilde{\psi}(Y_i), f_2 \ra \la \widetilde{\phi}(X_j)\widetilde{\psi}(Y_q), f_2 \ra \\
  		+  &\frac{1}{n(n-1)(n-2)}\sum_{\substack{1\leq i,j,q \leq n \\ \text{$i,j,q$ distinct}}} \la \widetilde{\phi}(X_i)\widetilde{\psi}(Y_j), f_2 \ra \la \widetilde{\phi}(X_i)\widetilde{\psi}(Y_q), f_2 \ra \\
  		+  &\frac{1}{n(n-1)(n-2)}\sum_{\substack{1\leq i,j,q \leq n \\ \text{$i,j,q$ distinct}}} \la \widetilde{\phi}(X_i)\widetilde{\psi}(Y_j), f_2 \ra \la \widetilde{\phi}(X_q)\widetilde{\psi}(Y_i), f_2 \ra \\
  		+  &\frac{1}{n(n-1)(n-2)}\sum_{\substack{1\leq i,j,q \leq n \\ \text{$i,j,q$ distinct}}} \la \widetilde{\phi}(X_j)\widetilde{\psi}(Y_i), f_2 \ra \la \widetilde{\phi}(X_i)\widetilde{\psi}(Y_q), f_2 \ra \\
  		+  &\frac{1}{n(n-1)(n-2)}\sum_{\substack{1\leq i,j,q \leq n \\ \text{$i,j,q$ distinct}}} \la \widetilde{\phi}(X_j)\widetilde{\psi}(Y_i), f_2 \ra \la \widetilde{\phi}(X_q)\widetilde{\psi}(Y_i), f_2 \ra \\
  		:= ~ & \sum_{i=1}^9 J_{(i)}.
  	\end{align}
  	By defining $f_{2,A}$, $f_{2,B}$ and $f_{2,C}$ as
  	\begin{align}
  		f_2 ~=~& \frac{1}{n} \sum_{t = n+1}^{2n} \widetilde{\phi}(X_t) \widetilde{\psi}(Y_t) - \frac{1}{n(n-1)} \sum_{n+1 \leq t\neq u \leq 2n} \widetilde{\phi}(X_t) \widetilde{\psi}(Y_u) \\
  		=~&  \frac{1}{n} \sum_{t = n+1}^{2n} \widetilde{\phi}(X_t) \widetilde{\psi}(Y_t) -  \frac{n}{n-1} \bigg\{ \frac{1}{n}\sum_{n+ 1 \leq t \leq 2n}  \widetilde{\phi}(X_t)  \bigg\} \bigg\{  \frac{1}{n} \sum_{n+ 1 \leq u \leq 2n}  \widetilde{\psi}(Y_u) \bigg\} \\
  		& \qquad  +\frac{1}{n(n-1)}  \sum_{t = n+1}^{2n} \widetilde{\phi}(X_t) \widetilde{\psi}(Y_t) \\
  		:=~ & f_{2,A} + f_{2,B} + f_{2,C}.
  	\end{align}
  	we analyze each $J_{(i)}$ term, separately.

  	\bigskip 
  	
  	\noindent \textbf{1.~Analyzing $J_{(1)}$.} Starting with $J_{(1)}$, we will show that 
  	\begin{align} \label{Eq: J1 approximation}
  		J_{(1)} = \frac{1}{n} \sum_{i=1}^n \la \widetilde{\phi}(X_i)\widetilde{\psi}(Y_i),  \frac{1}{n}\sum_{t=n+1}^{2n} \widetilde{\phi}(X_t)\widetilde{\psi}(Y_t) \ra^2 + o_P(n^{-1}).
  	\end{align}
  	First note that 
  	\begin{align}
  		J_{(1)} =~ & \frac{1}{n} \sum_{i=1}^n \la \widetilde{\phi}(X_i)\widetilde{\psi}(Y_i), f_{2,A} \ra^2 + \frac{2}{n} \sum_{i=1}^n \la \widetilde{\phi}(X_i)\widetilde{\psi}(Y_i), f_{2,A} \ra \la \widetilde{\phi}(X_i)\widetilde{\psi}(Y_i), f_{2,B} \ra \\
  		+ &  \frac{1}{n} \sum_{i=1}^n \la \widetilde{\phi}(X_i)\widetilde{\psi}(Y_i), f_{2,B} \ra^2 + \frac{2}{n} \sum_{i=1}^n \la \widetilde{\phi}(X_i)\widetilde{\psi}(Y_i), f_{2,B} \ra \la \widetilde{\phi}(X_i)\widetilde{\psi}(Y_i), f_{2,C} \ra \\
  		+ & \frac{1}{n} \sum_{i=1}^n \la \widetilde{\phi}(X_i)\widetilde{\psi}(Y_i), f_{2,C} \ra^2 + \frac{2}{n} \sum_{i=1}^n \la \widetilde{\phi}(X_i)\widetilde{\psi}(Y_i), f_{2,A} \ra \la \widetilde{\phi}(X_i)\widetilde{\psi}(Y_i), f_{2,C} \ra.
  	\end{align}
  	We can ignore the terms involving $f_{2,C}$ since they have a faster convergence rate than the same terms replacing $f_{2,C}$ with $f_{2,A}$. By noting that 
  	\begin{align}
  		\frac{1}{n} \sum_{i=1}^n \la \widetilde{\phi}(X_i)\widetilde{\psi}(Y_i), f_{2,A} \ra^2 = \frac{1}{n} \sum_{i=1}^n \la \widetilde{\phi}(X_i)\widetilde{\psi}(Y_i),  \frac{1}{n}\sum_{t=n+1}^{2n} \widetilde{\phi}(X_t)\widetilde{\psi}(Y_t) \ra^2,
  	\end{align}
  	we shall prove
  	\begin{align}
  		\underbrace{\frac{2}{n} \sum_{i=1}^n \la \widetilde{\phi}(X_i)\widetilde{\psi}(Y_i), f_{2,A} \ra \la \widetilde{\phi}(X_i)\widetilde{\psi}(Y_i), f_{2,B} \ra}_{:=J_{(1),a}} +  \underbrace{\frac{1}{n} \sum_{i=1}^n \la \widetilde{\phi}(X_i)\widetilde{\psi}(Y_i), f_{2,B} \ra^2}_{:=J_{(1),b}} = o_P(n^{-1}),
  	\end{align}
  	and thus establish approximation~\eqref{Eq: J1 approximation}. 
  	
  	Focusing on $J_{(1),b}$, notice that
   \begin{align}
   		n J_{(1),b} ~=~& n \times \frac{1}{n}\sum_{i=1}^n \la \widetilde{\phi}(X_i)\widetilde{\psi}(Y_i), f_{2,B} \ra^2 \\
   		=~ & \frac{n^2}{(n-1)^2} \frac{1}{n^2}\sum_{i=1}^n  \bigg\{ \sum_{k=1}^\infty \lambda_{X,k} e_{X,k}(X_i) \biggl( \frac{1}{\sqrt{n}}\sum_{n+ 1 \leq t \leq 2n}  e_{X,k}(X_t) \biggr) \bigg\}^2 \\
   		& \hskip 7em \bigg\{ \sum_{k'=1}^\infty \lambda_{Y,k'} e_{Y,k'}(Y_i) \biggl( \frac{1}{\sqrt{n}}\sum_{n+ 1 \leq t \leq 2n}  e_{Y,k'}(Y_t) \biggr) \bigg\}^2 := \frac{n^2}{(n-1)^2} \frac{1}{n} G_{(1)}. 
   \end{align}
   In addition, the orthonormal property of eigenfunctions yields
   \begin{align}
   	&\mathbb{E} \biggl[  \bigg\{ \sum_{k=1}^\infty \lambda_{X,k} e_{X,k}(X_i) \biggl( \frac{1}{\sqrt{n}}\sum_{n+ 1 \leq t \leq 2n}  e_{X,k}(X_t) \biggr) \bigg\}^2 \biggr] = \sum_{k=1}^\infty \lambda_{X,k}^2 \quad \text{and}\\
   	& \mathbb{E} \biggl[  \bigg\{ \sum_{k'=1}^\infty \lambda_{Y,k'} e_{Y,k'}(Y_i) \biggl( \frac{1}{\sqrt{n}}\sum_{n+ 1 \leq t \leq 2n}  e_{Y,k'}(Y_t) \biggr) \bigg\}^2   \biggr] = \sum_{k'=1}^\infty \lambda_{Y,k'}^2.
   \end{align}
	This implies that $\mathbb{E}[G_{(1)}] = \sum_{k=1}^\infty \lambda_k^2 < \infty$ since $X$ and $Y$ are independent. Therefore using Markov's inequality,
   \begin{align}  \label{Eq: approximation f2B}
   	\frac{1}{n}\sum_{i=1}^n \la \widetilde{\phi}(X_i)\widetilde{\psi}(Y_i), f_{2,B} \ra^2 = O_P(n^{-2})= o_P(n^{-1}).
   \end{align}
	Next, we turn to $J_{(1),a}$. As shown before in Appendix~\ref{Section: Proof of Step 3}, 
	\begin{align}
		\sum_{i=1}^n \la \widetilde{\phi}(X_i)\widetilde{\psi}(Y_i), f_{2,A} \ra^2 \convdist \sum_{k=1}^\infty \lambda_k^2, \widetilde{W}_k^2
	\end{align}
	which implies 
	\begin{align}  \label{Eq: approximation f2A}
		\frac{1}{n} \sum_{i=1}^n \la \widetilde{\phi}(X_i)\widetilde{\psi}(Y_i), f_{2,A} \ra^2 = O_P(n^{-1}).
	\end{align}
	Hence, by the Cauchy--Schwarz inequality,
	\begin{align}
		|J_{(1),a}| ~=~ & \bigg| \frac{2}{n} \sum_{i=1}^n \la \widetilde{\phi}(X_i)\widetilde{\psi}(Y_i), f_{2,A} \ra \la \widetilde{\phi}(X_i)\widetilde{\psi}(Y_i), f_{2,B} \ra\bigg| \\
		\leq ~ & 2 \sqrt{\frac{1}{n} \sum_{i=1}^n \la \widetilde{\phi}(X_i)\widetilde{\psi}(Y_i), f_{2,A} \ra^2}   \sqrt{\frac{1}{n} \sum_{i=1}^n \la \widetilde{\phi}(X_i)\widetilde{\psi}(Y_i), f_{2,B} \ra^2} \\
		= ~ & 2 \sqrt{O_P(n^{-1})} \sqrt{O_P(n^{-2})} =o_P(n^{-1})
	\end{align}
	where we use the approximation result \eqref{Eq: approximation f2B} for the second term in the upper bound. Combining results establishes the approximation~\eqref{Eq: J1 approximation}. 
	
	\bigskip 
	
	\noindent \textbf{2.~Analyzing $J_{(2)}$.} By the same logic as in the analysis of $J_{(1)}$, we can ignore the terms involving $f_{2,C}$ throughout the proof, and we only need to handle three terms
	\begin{align}
		J_{(2),a} := & \frac{2}{n(n-1)} \sum_{1 \leq i \neq j \leq n}  \la \widetilde{\phi}(X_i)\widetilde{\psi}(Y_i), f_{2,A} \ra \la \widetilde{\phi}(X_i)\widetilde{\psi}(Y_j), f_{2,A} \ra, \\
		J_{(2),b} := & \frac{2}{n(n-1)} \sum_{1 \leq i \neq j \leq n}  \la \widetilde{\phi}(X_i)\widetilde{\psi}(Y_i), f_{2,B} \ra \la \widetilde{\phi}(X_i)\widetilde{\psi}(Y_j), f_{2,B} \ra \quad \text{and} \\
		J_{(2),c} :=  & \frac{2}{n(n-1)} \sum_{1 \leq i \neq j \leq n}  \la \widetilde{\phi}(X_i)\widetilde{\psi}(Y_i), f_{2,A} \ra \la \widetilde{\phi}(X_i)\widetilde{\psi}(Y_j), f_{2,B} \ra.
	\end{align}
	For the first term $J_{(2),a}$,
	\begin{align}
		 J_{(2),a} =~ &\frac{2}{n(n-1)} \sum_{1 \leq i \neq j \leq n}  \la \widetilde{\phi}(X_i)\widetilde{\psi}(Y_i), f_{2,A} \ra \la \widetilde{\phi}(X_i) \widetilde{\psi}(Y_j), f_{2,A} \ra \\
		 =~ &\frac{2}{n(n-1)}  \sum_{i=1}^n\sum_{j=1}^n  \la \widetilde{\phi}(X_i)\widetilde{\psi}(Y_i), f_{2,A} \ra \la \widetilde{\phi}(X_i) \widetilde{\psi}(Y_j), f_{2,A} \ra \\
		 & - \frac{2}{n(n-1)}\sum_{i=1}^n \la \widetilde{\phi}(X_i)\widetilde{\psi}(Y_i), f_{2,A} \ra \la \widetilde{\phi}(X_i) \widetilde{\psi}(Y_i), f_{2,A} \ra \\
		 =~ & \underbrace{\frac{2n}{n-1} \times  \frac{1}{n}\sum_{i=1}^n  \la \widetilde{\phi}(X_i)\widetilde{\psi}(Y_i), f_{2,A} \ra \la \widetilde{\phi}(X_i) \biggl(\frac{1}{n}\sum_{j=1}^n\widetilde{\psi}(Y_j)\biggr), f_{2,A} \ra}_{O_P(n^{-3/2})} \\
		 & -  \underbrace{\frac{2}{n(n-1)}\sum_{i=1}^n \la \widetilde{\phi}(X_i)\widetilde{\psi}(Y_i), f_{2,A} \ra^2}_{O_P(n^{-2})},
	\end{align}
	where the second approximation result $O_P(n^{-2})$ holds by \eqref{Eq: approximation f2A}. On the other hand, the first approximation $O_P(n^{-3/2})$ holds since
	\begin{align}
		&\bigg| \frac{1}{n}\sum_{i=1}^n  \la \widetilde{\phi}(X_i)\widetilde{\psi}(Y_i), f_{2,A} \ra \la \widetilde{\phi}(X_i) \biggl(\frac{1}{n}\sum_{j=1}^n\widetilde{\psi}(Y_j)\biggr), f_{2,A} \ra\bigg| \\
		\leq ~ & \sqrt{ \frac{1}{n}\sum_{i=1}^n  \la \widetilde{\phi}(X_i)\widetilde{\psi}(Y_i), f_{2,A} \ra^2 }  \sqrt{ \frac{1}{n}\sum_{i=1}^n  \la \widetilde{\phi}(X_i) \biggl(\frac{1}{n}\sum_{j=1}^n\widetilde{\psi}(Y_j)\biggr), f_{2,A} \ra^2 } \\
		= ~ & O_P(n^{-1/2}) O_P(n^{-1}) = O_P(n^{-3/2}),
	\end{align}
	by the Cauchy--Schwarz inequality and Markov's inequality. In particular, it can be seen that 
	\begin{align}
		\mathbb{E} \Biggl[   \frac{1}{n}\sum_{i=1}^n  \la \widetilde{\phi}(X_i)  \widetilde{\psi}(Y_j), f_{2,A} \ra^2 \Biggr] = \frac{1}{n^3} \sum_{i=1}^n \sum_{j=1}^n \mathbb{E}\biggl[  \la \widetilde{\phi}(X_i)\widetilde{\psi}(Y_j), f_{2,A} \ra^2 \biggr] \lesssim \frac{1}{n^2} \mathbb{E}[\tildeg^2(Z_1,Z_2)]. 
	\end{align}	
	Therefore $J_{(2),a} = o_P(n^{-1})$.

	For the second term $J_{(2),b}$,
	\begin{align}
		J_{(2),b}= ~ &\frac{2}{n(n-1)} \sum_{1 \leq i \neq j \leq n}  \la \widetilde{\phi}(X_i)\widetilde{\psi}(Y_i), f_{2,B} \ra \la \widetilde{\phi}(X_i)\widetilde{\psi}(Y_j), f_{2,B} \ra \\[.5em]
		= ~ &  \underbrace{\frac{2n}{n-1} \times  \frac{1}{n}\sum_{i=1}^n  \la \widetilde{\phi}(X_i)\widetilde{\psi}(Y_i), f_{2,B} \ra \la \widetilde{\phi}(X_i) \biggl(\frac{1}{n}\sum_{j=1}^n\widetilde{\psi}(Y_j)\biggr), f_{2,B} \ra}_{O_P(n^{-3/2})} \\
		& -  \underbrace{\frac{2}{n(n-1)}\sum_{i=1}^n \la \widetilde{\phi}(X_i)\widetilde{\psi}(Y_i), f_{2,B} \ra^2}_{O_P(n^{-2})},
	\end{align}
	where the second approximation $O_P(n^{-2})$ uses \eqref{Eq: approximation f2B}. The first approximation follows by the Cauchy--Schwarz inequality and Markov's inequality as 
	\begin{align}
		&\bigg| \frac{1}{n}\sum_{i=1}^n  \la \widetilde{\phi}(X_i)\widetilde{\psi}(Y_i), f_{2,B} \ra \la \widetilde{\phi}(X_i) \biggl(\frac{1}{n}\sum_{j=1}^n\widetilde{\psi}(Y_j)\biggr), f_{2,B} \ra\bigg| \\
		\leq ~ & \sqrt{ \frac{1}{n}\sum_{i=1}^n  \la \widetilde{\phi}(X_i)\widetilde{\psi}(Y_i), f_{2,B} \ra^2 }  \sqrt{ \frac{1}{n}\sum_{i=1}^n  \la \widetilde{\phi}(X_i) \biggl(\frac{1}{n}\sum_{j=1}^n\widetilde{\psi}(Y_j)\biggr), f_{2,B} \ra^2 } \\
		= ~ & O_P(n^{-1}) O_P(n^{-1}) = O_P(n^{-2}).
	\end{align}
	Therefore $J_{(2),b} = o_P(n^{-1})$. 
	For the last term $J_{(2),c}$, 
	\begin{align}
		J_{(2),c} ~=~& \frac{2}{n(n-1)} \sum_{1 \leq i \neq j \leq n}  \la \widetilde{\phi}(X_i)\widetilde{\psi}(Y_i), f_{2,A} \ra \la \widetilde{\phi}(X_i)\widetilde{\psi}(Y_j), f_{2,B} \ra \\
		= ~ &  \underbrace{\frac{2n}{n-1} \times  \frac{1}{n}\sum_{i=1}^n  \la \widetilde{\phi}(X_i)\widetilde{\psi}(Y_i), f_{2,A} \ra \la \widetilde{\phi}(X_i) \biggl(\frac{1}{n}\sum_{j=1}^n\widetilde{\psi}(Y_j)\biggr), f_{2,B} \ra}_{O_P(n^{-3/2})} \\
		& -  \underbrace{\frac{2}{n(n-1)}\sum_{i=1}^n \la \widetilde{\phi}(X_i)\widetilde{\psi}(Y_i), f_{2,A} \ra \la \widetilde{\phi}(X_i)\widetilde{\psi}(Y_i), f_{2,B} \ra}_{O_P(n^{-2})},
	\end{align}
	which can be established by the Cauchy--Schwarz inequality and the previous results, and thus $J_{(2),c} = o_P(n^{-1})$. In summary, it holds that $J_{(2)} = o_P(n^{-1})$.
	
	\bigskip 
	
	\noindent \textbf{3.~Analyzing $J_{(3)}$.} By symmetry, $J_{(3)}$ has the same convergence rate as $J_{(2)}$ and thus $J_{(3)} = o_P(n^{-1})$. 
	
	\bigskip 
	
	\noindent \textbf{4.~Analyzing $J_{(4)}$.} For the fourth term $J_{(4)}$, we will show that $J_{(4)} = o_P(n^{-1})$. To this end, we only need to handle three terms
	\begin{align}
		J_{(4),a} := &  \frac{3}{n(n-1)} \sum_{1 \leq i \neq j \leq n}  \la \widetilde{\phi}(X_i)\widetilde{\psi}(Y_i), f_{2,A} \ra \la \widetilde{\phi}(X_j)\widetilde{\psi}(Y_j), f_{2,A} \ra , \\
		J_{(4),b} := &  \frac{3}{n(n-1)} \sum_{1 \leq i \neq j \leq n}  \la \widetilde{\phi}(X_i)\widetilde{\psi}(Y_i), f_{2,B} \ra \la \widetilde{\phi}(X_j)\widetilde{\psi}(Y_j), f_{2,B} \ra  \quad \text{and} \\
		J_{(4),c} :=  &  \frac{3}{n(n-1)} \sum_{1 \leq i \neq j \leq n}  \la \widetilde{\phi}(X_i)\widetilde{\psi}(Y_i), f_{2,A} \ra \la \widetilde{\phi}(X_j)\widetilde{\psi}(Y_j), f_{2,B} \ra .
	\end{align}
	Starting with $J_{(4),a}$,
	\begin{align}
		J_{(4),a} = \underbrace{\frac{3n^2}{n(n-1)}  \la \frac{1}{n} \sum_{i=1}^n \widetilde{\phi}(X_i)\widetilde{\psi}(Y_i), f_{2,A} \ra^2}_{O_P(n^{-2})} - \underbrace{\frac{3}{n(n-1)} \sum_{i=1}^n  \la \widetilde{\phi}(X_i)\widetilde{\psi}(Y_i), f_{2,A} \ra^2}_{O_P(n^{-2})},
	\end{align}
	where the second approximation is due to \eqref{Eq: approximation f2A} and the first approximation is by Markov's inequality. Indeed, we have shown in Appendix~\ref{Section: Proof of Step 1 (Numerator)} and \ref{Section: Proof of Step 3} that 
	\begin{align}
		n \la \frac{1}{n} \sum_{i=1}^n \widetilde{\phi}(X_i)\widetilde{\psi}(Y_i), f_{2,A} \ra \convdist \sum_{k=1}^\infty \lambda_k W_k \widetilde{W}_k . 
	\end{align}	
	This establishes $J_{(4),a} = o_P(n^{-1})$.
	
	Next, for the second term $J_{(4),b}$, 
		\begin{align}
		J_{(4),b} = \underbrace{\frac{3n^2}{n(n-1)}  \la \frac{1}{n} \sum_{i=1}^n \widetilde{\phi}(X_i)\widetilde{\psi}(Y_i), f_{2,B} \ra^2}_{O_P(n^{-2})} - \underbrace{\frac{3}{n(n-1)} \sum_{i=1}^n  \la \widetilde{\phi}(X_i)\widetilde{\psi}(Y_i), f_{2,B} \ra^2}_{O_P(n^{-3})},
	\end{align}
	where the second approximation uses \eqref{Eq: approximation f2B}. For the first approximation, we simply use Jensen's inequality along with \eqref{Eq: approximation f2B}:
	\begin{align}
		\la \frac{1}{n} \sum_{i=1}^n \widetilde{\phi}(X_i)\widetilde{\psi}(Y_i), f_{2,B} \ra^2 \leq \frac{1}{n} \sum_{i=1}^n \la \widetilde{\phi}(X_i)\widetilde{\psi}(Y_i), f_{2,B} \ra^2 =O_P(n^{-2}).
	\end{align}
	Thus we have $J_{(4),b} = o_P(n^{-1})$.
	
	For the last term $J_{(4),c}$,
	\begin{align}
		J_{(4),c} = &\underbrace{\frac{3n^2}{n(n-1)}  \la \frac{1}{n} \sum_{i=1}^n \widetilde{\phi}(X_i)\widetilde{\psi}(Y_i), f_{2,A} \ra \la \frac{1}{n} \sum_{i=1}^n \widetilde{\phi}(X_i)\widetilde{\psi}(Y_i), f_{2,B} \ra}_{O_P(n^{-2})} \\
		& - \underbrace{\frac{3}{n(n-1)} \sum_{i=1}^n  \la \widetilde{\phi}(X_i)\widetilde{\psi}(Y_i), f_{2,A} \ra \la \widetilde{\phi}(X_i)\widetilde{\psi}(Y_i), f_{2,B} \ra}_{O_P(n^{-2})},
	\end{align}
	which can be seen using the Cauchy--Schwarz inequality. Therefore it holds that $J_{(4)} = o_P(n^{-1})$.
	
	\bigskip 
	
	\noindent \textbf{5.~Analyzing $J_{(5)}$.} For the fifth term $J_{(5)}$, similarly as before, we need to study
	\begin{align}
		& J_{(5),a} :=\frac{1}{n(n-1)(n-2)}\sum_{\substack{1\leq i,j,q \leq n \\ \text{$i,j,q$ distinct}}} \la \widetilde{\phi}(X_i)\widetilde{\psi}(Y_i), f_{2,A} \ra \la \widetilde{\phi}(X_j)\widetilde{\psi}(Y_q), f_{2,A} \ra, \\
		& J_{(5),b} :=\frac{1}{n(n-1)(n-2)}\sum_{\substack{1\leq i,j,q \leq n \\ \text{$i,j,q$ distinct}}} \la \widetilde{\phi}(X_i)\widetilde{\psi}(Y_i), f_{2,B} \ra \la \widetilde{\phi}(X_j)\widetilde{\psi}(Y_q), f_{2,B} \ra \quad \text{and} \\
		& J_{(5),c} :=\frac{1}{n(n-1)(n-2)}\sum_{\substack{1\leq i,j,q \leq n \\ \text{$i,j,q$ distinct}}} \la \widetilde{\phi}(X_i)\widetilde{\psi}(Y_i), f_{2,A} \ra \la \widetilde{\phi}(X_j)\widetilde{\psi}(Y_q), f_{2,B} \ra.
	\end{align}
	Starting with $J_{(5),a}$, note that 
	\begin{align}
		\sum_{\substack{1\leq i,j,q \leq n \\ \text{$i,j,q$ distinct}}} \la \widetilde{\phi}(X_i)\widetilde{\psi}(Y_i), f_{2,A} \ra \la \widetilde{\phi}(X_j)\widetilde{\psi}(Y_q), f_{2,A} \ra = &  \la \sum_{i=1}^n \widetilde{\phi}(X_i)\widetilde{\psi}(Y_i), f_{2,A} \ra \la \sum_{1 \leq j \neq q \leq n} \widetilde{\phi}(X_j)\widetilde{\psi}(Y_q), f_{2,A} \ra \\
		 -  2 & \sum_{1 \leq j \neq q \leq n} \la \widetilde{\phi}(X_j)\widetilde{\psi}(Y_j), f_{2,A} \ra \la \widetilde{\phi}(X_j)\widetilde{\psi}(Y_q), f_{2,A} \ra.
	\end{align}
	Thus 
	\begin{align}
		J_{(5),a} ~=~ & O(1) \times \underbrace{\la \frac{1}{n}\sum_{i=1}^n \widetilde{\phi}(X_i)\widetilde{\psi}(Y_i), f_{2,A} \ra}_{O_P(n^{-1})} \underbrace{\la \frac{1}{n(n-1)}\sum_{1 \leq j \neq q \leq n} \widetilde{\phi}(X_j)\widetilde{\psi}(Y_q), f_{2,A} \ra}_{O_P(n^{-1})} \\
		+ ~& O(n^{-1}) \times \underbrace{\frac{1}{n(n-1)}   \sum_{1 \leq j \neq q \leq n} \la \widetilde{\phi}(X_j)\widetilde{\psi}(Y_j), f_{2,A} \ra \la \widetilde{\phi}(X_j)\widetilde{\psi}(Y_q), f_{2,A} \ra}_{o_P(n^{-1})},
	\end{align}
	where the first approximation holds since 
	\begin{align}
		n \times \la \frac{1}{n}\sum_{i=1}^n \widetilde{\phi}(X_i)\widetilde{\psi}(Y_i), f_{2,A} \ra \convdist \sum_{k=1}^K \lambda_k W_k \widetilde{W}_k,
	\end{align}
	as established in Appendix~\ref{Section: Proof of Step 3}. For the second approximation, we have
	\begin{align}
		\la \frac{1}{n(n-1)}\sum_{1 \leq j \neq q \leq n} \widetilde{\phi}(X_j)\widetilde{\psi}(Y_q), f_{2,A} \ra ~=~ & O(1) \times \la \biggl( \frac{1}{n} \sum_{i=1}^n \widetilde{\phi}(X_i) \biggr) \biggl( \frac{1}{n} \sum_{i=1}^n \widetilde{\psi}(Y_i)\biggr), f_{2,A}  \ra \\
		& - O(n^{-1}) \times  \underbrace{\la \frac{1}{n}\sum_{i=1}^n \widetilde{\phi}(X_i)\widetilde{\psi}(Y_i), f_{2,A} \ra}_{O_P(n^{-1})}. 
	\end{align}
	Moreover the following term
	\begin{align}
		\la \biggl( \frac{1}{n} \sum_{i=1}^n \widetilde{\phi}(X_i) \biggr) \biggl( \frac{1}{n} \sum_{i=1}^n \widetilde{\psi}(Y_i)\biggr), f_{2,A}  \ra
	\end{align}
	has the same convergence rate as
	\begin{align}
		\Bigg| \frac{1}{n} \sum_{i=1}^n \la  \widetilde{\phi}(X_i)\widetilde{\psi}(Y_i), f_{2,B} \ra \Bigg| \leq \sqrt{\frac{1}{n} \sum_{i=1}^n \la  \widetilde{\phi}(X_i)\widetilde{\psi}(Y_i), f_{2,B} \ra^2} = O_P(n^{-1}).
	\end{align}
	Thereby, the second approximation holds. The last approximation was established in the analysis of $J_{(2),a}$, and thus $J_{(5),a} = o_P(n^{-1})$. 
	
	For the second term $J_{(5),b}$, we simply use the Cauchy--Schwarz inequality and Markov's inequality, and see
	\begin{align}
		|J_{(5),b}| ~\leq~ O(1) \underbrace{\sqrt{\frac{1}{n} \sum_{i=1}^n \la \widetilde{\phi}(X_i)\widetilde{\psi}(Y_i), f_{2,B} \ra^2}}_{O_P(n^{-1})} \underbrace{\sqrt{\frac{1}{n^2} \sum_{i=1}^n\sum_{j=1}^n \la \widetilde{\phi}(X_i)\widetilde{\psi}(Y_j), f_{2,B} \ra^2}}_{O_P(n^{-1})} =o_P(n^{-1}).
	\end{align}
  	Similarly, for the third term $J_{(5),c}$, we apply the Cauchy--Schwarz inequality and see
  	\begin{align}
  		|J_{(5),c}| ~ \leq~  O(1) \underbrace{\sqrt{\frac{1}{n} \sum_{i=1}^n \la \widetilde{\phi}(X_i)\widetilde{\psi}(Y_i), f_{2,A} \ra^2}}_{O_P(n^{-1/2})} \underbrace{\sqrt{\frac{1}{n^2} \sum_{i=1}^n\sum_{j=1}^n \la \widetilde{\phi}(X_i)\widetilde{\psi}(Y_j), f_{2,B} \ra^2}}_{O_P(n^{-1})} =o_P(n^{-1}).
  	\end{align}
  	Therefore we conclude $J_{(5)} = o_P(n^{-1})$.
  	
  	\bigskip 
  	
  	\noindent \textbf{6.~Analyzing $J_{(6)}$.} For the sixth term $J_{(6)}$, similarly as before, we need to study
  	\begin{align}
  		& J_{(6),a} :=\frac{1}{n(n-1)(n-2)}\sum_{\substack{1\leq i,j,q \leq n \\ \text{$i,j,q$ distinct}}} \la \widetilde{\phi}(X_i)\widetilde{\psi}(Y_j), f_{2,A} \ra \la \widetilde{\phi}(X_i)\widetilde{\psi}(Y_q), f_{2,A} \ra, \\
  		& J_{(6),b} :=\frac{1}{n(n-1)(n-2)}\sum_{\substack{1\leq i,j,q \leq n \\ \text{$i,j,q$ distinct}}} \la \widetilde{\phi}(X_i)\widetilde{\psi}(Y_j), f_{2,B} \ra \la \widetilde{\phi}(X_i)\widetilde{\psi}(Y_q), f_{2,B} \ra \quad \text{and} \\
  		& J_{(6),c} :=\frac{1}{n(n-1)(n-2)}\sum_{\substack{1\leq i,j,q \leq n \\ \text{$i,j,q$ distinct}}} \la \widetilde{\phi}(X_i)\widetilde{\psi}(Y_j), f_{2,A} \ra \la \widetilde{\phi}(X_i)\widetilde{\psi}(Y_q), f_{2,B} \ra.
  	\end{align}
  	Starting with $J_{(6),a}$, there exist some constants $C_1,\ldots,C_4$ such that 
  	\begin{align}
  		& \sum_{\substack{1\leq i,j,q \leq n \\ \text{$i,j,q$ distinct}}} \la \widetilde{\phi}(X_i)\widetilde{\psi}(Y_j), f_{2,A} \ra \la \widetilde{\phi}(X_i)\widetilde{\psi}(Y_q), f_{2,A} \ra \\
  		~=~ & \sum_{i=1}^n \la \widetilde{\phi}(X_i) \sum_{j=1}^n \widetilde{\psi}(Y_j), f_{2,A} \ra \la \widetilde{\phi}(X_i) \sum_{q=1}^n \widetilde{\psi}(Y_q), f_{2,A} \ra \\[.5em]
  		-  &C_1  \sum_{1 \leq i \neq j \leq n} \la \widetilde{\phi}(X_i)\widetilde{\psi}(Y_j), f_{2,A} \ra \la \widetilde{\phi}(X_i)\widetilde{\psi}(Y_i), f_{2,A} \ra - C_2 \sum_{i=1}^n \la \widetilde{\phi}(X_i)\widetilde{\psi}(Y_i), f_{2,A} \ra^2 \\[.5em]
  		- & C_3  \sum_{1 \leq i \neq q \leq n}  \la \widetilde{\phi}(X_i)\widetilde{\psi}(Y_i), f_{2,A} \ra \la \widetilde{\phi}(X_i)\widetilde{\psi}(Y_q), f_{2,A} \ra - C_4\sum_{1 \leq i \neq q \leq n}  \la \widetilde{\phi}(X_i)\widetilde{\psi}(Y_q), f_{2,A} \ra^2.
  	\end{align}
  	By the Cauchy--Schwarz inequality and Markov's inequality,
  	\begin{align}
  		 &\Bigg| \frac{1}{n}\sum_{i=1}^n \la \widetilde{\phi}(X_i) \frac{1}{n}\sum_{j=1}^n \widetilde{\psi}(Y_j), f_{2,A} \ra \la \widetilde{\phi}(X_i) \biggl(\frac{1}{n}\sum_{q=1}^n \widetilde{\psi}(Y_q)\biggr), f_{2,A} \ra \Bigg|^2 \\
  		 \leq ~ &  \underbrace{\frac{1}{n}\sum_{i=1}^n \la \widetilde{\phi}(X_i) \frac{1}{n}\sum_{j=1}^n \widetilde{\psi}(Y_j), f_{2,A} \ra^2}_{O_P(n^{-2})}  \underbrace{\frac{1}{n}\sum_{i=1}^n  \la \widetilde{\phi}(X_i) \biggl(\frac{1}{n}\sum_{q=1}^n \widetilde{\psi}(Y_q)\biggr), f_{2,A} \ra^2}_{O_P(n^{-2})}.
  	\end{align}
  	Thus
  	\begin{align}
  		  \frac{1}{n^3}\sum_{i=1}^n \la \widetilde{\phi}(X_i) \sum_{j=1}^n \widetilde{\psi}(Y_j), f_{2,A} \ra \la \widetilde{\phi}(X_i) \sum_{q=1}^n \widetilde{\psi}(Y_q), f_{2,A} \ra = O_P(n^{-2}).
  	\end{align}
  	Based on the previous results, we also have
  	\begin{align}
  		& \frac{1}{n^3} \sum_{1 \leq i \neq j \leq n} \la \widetilde{\phi}(X_i)\widetilde{\psi}(Y_j), f_{2,A} \ra \la \widetilde{\phi}(X_i)\widetilde{\psi}(Y_i), f_{2,A} \ra = o_P(n^{-2}), \\
  		& \frac{1}{n^3} \sum_{i=1}^n \la \widetilde{\phi}(X_i)\widetilde{\psi}(Y_i), f_{2,A} \ra^2 = o_P(n^{-2}), \\
  		& \frac{1}{n^3} \sum_{1 \leq i \neq q \leq n}  \la \widetilde{\phi}(X_i)\widetilde{\psi}(Y_i), f_{2,A} \ra \la \widetilde{\phi}(X_i)\widetilde{\psi}(Y_q), f_{2,A} \ra = o_P(n^{-2}), \\
  		& \frac{1}{n^3} \sum_{1 \leq i \neq q \leq n}  \la \widetilde{\phi}(X_i)\widetilde{\psi}(Y_q), f_{2,A} \ra^2 = o_P(n^{-2}),
  	\end{align}
  	which concludes $J_{(6),a} = o_P(n^{-1})$. In addition $J_{(6),b}$ and $J_{(6),c}$ are shown to be $o_P(n^{-1})$ by the Cauchy--Schwarz inequality as in the analysis of $J_{(5),b}$ and $J_{(5),c}$. This proves that $J_{(6)} = o_P(n^{-1})$.
  	
  	\bigskip 
  	
  	\noindent \textbf{7.~Analyzing $J_{(7)}$.}  For the sixth term $J_{(7)}$, similarly as before, we need to study
  	\begin{align}
  		& J_{(7),a} :=\frac{1}{n(n-1)(n-2)}\sum_{\substack{1\leq i,j,q \leq n \\ \text{$i,j,q$ distinct}}} \la \widetilde{\phi}(X_i)\widetilde{\psi}(Y_j), f_{2,A} \ra \la \widetilde{\phi}(X_q)\widetilde{\psi}(Y_i), f_{2,A} \ra, \\
  		& J_{(7),b} :=\frac{1}{n(n-1)(n-2)}\sum_{\substack{1\leq i,j,q \leq n \\ \text{$i,j,q$ distinct}}} \la \widetilde{\phi}(X_i)\widetilde{\psi}(Y_j), f_{2,B} \ra \la \widetilde{\phi}(X_q)\widetilde{\psi}(Y_i), f_{2,B} \ra \quad \text{and} \\
  		& J_{(7),c} :=\frac{1}{n(n-1)(n-2)}\sum_{\substack{1\leq i,j,q \leq n \\ \text{$i,j,q$ distinct}}} \la \widetilde{\phi}(X_i)\widetilde{\psi}(Y_j), f_{2,A} \ra \la \widetilde{\phi}(X_q)\widetilde{\psi}(Y_i), f_{2,B} \ra.
  	\end{align}
  	Starting with $J_{(7),a}$, there exist some constants $C_1',\ldots,C_4'$ such that 
  	\begin{align}
  		& \sum_{\substack{1\leq i,j,q \leq n \\ \text{$i,j,q$ distinct}}} \la \widetilde{\phi}(X_i)\widetilde{\psi}(Y_j), f_{2,A} \ra \la \widetilde{\phi}(X_q)\widetilde{\psi}(Y_i), f_{2,A} \ra \\
  		~=~ & \sum_{i=1}^n \la \widetilde{\phi}(X_i) \sum_{j=1}^n \widetilde{\psi}(Y_j), f_{2,A} \ra \la  \sum_{q=1}^n \widetilde{\phi}(X_q) \widetilde{\psi}(Y_i), f_{2,A} \ra \\[.5em]
  		- & C_1' \sum_{1 \leq i \neq j \leq n} \la \widetilde{\phi}(X_i)\widetilde{\psi}(Y_j), f_{2,A} \ra \la \widetilde{\phi}(X_i)\widetilde{\psi}(Y_i), f_{2,A} \ra - C_2' \sum_{i=1}^n \la \widetilde{\phi}(X_i)\widetilde{\psi}(Y_i), f_{2,A} \ra^2 \\[.5em]
  		- & C_3' \sum_{1 \leq i \neq q \leq n}  \la \widetilde{\phi}(X_i)\widetilde{\psi}(Y_i), f_{2,A} \ra \la \widetilde{\phi}(X_q)\widetilde{\psi}(Y_i), f_{2,A} \ra - C_4'\sum_{1 \leq i \neq q \leq n}  \la \widetilde{\phi}(X_i)\widetilde{\psi}(Y_q), f_{2,A} \ra^2\Biggr].
  	\end{align}
  	 From here, we can follow the same steps as in the analysis for $J_{(6)}$ and conclude that $J_{(7)} = o_P(n^{-1})$.
  	
  	\bigskip 
  	
  	\noindent \textbf{8.~Analyzing $J_{(8)}$.} By switching the role between $X$ and $Y$, $J_{(8)}$ can be analyzed similarly as $J_{(7)}$ and it can be shown that $J_{(8)} = o_P(n^{-1})$. 
  	
  	\bigskip 
  	
  	\noindent \textbf{9.~Analyzing $J_{(9)}$.}  By switching the role between $X$ and $Y$, $J_{(9)}$ can be analyzed similarly as $J_{(6)}$ and it can be shown that $J_{(9)} = o_P(n^{-1})$. 
  	
  	\bigskip
  	
  	Throughout, we have shown that $\sum_{i=2}^9 J_{(i)} = o_P(n^{-1})$ and $J_{(1)}$ satisfies \eqref{Eq: J1 approximation}, which concludes approximation~\eqref{Eq: s_n approximation}.


\section[Uniform Asymptotic Null Distribution]{Uniform Asymptotic Null Distribution~\texorpdfstring{(\Cref{theorem:null-dist-1})}{(Theorem~7)}}
    \label{proof:null-distribution}
    We first describe an outline of the general steps of the proof in~\Cref{appendix:outline-null-1}, breaking down the argument into three parts. Then, we present the details of the steps in the subsequent subsections. 
        
    Before proceeding, we recall some notation.  We use $\tildea_{ij} = \tildek(X_i, \cdot) \tildel(Y_j, \cdot)$,  $\tildeb_{tu} = \tildea_{n+t, n+u}$ for $1\leq i, j, t, u \leq n$, and  
    $\tildeg_{12}$ to denote $\tildeg(Z_1, Z_2) = \la \tildea_{11}, \tildea_{22} \ra$. 
        
    \subsection{Outline of the proof}   
        \label{appendix:outline-null-1}
          
            To show the asymptotic normality of~$\cshsic$, we verify that the sufficient conditions for the Berry--Esseen theorem for studentized U-statistics are satisfied. In particular, by \citet[Theorem 3.1]{jing2000berry}, it suffices to show that 
            \begin{align}
            \label{eq:null-outline-1}
                C_n \defined \frac{ \mathbb{E}[|\la h(Z_1, Z_2), f_2 \ra|^3|\datatwo]}
                {\sqrt{n} \sigma_g^3} \convprob 0,
            \end{align}
            where we have used the notation $\sigma_g^2 \defined \mathbb{E}[\{\mathbb{E}[\la h(Z_1,Z_2), f_2 \ra |Z_2,\datatwo]\}^2|\datatwo]$ following~\citet{jing2000berry}. 

            To establish this result, we proceed in the following steps: 
            \begin{itemize}
                \item Step 1: First, we observe in~\Cref{lemma:null-distribution-1}, that a sufficient condition for~\eqref{eq:null-outline-1} to hold is if the following holds: 
                    \begin{align}
                    \label{eq:null-outline-2}
                        B_n \defined \frac{\mathbb{E}[\langle f_2, \tildek(X_{1,n}, \cdot) \tildel(Y_{1,n}, \cdot) \rangle^4|\datatwo]}{n \big\{  \mathbb{E}[\langle f_2, \tildek(X_{1,n}, \cdot) \tildel(Y_{1,n}, \cdot) \rangle^2|\datatwo]\big\}^2}  \convprob 0,
                    \end{align}
                \item Step 2: In the next step, we introduce the term $B_{1,n}$ defined as 
                    \begin{align}
                         B_{1,n} \defined \frac{n\mathbb{E}[\langle f_2, \tildek(X_{1,n}, \cdot) \tildel(Y_{1,n}, \cdot) \rangle^4|\datatwo]}{\mathbb{E}[\tildeg(Z_{1,n}, Z_{2,n})^2]^2}, \label{eq:null-outline-3}
                    \end{align}
                    and show that $B_{1,n} \convprob 0$ in~\Cref{lemma:null-distribution-2}. 
                \item Step 3: Finally, in we introduce the term $B_{2,n}$ defined as 
                    \begin{align}
                        B_{2,n} \defined \frac{\mathbb{E}[\tildeg(Z_{1,n}, Z_{2,n})^2]^2}{n^2 \big\{  \mathbb{E}[\langle f_2, \tildek(X_{1,n}, \cdot) \tildel(Y_{1,n}, \cdot) \rangle^2|\datatwo]\big\}^2}, \label{eq:null-outline-4}
                    \end{align}
                    and show that $B_{2,n} = \mc{O}_P(1)$ in~\Cref{lemma:null-distribution-3}
            \end{itemize}
            Since $B_n = B_{1,n} \times B_{2,n}$, together~\eqref{eq:null-outline-3} and~\eqref{eq:null-outline-4} imply~\eqref{eq:null-outline-2}, to complete the proof. The details are in~\Cref{proof:null-distribution}. 

    \subsection[Proof of Step 1]{Proof of Step 1}
        
        \begin{lemma}
            \label{lemma:null-distribution-1}        
            Introduce the term $B_n \defined \frac{\mathbb{E}[\langle f_2, \tildek(X_{1,n}, \cdot) \tildel(Y_{1,n}, \cdot) \rangle^4|\datatwo]}{n \big\{  \mathbb{E}[\langle f_2, \tildek(X_{1,n}, \cdot) \tildel(Y_{1,n}, \cdot) \rangle^2|\datatwo]\big\}^2}$. Then, we have the following: 
            \begin{align}
                B_n \convprob 0 \quad \text{implies}\quad  C_n \convprob 0. 
            \end{align}
        \end{lemma}

        \begin{proof}
                We begin by noting that due to Cauchy--Schwarz inequality,  we have 
            \begin{align}
                \mathbb{E}[|\la h(Z_1,Z_2), f_2 \ra|^3|\datatwo] &= \mathbb{E}[|\la h(Z_1,Z_2), f_2 \ra|\times  |\la h(Z_1,Z_2), f_2 \ra|^2 |\datatwo] \\
                &\leq \big\{ \mathbb{E}[\la h(Z_1,Z_2), f_2 \ra^2|\datatwo] \big\}^{1/2} \big\{ \mathbb{E}[\langle h(Z_1,Z_2), f_2\rangle^4|\datatwo] \big\}^{1/2}.
            \end{align}
            Further note 
            \begin{align}
                \mathbb{E}[\la h(Z_1,Z_2), f_2 \ra^2|\datatwo] ~=~ & \frac{1}{4}\mathbb{E}[\langle \tildef_2, \tildea_{11} +\tildea_{22} - \tildea_{12} - \tildea_{21} \rangle^2|\datatwo] \\
                \lesssim ~& \mathbb{E}[\langle \tildef_2, \tildea_{11} \rangle^2 + \langle \tildef_2, \tildea_{22} \rangle^2 + \langle \tildef_2, \tildea_{12} \rangle^2 + \langle \tildef_2, \tildea_{21} \rangle^2|\datatwo] \\
                \lesssim ~& \mathbb{E}[\langle \tildef_2, \tildea_{11} \rangle^2|\datatwo], 
                \label{eq:proof-null-1}
            \end{align}
            where the first inequality  uses Jensen's inequality, while the second inequality relies on  the observation under the null:
            \begin{align}
                \mathbb{E}[\langle \tildef_2, \tildea_{11} \rangle^2|\datatwo] = \mathbb{E}[\langle \tildef_2, \tildea_{ij} \rangle^2|\datatwo] \quad \text{for any $i,j \in \{1,2\}$.}
            \end{align}
            By the same logic along with Jensen's inequality, 
            \begin{align}
               \mathbb{E}[\la h(Z_1,Z_2), f_2 \ra^4|\datatwo] 
               \lesssim~& \mathbb{E}[\la f_2, \tildea_{11} \ra^4+ \la f_2, \tildea_{22} \ra^4 + \la f_2, \tildea_{12} \ra^4 + \la f_2, \tildea_{21} \ra^4|\datatwo]\\
               \lesssim~& \mathbb{E}[\langle \tildef_2, \tildea_{11} \rangle^4|\datatwo]. \label{eq:proof-null-2}
            \end{align}
            Thus, combining~\eqref{eq:proof-null-1} and~\eqref{eq:proof-null-2}, we get the following bound on the numerator of the term $C_n$: 
            \begin{align}
                \lp \mathbb{E}[|\la h(Z_1, Z_2), f_2 \ra|^3|\datatwo] \rp^2 \leq  \mathbb{E}[\langle \tildef_2, \tildea_{11} \rangle^2|\datatwo] \times \mathbb{E}[\langle \tildef_2, \tildea_{11} \rangle^4|\datatwo] \label{eq:proof-null-3}
            \end{align}
            
            We now evaluate $\sigma_g^2$ from the denominator of $C_n$.
            \begin{align}
                \sigma_g^2 ~=~ \frac{1}{4} \mathbb{E}[\langle \tildef_2, \tildea_{11} \rangle^2|\datatwo]. \label{eq:proof-null-4}
            \end{align}
            Combining the pieces, we obtain the following: 
            \begin{align}
                C_n =~& \frac{ \mathbb{E}[|\la h(Z_1, Z_2), f_2 \ra|^3|\datatwo]}{\sqrt{n}\sigma_g^3} 
                \leq ~ \frac{\lp \mathbb{E}[\langle \tildef_2, \tildea_{11} \rangle^2|\datatwo]\rp^{1/2} \times \lp \mathbb{E}[\langle \tildef_2, \tildea_{11} \rangle^4|\datatwo] \rp^{1/2}}{\sqrt{n} \sigma_g^3} \\
                = ~& 
                \lp \frac{\mathbb{E}[\langle \tildef_2, \tildea_{11} \rangle^4|\datatwo]}{n \big\{  \mathbb{E}[\langle \tildef_2, \tildea_{11} \rangle^2|\datatwo]\big\}^2} \rp^{1/2} = \sqrt{B_n}.
            \end{align}
            In the first inequality above, we used~\eqref{eq:proof-null-3}, while the second equality uses~\eqref{eq:proof-null-4}. This completes the proof. 
        \end{proof}
        
    \subsection{Proof of Step 2}
        \begin{lemma}
            \label{lemma:null-distribution-2}
            Under~\Cref{assump:null-main}, we have $B_{1,n} \convprob 0$. 
        \end{lemma}
        
         \begin{proof}

                It suffices to show that $\mathbb{E}[B_{1,n}] \to 0$, which in turn will imply the convergence in probability by an application of Markov's inequality.
                To verify this, we observe the following: 
                \begin{align}
                \mathbb{E}[B_{1,n}] &= \frac{n}{\mathbb{E}[\tildeg_{12}^2]^2} \mc{O}\lp \mathbb{E}\lb \left\langle \tildea_{ii}, \frac{1}{n} \sum_{t=1}^n \tildeb_{tt} - \frac{1}{n(n-1)} \sum_{t \neq u} \tildeb_{tu} \right \rangle^4 \rb \rp \\
                &\leq \frac{16 n}{\mathbb{E}[\tildeg_{12}^2]^2} \lp \mathbb{E}\lb \left \langle \tildea_{ii}, \; \frac{1}{n} \sum_{t=1}^n \tildeb_{tt}\right  \rangle ^4 \rb + \mathbb{E}\lb\left  \langle \tildea_{ii}, \; \frac{1}{n(n-1)} \sum_{t \neq u} \tildeb_{tu} \right \rangle ^4 \rb \rp \label{eq:proof-lyapunav-1}\\
                & = \frac{16n}{\mathbb{E}[\tildeg_{12}^2]^2} \mc{O}\lp
                \frac{1}{n^4} \mathbb{E}\lb \lp \sum_{t=1}^n \langle \tildea_{ii}, \tildeb_{tt} \rangle \rp^4 \rb  + \frac{1}{n^8} \mathbb{E} \lb \lp \sum_{u\neq t} \langle \tildea_{ii}, \tildeb_{tu} \rangle \rp^4 \rb
                \rp \label{eq:proof-lyapunav-2} \\
                & = \frac{16n}{\mathbb{E}[\tildeg_{12}^2]^2} \mc{O}\lp \mathbb{E}[\tildeg_{12}^4]\lp \frac{1}{n^3}+ \frac{1}{n^6} \rp + \mathbb{E}[\tildeg_{12}^2 \tildeg_{13}^2] \lp \frac{1}{n^2} + \frac{1}{n^4} \rp \rp  \label{eq:proof-lyapunav-3}\\ 
                & =  \mc{O}\lp  \frac{\mathbb{E}[\tildeg_{12}^4]n^{-1} + \mathbb{E}[\tildeg_{12}^2 \tildeg_{13}^2] }{n\mathbb{E}[\tildeg_{12}^2]^2} \rp \to 0.\label{eq:proof-lyapunav-4}
                \end{align}
                In the above display, \eqref{eq:proof-lyapunav-1} uses the fact that $(x-y)^4 \leq 16(x^4 + y^4)$. To obtain \eqref{eq:proof-lyapunav-3}, we note that $\mathbb{E}[(\sum_{t=1}^n \langle \tildea_{ii}, \tildeb_{tt}\rangle^4]$ can be expanded into $n^4$ terms, each of the form $\mathbb{E}[\tildeg_{it_1}\tildeg_{it_2}\tildeg_{it_3}\tildeg_{it_4}]$, for $1 \leq t_1, t_2, t_3, t_4 \leq n$. Of these $n^4$ terms, only the terms with two or four common $t's$ have non-zero expectation under the null. There are a total of $n$ terms of the form $\mathbb{E}[\tildeg_{it}^4]$ and $\mc{O}(n^2)$ terms of the form $\mathbb{E}[\tildeg_{it}^2\tildeg_{iu}^2]$ for $t \neq u$. Such terms appear as $\mathbb{E}[\tildeg_{it}^4]/n^3$ and $\mathbb{E}[\tildeg_{it}^2\tildeg_{iu}^2]/n^2$ respectively in~\eqref{eq:proof-lyapunav-3}. Repeating the same argument for $\mathbb{E}[\sum_{t\neq u} \langle \tildea_{ii}, \tildeb_{tu}\rangle^4]$ gives us the other two terms in~\eqref{eq:proof-lyapunav-3}.

        \end{proof}
   
    \subsection{Proof of Step 3}
        \begin{lemma}
            \label{lemma:null-distribution-3}
            Under~\Cref{assump:null-main}, we have $B_{2,n} = \mc{O}_P(1)$. 
        \end{lemma}

        \begin{proof}
     Introduce the notation $f_{21} = \frac{1}{n}\sum_{t=1}^n \tildeb_{tt}$ and $f_{22} = \frac{1}{n(n-1)} \sum_{t=1}^n \sum_{u \neq t} \tildeb_{tu}$, and observe that $f_2 = f_{21} - f_{22}$. Next, we define the following terms: 
            \begin{align}
                B_{3,n} = \frac{1}{\sqrt{B_{2,n}}} = \frac{n  \mathbb{E}[\langle \tildea_{ii}, f_2 \rangle^2| \mathcal{D}_{n+1}^{2n}]}{\mathbb{E}[\tildeh_{it}^2]}, \quad \text{and} \quad  B_{4,n} = \frac{n  \mathbb{E}[\langle \tildea_{ii}, f_{21} \rangle^2| \mathcal{D}_{n+1}^{2n}]}{\mathbb{E}[\tildeh_{it}^2]}
            \end{align}
            Next, we observe that the random variable $B_{3,n} - B_{4,n}$ converges in probability to $0$. We do this by proving that the second moment of $B_{3,n} - B_{4,n}$ converges to $0$ with $n$ under~\Cref{assump:null-main}. 
            \begin{align}
                \mathbb{E}[(B_{3,n} - B_{4,n})^2] &= \frac{n^2}{\mathbb{E}[\tildeh_{tt}^2]} \mathbb{E}\lb \lp   \mathbb{E}[\langle \tildea_{ii}, f_{22} \rangle^2 | \datatwo ] - 2 \mathbb{E}[\langle \tildea_{ii}, f_{21} \rangle  \langle \tildea_{ii}, f_{22} \rangle | \datatwo ]  \rp^2 \rb \label{eq:null-proof-1} \\
                & \leq \frac{n^2}{\mathbb{E}[\tildeh_{tt}^2]} \mathbb{E} \lb \lp   \langle \tildea_{ii}, f_{22} \rangle^2  - 2 \langle \tildea_{ii}, f_{21} \rangle  \langle \tildea_{ii}, f_{22} \rangle   \rp^2\rb \label{eq:null-proof-2}\\
                & \leq \frac{2n^2}{\mathbb{E}[\tildeh_{tt}^2]} \mathbb{E} \lb    \langle \tildea_{ii}, f_{22} \rangle^4  + 4 \langle \tildea_{ii}, f_{21} \rangle^2  \langle \tildea_{ii}, f_{22} \rangle^2 \rb \label{eq:null-proof-3}\\
                & \leq \frac{2n^2}{\mathbb{E}[\tildeh_{tt}^2]} \lp  \mathbb{E} \lb    \langle \tildea_{ii}, f_{22} \rangle^4 \rb  + 4 \mathbb{E}\lb \langle \tildea_{ii}, f_{21} \rangle^4 \rb^{1/2}  \mathbb{E}\lb \langle \tildea_{ii}, f_{22} \rangle^4 \rb^{1/2} \rp. \label{eq:null-proof-4} \\
                & = \mc{O}\lp \frac{\mathbb{E}[ n^{-1} \tildeh_{it}^4 + \tildeh_{it}^2 + \tildeh_{iu}^2 ]}{n \mathbb{E}[\tildeh_{tt}^2]} \rp \to 0. \label{eq:null-proof-5}
            \end{align}
            In the above display, \\
            \eqref{eq:null-proof-2} uses the (conditional) version of Jensen's inequality along with the convexity of $x \mapsto x^2$, \\
            \eqref{eq:null-proof-3} uses the fact that $(x+y)^2 \leq 2 (x^2+y^2)$, \\
            \eqref{eq:null-proof-4} uses the Cauchy-Schwarz inequality, and \\
            \eqref{eq:null-proof-5} follows by expanding the terms $f_{21}$ and $f_{22}$ in terms of $\tildeb_{tu}$, and simplifying the expressions exploiting the fact that the terms containing odd powers of $\langle \tildea_{ii}, \tildeb_{tu} \rangle$ are zero in expectation. The final terms in~\eqref{eq:null-proof-5} is exactly the condition in~\Cref{assump:null-main}. 
            
             Note that $1/B_{4,n}$ can be written as follows: 
            \begin{align}
                \frac{1}{B_{4,n}} = \frac{ \mathbb{E}[\tildeh_{it}^2]}{n \mathbb{E}\lb \langle \tildea_{ii}, \, \frac{1}{n}\sum_{t=1}^n \tildeb_{tt} \rangle^2 \rb}, 
            \end{align}
            which can be shown to be stochastically bounded under the conditions of~\Cref{assump:null-main}, by following the exact same argument used by~\citet{kim2020dimension} in proving that the term (II)  in their~Eq.(53) is stochastically boundeded.

            To complete the proof, we will show that the combination of the two facts proved above; that is, \textbf{(i)} $B_{3,n}-B_{4,n} \convprob 0$ and \textbf{(ii)} $1/B_{4,n} = \mc{O}_P(1)$, are sufficient to conclude that $1/B_{3,n} = \sqrt{B_{2,n}}$ is also stochastically bounded. 
            In particular, these two results imply that for any $\epsilon>0$, we can find a real number $1\leq m < \infty$, and two integers $n_1, n_2 <\infty$, such that the following statements hold: 
            \begin{align}
                \mathbb{P}(1/B_{4,n}> 2m) \leq \epsilon/2, \quad \text{and} \quad \mathbb{P}(|B_{3,n}-B_{4,n}|>m) \leq \epsilon/2. 
            \end{align}
            Hence, we have the following for any $n \geq n_\epsilon \defined \max\{n_1, n_2\}$: 
            \begin{align}
                \mathbb{P}\lp \frac{1}{B_{3,n}} > m \rp &\leq \mathbb{P}\lp \frac{1}{B_{4,n} - |B_{4,n}-B_{3,n}|} > m \rp \\
                & \leq \mathbb{P}\lp \frac{1}{B_{4,n}}> 2m \rp + \mathbb{P} \lp |B_{3,n} - B_{4,n}| > m \rp \leq \epsilon. 
            \end{align}
            Thus, we have shown that $1/B_{3,n}$ is stochastically bounded; that is,  for every $\epsilon>0$, there exists an $m<\infty$, such that  for all $n \geq n_{\epsilon}$, we have $\mathbb{P}\lp 1/B_{3,n} > m \rp \leq \epsilon$.

        \end{proof}

    \section{Power of the~cross-HSIC Test}
    \label{appendix:power} 
        In this section, we prove the results on consistency against fixed and local alternatives~(\Cref{prop:fixed-alternative} and ~\Cref{theorem:local-alternative}) of our cross-HSIC test, stated in~\Cref{sec:power}. The proofs of both of these results can be obtained from a more abstract result, identifying sufficient condtions for the cross-HSIC test to be consistent, which we state and prove in~\Cref{appendix:general-consistency}. Then, in the next two subsections, we use this general result to prove~\Cref{prop:fixed-alternative} and~\Cref{theorem:local-alternative}.

        \paragraph{Additional Notation.} Throughout this section, we will use shorthand notation for some commonly used terms. For any $1\leq i, j, \leq n$, we use $h_{ij}$ to represent $h(Z_i, Z_j) = \frac{1}{2} a_{ii} + a_{jj} - a_{ij} - a_{ij}$, use $\tildeh_{ij}$ for denoting the centered version of $h_{ij}$, i.e., $\tildeh_{ij} = h_{ij} - (\omega - \mu \nu)$. Recall that $\omega$, $\mu$ and $\nu$ denote the kernel mean embeddings of $P_{XY}$, $P_X$ and $P_Y$, and $a_{ij} = k(X_i, \cdot) \ell(Y_j, \cdot)$. Furthermore, recall the definition of  the cross-HSIC statistic, 
        \begin{align}
            \crosshsic = \langle f_1, f_2 \rangle, \quad \text{where } f_1 = \frac{1}{n(n-1)} \sum_{i \neq j} h_{ij}, \quad \text{and } f_2 = \frac{1}{n(n-1)} \sum_{t \neq u} h_{tu}. 
        \end{align}
        We will use $\tildef_1$ and $\tildef_2$ to denote the centered versions of $f_1$ and $f_2$ respectively; that is, $\tildef_1 = f_1 - (\omega - \mu \nu)$ and $\tildef_2 = f_2 - (\omega - \mu \nu)$. Finally, introduce the following term, for any $1 \leq i \leq n$, 
        \begin{align}
            A_i = \frac{1}{n-1} \sum_{j=1}^n h_{ij} = \frac{1}{n-1} \sum_{j=1, j\neq i}^n h_{ij}.  
        \end{align}
        As before, we use $\tildeA_i$ to denote the centered version of $A_i$, and use $\bar{A}_n$ to denote the average of the $A_1, \ldots, A_n$.

    \subsection{General Conditions for Consistency}    
    \label{appendix:general-consistency}

   To identify sufficient conditions for the consistency of the cross-HSIC test, $\Psi = \bone_{\cshsic>z_{1-\alpha}}$,  we first study this problem in a general scenario, in which, the distribution as well as the kernels can change with $n$. In particular, we consider a sequence of distributions $\{P_{XY,n}: n \geq 1\}$ and kernels $\{(k_n, \ell_n): n \geq 1\}$, and let $\data$ denote $2n$ \iid draws from $P_{XY,n}$, for $n \geq 1$.

           \begin{theorem}[General conditions for consistency]
            \label{theorem:general-consistency}
            Consider the independence testing problem with observations $\data = \{(X_i, Y_i): 1 \leq i \leq 2n\}$ drawn \iid from the distribution $P_{XY, n}$, with marginals $P_{X,n}$ and $P_{Y,n}$. Let $\gamma_n^2$ denote $\hsic(P_{XY,n}, \calK, \calL)$, and suppose there exists a non-negative sequence $\{\delta_n: n \geq 1\}$, such that $\lim_{n \to \infty} \delta_n = 0$, satisfying: 
            \begin{align}
                \lim_{n \to \infty} \frac{1}{n^2 \delta_n \gamma_n^4} \lp \mathbb{E}\lb \la \tildeh_{12}, \tildeh_{34} \ra^2  + \gamma_n^2 \lp \la \tildeh_{12}, \tildeh_{12} \ra + n \la \tildeh_{12}, \tildeh_{13} \ra\rp \rb \rp = 0. \label{eq:general-consistency}
            \end{align}
            Then, the test $\Psi = \bone_{\cshsic>z_{1-\alpha}}$ is consistent against the sequence of alternatives $\{P_{XY,n}: n \geq 1\}$. 
        \end{theorem}

        
        \begin{proof} 
            Recall that the test $\Psi = \bone_{\cshsic>z_{1-\alpha}}$ rejects the null if the statistic $\cshsic$ exceeds the $(1-\alpha)$-quantile of the standard normal distribution. Hence, its type-II error can be bounded above, as follows: 
            \begin{align}
                \mathbb{P}\lp \Psi = 0 \rp &= \mathbb{P}\lp \crosshsic <  \frac{z_{1-\alpha} s_n}{\sqrt{n}} \rp \\
                & \leq  \mathbb{P} \lp \crosshsic < z_{1-\alpha}\sqrt{ \mathbb{E}[s_n^2]/n\delta_n} \rp + \delta_n. \label{eq:general-consistency-0}
            \end{align}
            For the inequality in the above display, we introduce the event $E_n = \{ s_n^2 < \mathbb{E}[s_n^2]/\delta_n\}$, and note that, by Markov's inequality, we have $\mathbb{P}(E_n^c) \leq \delta_n$.  
            Recall that $\crosshsic$ and $s_n$ were defined in~\eqref{eq:cross-hsic-def1} and~\eqref{eq:sn2-def} respectively.

            We first note that under the conditions of~\Cref{theorem:general-consistency}, the expectation of $s_n^2$ grows at a rate smaller than $n\delta_n\gamma_n^4$. 
            \begin{lemma}
                \label{lemma:general-consistency-1}
                Under the conditions of~\Cref{theorem:general-consistency}, we have  $\lim_{n \to \infty} \mathbb{E}[s_n^2]/(n\delta_n \gamma_n^4) = 0$. 
            \end{lemma}
            As a consequence of this result, proved in~\Cref{proof:lemma-general-consistency-1},  we note that $\lim_{n \to \infty} z_{1-\alpha} \sqrt{ \mathbb{E}[s_n^2]/n\delta_n \gamma_n^4} \to 0$, which implies that for some finite $n_0$, we have $z_{1-\alpha} \sqrt{ \mathbb{E}[s_n^2]/n\delta_n} \leq \gamma_n^2/2$, for all $n \geq n_0$. Hence, for all $n \geq n_0$, we have the following, 
            \begin{align}
                \mathbb{P}\lp \Psi=0 \rp \leq \mathbb{P}\lp \crosshsic < \gamma_n^2/2 \rp  + \delta_n  = \mathbb{P}\lp \crosshsic - \gamma_n^2< -\gamma_n^2/2 \rp  + \delta_n. 
            \end{align}
            Since $\mathbb{E}[\crosshsic]= \gamma_n^2$, we have the following, by Chebyshev's inequality, 
            \begin{align}
                \mathbb{P}\lp \Psi=0\rp & \leq \frac{ 4\var(\crosshsic)}{\gamma_n^4} + \delta_n. 
            \end{align}
            By assumption, $\lim_{n \to \infty} \delta_n = 0$, and we complete the proof by showing that the first term in the right-hand-side of the above display also goes to zero. 
            
            \begin{lemma}
                    \label{lemma:general-consistency-2} 
                    Under the conditions of~\Cref{theorem:general-consistency}, we have $\lim_{n \to \infty} \frac{ \var(\crosshsic)}{\gamma_n^4} = 0$. 
            \end{lemma}
            The proof of this lemma is in~\Cref{proof:lemma-general-consistency-2}. 
        \end{proof}

        \subsubsection{Proof of~\Cref{lemma:general-consistency-1}}
        \label{proof:lemma-general-consistency-1}
            \begin{proof}
            With the notation introduced at the beginning of this section, we have the following: 
            \begin{align}
                \mathbb{E}[s_n^2] & \lesssim \mathbb{E} \lb \frac{1}{n} \sum_{i=1}^n  \la A_i - \bar{A}_n, f_2 \ra^2  \rb \label{eq:general-consistency-1} \\
                &  = \mathbb{E}\lb  \langle A_i - \bar{A}_n  , f_2 \rangle^2 \rb = \mathbb{E}\lb \frac{1}{n^2} \left \langle \sum_{j=1}^n A_i - A_j, f_2 \right\rangle^2  \rb  \label{eq:general-consistency-2}\\
                & \lesssim \mathbb{E}\lb  \langle A_1 - A_2, f_2 \rangle^2 \rb  = \mathbb{E}\lb  \langle \tildeA_1 - \tildeA_2, f_2 \rangle^2 \rb \label{eq:general-consistency-3} \\
                &\lesssim \mathbb{E}\lb \langle \tildeA_1, f_2 \rangle^2 \rb  = \mathbb{E}\lb \la \tildeA_1, \tildef_2 + (\omega - \mu \nu) \ra^2 \rb. \label{eq:general-consistency-4}
            \end{align}
            In the above display, 
            the first equality in \eqref{eq:general-consistency-2} uses the fact that $(A_i - \bar{A}_n)$ and $(A_j-\bar{A}_n)$ are equal in distribution.
            \eqref{eq:general-consistency-3} uses Cauchy-Schwarz inequality on the `cross-terms' in the expansion of $\langle \sum_{j=1}^n A_i-A_j, f_2 \rangle^2$. 
            The first term in \eqref{eq:general-consistency-4}  follows by using $(x+y)^2 \lesssim x^2 + y^2$, and the fact that $\tildeA_1$ and $\tildeA_2$ are equal in distribution. 
            
                Upper bounding the last term in~\eqref{eq:general-consistency-4}, we obtain 
                \begin{align}
                    \mathbb{E}[s_n^2] & \lesssim \mathbb{E}\lb \la \tildeA_{1}, \tildef_2\ra^2 + \la \tildeA_{1}, \omega - \mu \nu \ra^2 \rb \label{eq:general-consistency-5} \\
                    & \lesssim \mathbb{E}\lb \la \tildeh_{12}, \tildef_2\ra^2 + \la \tildeA_1, \omega - \mu \nu \ra^2 \rb \label{eq:general-consistency-5} \\
                    & = \term_1 + \term_2. \label{eq:general-consistency-6}
                \end{align}
                First, we upper bound $\term_1$ as follows: 
                \begin{align}
                    \term_1 &= \frac{1}{n^2(n-1)^2} \mathbb{E} \lb \la \tildeh_{12}, \sum_{t \neq u} \tildeh_{tu} \ra^2 \rb \\
                    & \lesssim \frac{1}{n^4} \lp n^2 \mathbb{E}\lb \la \tildeh_{12}, \tildeh_{34} \ra^2 \rb + n^3 \mathbb{E}\lb \la \tildeh_{12}, \tildeh_{34} \ra \la \tildeh_{12}, \tildeh_{35} \ra \rb \rp \\
                    & \lesssim \frac{1}{n}\mathbb{E}\lb \la \tildeh_{12}, \tildeh_{34} \ra^2 \rb. \label{eq:general-consistency-7}
                \end{align}
                
                Next, we can get an upper bound on $\term_2$ as follows: 
                \begin{align}
                    \term_2 & = \la \tildeA_1, \omega - \mu\nu \ra^2 
                     \stackrel{(a)}{\leq} \| \tildeA_1 \|^2 \|\omega - \mu \nu \|^2 = \| \tildeA_1 \|^2 \gamma_n^2 \\
                     & \lesssim \frac{\gamma_n^2}{n^2} \lp n \la \tildeh_{12}, \tildeh_{12} \ra + n^2 \la \tildeh_{12}, \tildeh_{13} \ra \rp   \\
                     & \lesssim \gamma_n^2 \lp \frac{\la \tildeh_{12}, \tildeh_{12} \ra}{n} + \la \tildeh_{12}, \tildeh_{13} \ra \rp.   \label{eq:general-consistency-8}
                \end{align}
                    
                Combining~\eqref{eq:general-consistency-7} and~\eqref{eq:general-consistency-8}, we get 
                \begin{align}
                    \mathbb{E}[s_n^2] \lesssim \frac{1}{n} \lp \mathbb{E}\lb \la \tildeh_{12}, \tildeh_{34} \ra^2 \rb + \gamma_n^2 \lp \la \tildeh_{12}, \tildeh_{12} \ra + n \la \tildeh_{12}, \tildeh_{13} \ra \rp \rp, 
                \end{align}
                which implies that 
                \begin{align}
                    \frac{\mathbb{E}[s_n^2]}{n \delta_n \gamma_n^4} \lesssim \frac{1}{n^2 \delta_n \gamma_n^4} \lp \mathbb{E}\lb \la \tildeh_{12}, \tildeh_{34} \ra^2 \rb + \gamma_n^2 \lp \la \tildeh_{12}, \tildeh_{12} \ra + n \la \tildeh_{12}, \tildeh_{13} \ra \rp \rp. \label{eq:general-consistency-9}  
                \end{align}
                By the assumptions of~\Cref{theorem:general-consistency}, the right-hand-side of~\eqref{eq:general-consistency-9} converges to $0$. 
            \end{proof}
        
        \subsubsection{Proof of~\Cref{lemma:general-consistency-2}}
            \label{proof:lemma-general-consistency-2}
            
            Recall that $\crosshsic$ can be rewritten as  
            \begin{align}
                \crosshsic &= \langle \tildef_1 + (\omega - \mu\nu), \tildef_2 + (\omega - \mu\nu)\rangle \\
                &= \la \tildef_1, \tildef_2  \ra  + \la \tildef_1, \omega - \mu \nu \ra +  \la \tildef_2, \omega - \mu \nu \ra + \gamma_n^2. \label{eq:general-consistency-10}
            \end{align}
            Using this, we can write the variance of the $\crosshsic$ statistic as 
            \begin{align}
                \var(\crosshsic) &= \mathbb{E}[(\crosshsic - \gamma_n^2)^2] \\
                &= \mathbb{E}[\langle  \tildef_1, \tildef_2 \rangle^2]+ \mathbb{E}[\langle \tildef_1, \omega - \mu\nu\rangle^2]+ \mathbb{E}[\langle \tildef_2, \omega - \mu\nu\rangle^2] \label{eq:general-consistency-11}\\
                & \lesssim \mathbb{E}[\|\tildef_1\|^2] \mathbb{E}[\|\tildef_2\|^2] + \gamma_n^2 \lp \mathbb{E}[\|\tildef_1\|^2] + \mathbb{E}[\|\tildef_2\|^2] \rp \label{eq:general-consistency-12}\\
                & \lesssim \mathbb{E}[\|\tildef_1\|^2] \lp \mathbb{E}[\|\tildef_1\|^2] + \gamma_n^2 \rp. \label{eq:general-consistency-13}
            \end{align}
            For the equality in~\eqref{eq:general-consistency-11}, we used~\eqref{eq:general-consistency-10}, along with the fact that all the `cross-terms' in expansion of $\mathbb{E}[(\crosshsic-\gamma_n^2)^2]$ have zero expectation.~\eqref{eq:general-consistency-12} follows by applying Cauchy-Schwarz inequality to all terms of~\eqref{eq:general-consistency-11}, while~\eqref{eq:general-consistency-13} uses the fact that $\tildef_1$ and $\tildef_2$ are equal in distribution. 
            
            Since $\tildef_1 = \frac{1}{n(n-1)} \sum_{i \neq j} \tildeh_{ij}$, we can upper bound its norm as follows: 
            \begin{align}
                \mathbb{E}[\|\tildef_1\|^2] &\lesssim \frac{1}{n^4}  \sum_{i \neq j} \sum_{r\neq s}\mathbb{E}[ \langle \tildeh_{ij}, \tildeh_{rs} \rangle]  \\
                &\lesssim \frac{1}{n^4} \lp  n^2 \mathbb{E}[\langle \tildeh_{12}, \tildeh_{12} \rangle ] + n^3 \mathbb{E}[\langle \tildeh_{12}, \tildeh_{13} \rangle] \rp \\
                &\lesssim \frac{1}{n^2} \lp  \mathbb{E}[\langle \tildeh_{12}, \tildeh_{12} \rangle ] + n \mathbb{E}[\langle \tildeh_{12}, \tildeh_{13} \rangle] \rp. \label{eq:general-consistency-14}
            \end{align}
            Combining~\eqref{eq:general-consistency-14} with~\eqref{eq:general-consistency-13}, we get 
            \begin{align}
                \frac{\var(\crosshsic)}{\gamma_n^4} & \lesssim \lp \frac{ \mathbb{E}[\langle \tildeh_{12}, \tildeh_{12} \rangle + n \langle \tildeh_{12}, \tildeh_{13} \rangle]}{n^2 \gamma_n^2} \rp, 
            \end{align}
            which converges to $0$ under the conditions of~\Cref{theorem:general-consistency}. 
            
    \subsection[Consistency against fixed alternatives]{Consistency against fixed alternatives~(\texorpdfstring{\Cref{prop:fixed-alternative}}{Theorem 9})}
    \label{proof:fixed-alternative}
        In the case of fixed alternatives, the term $\gamma_n^2 = \hsic(P_{XY}, k, \ell)$ does not change with $n$. Hence, to prove~\Cref{prop:fixed-alternative}, it suffices to verify~\eqref{eq:general-consistency} with $\gamma_n$ set to some fixed $\gamma>0$.  We proceed in the following steps.
        
        \emph{Verifying $\mathbb{E}[\langle \tildeh_{12}, \tildeh_{34}\rangle^2]<\infty$.} Using the fact that $\tildeh_{12} = h_{12} - (\omega - \mu \nu)$, we have the following: 
        \begin{align}
            \mathbb{E}\lb \langle \tildeh_{12}, \tildeh_{34} \rangle^2 \rb & \leq \mathbb{E}\lb \| \tildeh_{12} \|^2 \|\tildeh_{34} \|^2 \rb  = \mathbb{E}\lb \| \tildeh_{12} \|^2 \rb^2  \\
            & \lesssim \mathbb{E}\lb\lp  \| h_{12}\|^2 + \gamma^2  \rp\rb^2 \\
            & \lesssim \mathbb{E}\lb  \| h_{12}\|^2  \rb^2  + \gamma^4. \label{eq:proof-fixed-alt-1}
        \end{align}
        Hence it suffices to show that under the conditions of~\Cref{prop:fixed-alternative}, $\mathbb{E}[\|h_{12}\|^2]<\infty$, which we do in~\Cref{lemma:fixed-alternative-1} below. 
        
        \emph{Verifying $\mathbb{E}[\langle \tildeh_{12}, \tildeh_{12}\rangle]<\infty$.} Expanding this term, we have 
        \begin{align}
            \mathbb{E}[\langle \tildeh_{12}, \tildeh_{12} \rangle] & = \mathbb{E} \lb \|h_{12} - (\omega - \mu \nu) \|^2 \rb  \lesssim \mathbb{E}[\|h_{12}\|^2] + \gamma^2. \label{eq:proof-fixed-alt-2}
        \end{align}
        Again, this reduces to showing that $\mathbb{E}[\|h_{12}\|^2]<\infty$, which we do in~\Cref{lemma:fixed-alternative-1}.

        \emph{Verifying $\mathbb{E}[\langle \tildeh_{12}, \tildeh_{13}\rangle]<\infty$.} We again reduce this condition to~\Cref{lemma:fixed-alternative-1} as follows: 
        \begin{align}
            \mathbb{E}\lb \langle \tildeh_{12}, \tildeh_{13} \rangle \rb & \leq \mathbb{E}\lb \|\tildeh_{12} \| \|\tildeh_{13}\| \rb = \mathbb{E}\lb \sqrt{\|\tildeh_{12} \|^2 \|\tildeh_{13}\|^2} \rb \\
            & \stackrel{(i)}{\lesssim}  \mathbb{E}\lb \|\tildeh_{12}\|^2 + \|\tildeh_{13}\|^2 \rb  
             \lesssim \gamma^2 +  \mathbb{E}\lb \|h_{12} \|^2 \rb . \label{eq:proof-fixed-alt-3}
        \end{align}
        In the above display,~(i) follows from an application of the AM-GM inequality; $\sqrt{x^2 y^2} \leq (x^2+y^2)/2$.

        \begin{lemma}
            \label{lemma:fixed-alternative-1} For $(X, Y)$ drawn according to $P_{XY}$, if $  \mathbb{E}[k(X, X) \ell(Y, Y)] + \mathbb{E}[k(X, X)]\mathbb{E}[\ell(Y, Y)]< \infty$, then  we have $\mathbb{E}[\|h_{12}\|^2]<\infty$. 
        \end{lemma} 
        \begin{proof}
            Recall that we have $2h_{12} = a_{11} + a_{22} - a_{12} - a_{21}$, where $a_{ij} = k(X_i, \cdot)\ell(Y_j, \cdot)$. Thus, on expanding $\|h_{12}\|^2$, we have 
            \begin{align}
                \|h_{12}\|^2 & \simeq \| a_{11} + a_{22} - a_{12} - a_{21}\|^2 \\
                & \leq \|a_{11} - a_{12} \|^2 + \|a_{22} - a_{21} \|^2 \\
                & \simeq \|a_{11} - a_{12} \|^2 \leq \|a_{11}\|^2 + \|a_{12}\|^2 \\
                & = k(X_1, X_1) \ell(Y_1, Y_1) + k(X_1, X_1) \ell(Y_2, Y_2). 
            \end{align}
            This implies that 
            \begin{align}
                \mathbb{E}[\|h_{12}\|^2] & \leq \mathbb{E}[k(X_1, X_1) \ell(Y_1, Y_1) + k(X_1, X_1) \ell(Y_2, Y_2)]  \\
                & =  \mathbb{E}[k(X, X) \ell(Y, Y)] + \mathbb{E}[k(X, X)]\mathbb{E}[\ell(Y, Y)], 
            \end{align}
            where $(X,Y)$ are drawn according to $P_{XY}$. This completes the proof. 
        \end{proof}
        
    \subsection[Type-I error and consistency against local alternatives]{Type-I error and consistency against local alternatives~(\texorpdfstring{\Cref{theorem:local-alternative}}{Theorem 10})}
    \label{proof:local-alternative}
        To prove this result, we will verify the conditions required by~\Cref{theorem:null-dist-1} and~\Cref{theorem:general-consistency} to prove the type-I error control and consistency respectively. For verifying these conditions, we will need to use some facts about the Gaussian kernel that were derived by~\citet{li2019optimality}. We collect the required properties below. 
        \begin{fact}
        \label{fact:gaussian-kernel-computations-null}
            Suppose $Z_i = (X_i, Y_i) \sim P_X \times P_Y$ for $i=1,2,3,4$ be independent random variables, with $X_i$ and $Y_i$ taking values in $\mathbb{R}^{d_1}$ and $\mathbb{R}^{d_2}$ respectively. Assume that $P_X$ and $P_Y$ admit densities $p_X$ and $p_Y$ respectively, with $p_X \in \Sobolev(M_1)$ and $P_Y \in \Sobolev(M_2)$ with $M_1 \times M_2 = M$. Let $k_n(x, x') = \exp\lp -c_n \|x-x'\|_2^2\rp$ and $\ell_n(y, y') = \exp (-c_n\|y-y'\|_2^2)$ denote Gaussian kernels  on $\mathbb{R}^{d_1}$ and $\mathbb{R}^{d_1}$ respectively, with $c_n = o(n^{4/d})$. Then, we have the following: 
            \begin{align}
                &\mathbb{E}_{P_X}\lb \tildek_n^2(X_1, X_2) \rb \asymp c_n^{-d_1/2}, \quad \text{and} \quad \mathbb{E}_{P_Y}\lb \tildel_n^2(Y_1, Y_2) \rb \asymp c_n^{-d_2/2}, \label{eq:gaussian-kernel-null-1} \\
                &\mathbb{E}_{P_X} \lb \tildek_n^4(X_1, X_2)\rb \asymp c_n^{-d_1/2}, \quad \text{and} \quad \mathbb{E}_{P_Y} \lb \tildel_n^4(Y_1, Y_2)\rb \asymp c_n^{-d_2/2}, \label{eq:gaussian-kernel-null-2}\\
                & \mathbb{E}_{P_X}\lb \tildek_n^2(X_1, X_2)\tildek_n^2(X_1, X_3) \rb \lesssim c_n^{-3d_1/4}, \quad \text{and} \quad \mathbb{E}_{P_Y}\lb \tildel_n^2(Y_1, Y_2)\tildel_n^2(Y_1, Y_3) \rb \lesssim c_n^{-3d_2/4}. \label{eq:gaussian-kernel-null-3}
            \end{align}
            In the expressions above, $\asymp$ and $\lesssim$ hide multiplicative factors depending on $d_1, d_2, d$ and $M$. 
        \end{fact}
        
        The above stated conditions will be used to show the type-I error control by the cross-HSIC test in~\Cref{proof:local-alternative-type-I}. We now state the properties required for the proof of consistency. 
        \begin{fact}
        \label{fact:gaussian-kernel-computations-alt} 
        Suppose $Z_i = (X_i, Y_i) \sim P_{XY}$, for $i=1,2,3$, and assume that $P_{XY}, P_X$ and $P_Y$ have smooth densities $p_{XY}, p_X$ and $p_Y$ respectively.  Then, we have the following: 
        \begin{align}
            &\max \big( \mathbb{E}_{P_{XY}}[k_n(X, X)\ell_n(Y, Y)], \; \mathbb{E}_{P_{XY}}[k_n(X, X)] \mathbb{E}_{P_{XY}}[\ell_n(Y, Y)] \big) \lesssim M^2 c_n^{-d/2},  \label{eq:guassian-kernel-alt-1} \\
            \text{and} \quad & \gamma_n^2 = \mathrm{HSIC}(P_{XY}, k_n, \ell_n) \gtrsim c_n^{-d/2} \|p_{XY} - p_X \times p_Y\|_{L^2}^2.  \label{eq:gaussian-kernel-alt-2}
        \end{align}
        \end{fact}

        \subsubsection{Type-I error bound}
        \label{proof:local-alternative-type-I}
            To show that the cross-HSIC test controls the type-I error at level $\alpha$ asymptotically, we verify that the two conditions of~\Cref{assump:null-main} are satisfied for the kernels considered in~\Cref{theorem:local-alternative}.  In particular, using the fact that $\tildeg(z_1, z_2) = \tildek(x_1, x_2) \times \tildel(y_1, y_2)$ for $z_i = (x_i, y_i)$ and $i=1,2$; we note that it suffices to verify the following conditions under the null: 
            \begin{align}
                & \lim_{n \to \infty}\; \frac{1}{n^2} \frac{ \mathbb{E}[\tildek(X_1, X_2)^4]}{ \mathbb{E}[\tildek(X_1, X_2)^2]^2 } \;  \frac{\mathbb{E}[\tildel(Y_1, Y_2)^4] }{\mathbb{E}[\tildel(Y_1, Y_2)^2]^2  } = 0,  \label{eq:local-typeI-1}\\
                & \lim_{n \to \infty}\; \frac{1}{n} \frac{ \mathbb{E}[\tildek(X_1, X_2)^2 \tildek(X_1, X_3)^2]}{ \mathbb{E}[\tildek(X_1, X_2)^2]^2 } \;  \frac{\mathbb{E}[\tildel(Y_1, Y_2)^2 \tildel(Y_1, Y_3)^2] }{\mathbb{E}[\tildel(Y_1, Y_2)^2]^2  } = 0,  \label{eq:local-typeI-2}\\
                &\lim_{n \to \infty}\; \frac{\lambda_{1,n}^2}{\sum_{i=1}^{\infty} \lambda_{i,n}^2} = 0. \label{eq:local-typeI-3}
            \end{align}
            We proceed in the following steps: 
            
            \textbf{Step 1: verification of~\eqref{eq:local-typeI-1}.} Using the bounds stated in~\eqref{eq:gaussian-kernel-null-1} and~\eqref{eq:gaussian-kernel-null-2}, we obtain 
            \begin{align}
                \frac{ \mathbb{E}[\tildek(X_1, X_2)^4]}{ \mathbb{E}[\tildek(X_1, X_2)^2]^2 } \lesssim \frac{ c_n^{-d_1/2}}{(c_n^{-d_1/2})^2}, \quad \text{and} \quad  \frac{ \mathbb{E}[\tildel(Y_1, Y_2)^4]}{ \mathbb{E}[\tildel(Y_1, Y_2)^2]^2 } \lesssim \frac{ c_n^{-d_2/2}}{(c_n^{-d_2/2})^2}, 
            \end{align}
            which implies that 
            \begin{align}
                 \frac{1}{n^2} \frac{ \mathbb{E}[\tildek(X_1, X_2)^4]}{ \mathbb{E}[\tildek(X_1, X_2)^2]^2 } \;  \frac{\mathbb{E}[\tildel(Y_1, Y_2)^4] }{\mathbb{E}[\tildel(Y_1, Y_2)^2]^2  } \lesssim \frac{1}{n^2} \; \frac{ c_n^{-d/2}}{(c_n^{-d/2})^2}  = \frac{c_n^{d/2}}{n^2} \to 0, 
            \end{align}
            as required in~\eqref{eq:local-typeI-1}.
        
            \textbf{Step 2: verification of~\eqref{eq:local-typeI-2}.} 
            Considering the $\tildek$ dependent term of~\eqref{eq:local-typeI-2}, we note that~\eqref{eq:gaussian-kernel-null-2} and~\eqref{eq:gaussian-kernel-null-3} together imply the following bound: 
            \begin{align}
                \frac{ \mathbb{E}[\tildek(X_1, X_2)^2 \tildek(X_1, X_3)^2]}{ \mathbb{E}[\tildek(X_1, X_2)^2]^2 } \lesssim \frac{c_n^{-3d_1/4}}{(c_n^{-d_1/2})^2} = c_n^{d_1/4}. 
            \end{align}
            Similarly, the $\tildel$ dependent term is upper bounded by $c_n^{d_2/4}$, reducing the condition to 
            \begin{align}
                \frac{1}{n} \frac{ \mathbb{E}[\tildek(X_1, X_2)^2 \tildek(X_1, X_3)^2]}{ \mathbb{E}[\tildek(X_1, X_2)^2]^2 } \;  \frac{\mathbb{E}[\tildel(Y_1, Y_2)^2 \tildel(Y_1, Y_3)^2] }{\mathbb{E}[\tildel(Y_1, Y_2)^2]^2  } \lesssim \frac{ c_n^{d/4}}{n} = \lp \frac{c_n}{n^{4/d}} \rp^{d/4}.  
            \end{align}
            Since $c_n = o(n^{4/d})$, the result follows.

            \textbf{Step 3: verification of~\eqref{eq:local-typeI-3}.} 
            We first note that the condition in~\eqref{eq:local-typeI-3} is equivalent to \begin{align}
                 \lim_{n \to \infty}\; \frac{ \mathbb{E}\lb \mathbb{E}[\tildeg(Z_1, Z_2)\tildeg(Z_1, Z_3) | Z_2, Z_3] \rb }{\mathbb{E}[\tildeg(Z_1, Z_2)^2]^2}  \to 0. 
            \end{align}
            Now, the denominator term satisfies: 
            \begin{align}
                 \mathbb{E}[\tildek(X_1,X_2)^2] \times \mathbb{E}[\tildel(X_1,Y_2)^2] \asymp c_n^{-d_1/2} \times c_n^{-d_2/2} = c_n^{-d/2}, 
            \end{align} 
            which implies that it suffices to show 
            \begin{align}
                c_n^{d}\, \mathbb{E}\lb \mathbb{E}[\tildeg(Z_1, Z_2)\tildeg(Z_1, Z_3) | Z_2, Z_3] \rb \to 0. 
            \end{align}
            For Gaussian kernels with $c_n = o(n^{4/d})$, \citet{li2019optimality} showed that the above condition is true, in the course of proving their Theorem~1. 
             
             Hence, we have verified all the requirements for the limiting null distribution of $\cshsic$ to be $N(0,1)$, which in turn, implies that the cross-HSIC test controls the type-I error at level-$\alpha$ asymptotically. 
             
        \subsubsection{Consistency} 
        \label{proof:local-alternative-consistency} 
            To show the consistency of the cross-HSIC test against local alternatives, we need to show that 
            \begin{align}
                 &\lim_{n \to \infty}  \; D_n \equiv D_n(P_{XY}, k, \ell) = 0, \\
                 \text{where} \quad &    D_n(P_{XY}, k, \ell) \defined 
                 \frac{1}{n^2 \delta_n \gamma_n^4} \lp \mathbb{E}\lb \la \tildeh_{12}, \tildeh_{34} \ra^2  + \gamma_n^2 \lp \la \tildeh_{12}, \tildeh_{12} \ra + n \la \tildeh_{12}, \tildeh_{13} \ra\rb \rp \rp. 
            \end{align}
             Following the bounds derived in the proof of~\Cref{prop:fixed-alternative} in~\Cref{proof:fixed-alternative}, we have        
             \begin{align} 
                \delta_n D_n \lesssim \frac{1}{n^2\gamma_n^4} \bigg( \mathbb{E}[\|h_{12}\|^2]^2 + \gamma_n^4 +  \gamma_n^2(1 + n)(\mathbb{E}[\|h_{12}\|^2] + \gamma_n^2)  \bigg).  \label{eq:proof-local-alt-1}
             \end{align}
             In the above display, we used~\eqref{eq:proof-fixed-alt-1},~\eqref{eq:proof-fixed-alt-2} and~\eqref{eq:proof-fixed-alt-3} to upper bound the three terms involved in the definition of~$D_n$. Now, from~\Cref{lemma:fixed-alternative-1}, we know that $\mathbb{E}[\|h_{12}\|^2] \leq \mathbb{E}[k(X, X)\ell(Y, Y)] + \mathbb{E}[k(X, X)] \mathbb{E}[\ell(Y, Y)]$. Now, for the case of Gaussian kernels, this term is further upper bounded as follows, using~\eqref{eq:guassian-kernel-alt-1}: 
             \begin{align}
             \max \big( \mathbb{E}[k(X, X)\ell(Y, Y)], \; \mathbb{E}[k(X, X)] \mathbb{E}[\ell(Y, Y)] \big) \lesssim M^2 c_n^{-d/2}. \label{eq:proof-local-alt-2}
             \end{align}
            The final component of the proof is the fact that under the conditions of~\Cref{theorem:local-alternative}, we also have the following bound on the true HSIC value using~\eqref{eq:gaussian-kernel-alt-2}: 
            \begin{align}
                \gamma_n^2 \gtrsim c_n^{-d/2} \|p_{XY} - p_X \times p_Y\|_{L^2}^2 > c_n^{-d/2} \Delta_n^2.  \label{eq:proof-local-alt-3}
            \end{align}
            Plugging~\eqref{eq:proof-local-alt-2} and~\eqref{eq:proof-local-alt-3} into~\eqref{eq:proof-local-alt-1}, we get 
            \begin{align}
                D_n \delta_n &\lesssim \frac{ \mathbb{E}[\|h_{12}\|^2]^2}{n^2 \gamma_n^4} + \frac{ \mathbb{E}[\|h_{12}\|^2]}{n \gamma_n^2} +   \frac{1}{n^2} + \frac{1}{n} \\
                & \lesssim \frac{M^4 \nu_{n}^{-d}}{n^2 c_n^{-d} \Delta_n^4}  + \frac{ M^2 c_n^{-d/2}}{n c_n^{-d/2} \Delta_n^2} + \frac{1}{n} \\
                & \lesssim \frac{1}{\lp n^{1/2} \Delta_n \rp^{4}} + \frac{1}{\lp n^{1/2} \Delta_n\rp^2} \lesssim \frac{1}{\lp n^{1/2} \Delta_n \rp^{2}}. \label{eq:proof-local-alt-4}
            \end{align}
             Introduce the term $\beta' = \frac{1}{2} - \frac{2\beta}{d + 4\beta} = \frac{d}{2(d+4\beta)}>0$, and note that~\eqref{eq:proof-local-alt-4} implies  
            \begin{align}
                D_n \lesssim \frac{1}{\delta_n n^{2\beta'} \lp \Delta_n n^{2\beta/(d+4\beta)} \rp^2}. 
            \end{align}
            By selecting $\delta_n = n^{-2\beta''}$ for any $0<\beta''<\beta'$, and using the assumptions  that \textbf{(i)} $ \lim_{n \to \infty} \Delta_n n^{2\beta/(d + 4\beta)} =\infty$, and \textbf{(ii)} $\|p_{XY} - p_X \times p_Y\|_{L^2} > \Delta_n$ for all $P_{XY} \in \altclass$ with density $p_{XY}$; we have  
            \begin{align}
                \lim_{n \to \infty} \sup_{P_{XY} \in \altclass} D_n = 0. 
            \end{align}
            By~\Cref{theorem:general-consistency}, the above condition implies the required consistency against smooth local alternatives of our cross-HSIC test.

  	
    
\end{document}